%% file: alexakis-schlue-time-periodic.tex
\documentclass{amsart}

\usepackage{graphicx}
\usepackage{amssymb}

\numberwithin{equation}{section}

\newcommand{\ud}{\mathrm{d}}
\newcommand{\dm}[1]{\mathrm{d}\mu_{#1}}

\newcommand{\ld}{{}^\ast}
\newcommand{\otimesh}{\hat{\otimes}}

\newcommand{\gs}{g\!\!\!/}
\newcommand{\gammao}{\stackrel{\circ}{\gamma}}

\newcommand{\nablas}{\nabla\!\!\!\!/\:}
\newcommand{\nablasc}{\stackrel{\circ}{\nablas}}
\newcommand{\divs}{\divergence\!\!\!\!\!/\:}
\newcommand{\divsc}{\stackrel{\circ}{\divs}}
\newcommand{\curls}{\curl\!\!\!\!\!/\:}
\newcommand{\curlsc}{\stackrel{\circ}{\curls}}

\newcommand{\Lb}{\underline{L}}

\newcommand{\alphab}{\underline{\alpha}}
\newcommand{\betab}{\underline{\beta}}
\newcommand{\Ab}{\underline{A}}
\newcommand{\Bb}{\underline{B}}
\newcommand{\Nb}{\underline{N}}

\newcommand{\chib}{\underline{\chi}}
\newcommand{\chih}{\hat{\chi}}
\newcommand{\chibh}{\hat{\chib}}

\newcommand{\omegab}{\underline{\omega}}

\newcommand{\limd}{\lim_{u;r\to\infty}}

\newcommand{\VertW}[1]{\lVert #1\rVert_{\mathcal{W}}}

\DeclareMathOperator{\tr}{tr}
\DeclareMathOperator{\divergence}{div}
\DeclareMathOperator{\curl}{curl}

\theoremstyle{plain}
\newtheorem{prop}{Proposition}[section]
\newtheorem{thm}[prop]{Theorem}

\newtheorem{lemma}[prop]{Lemma}

\theoremstyle{definition}
\newtheorem{defn}{Definition}[section]

\theoremstyle{remark}
\newtheorem*{rmk}{Remark}
\newtheorem{nota}[defn]{Notation}

\begin{document}

\title[Non-existence of time-periodic spacetimes]{Non-existence of time-periodic vacuum spacetimes}

\author{Spyros Alexakis}
\address{Department of Mathematics\\ University of Toronto\\ 40 St George Street\\ Rm 6290\\ Toronto, ON M5S 2E4\\ Canada}
\email{alexakis@math.toronto.edu}

\author{Volker Schlue}
\address{Department of Mathematics\\ University of Toronto\\ 40 St George Street\\ Rm 6290\\ Toronto, ON M5S 2E4\\ Canada}
\email{vschlue@math.utoronto.ca}

\date{\today}

\begin{abstract}
We prove that smooth asymptotically flat solutions to the Einstein vacuum equations which are assumed to be periodic in time, are in fact stationary in a neighborhood of infinity. 
Our result applies under physically relevant regularity assumptions purely at the level of the initial data. In particular, our work removes the assumption of analyticity up to null infinity in [Bi{\v{c}}{\'a}k, Scholtz, and Tod; 2010].
The proof relies on extending a suitably constructed ``candidate'' Killing vector field from null infinity, via Carleman-type estimates obtained in [Alexakis, Schlue, Shao; 2013]. 
\end{abstract}

\maketitle

\tableofcontents

\section{Introduction}

\input{intro}

\paragraph{Acknowledgements} We would like to thank Piotr Chru\'sciel for pointing out to us the results in \cite{tod:I} in the early stages of this project. We also thank Greg Galloway for insightful discussions on \cite{galloway:splitting,galloway-chrusciel:poorman}, and Arick Shao for many useful suggestions.

\section{Asymptotically flat dynamical vacuum spacetimes}
\label{sec:asymptotics}

\input{asymptotics}

\section{Time-periodicity and stationarity to all orders at infinity}
\label{sec:vanishing}

\input{vanishing}

\section{Asymptotic extension of time-like Killing vectorfields}
\label{sec:extension}

\input{extension}

%\bibliographystyle{amsalpha}
%\bibliography{time-periodicity}

\input{alexakis-schlue-time-periodic.bbl}

\end{document}

%% file: intro.tex
This paper addresses the question whether there exist asymptotically flat solutions of the Einstein vacuum equations which are ``periodic in time'' in a suitable sense. We show that any such solution must necessarily be stationary near infinity. Thus, genuinely ``time periodic'' solutions do not exist, at least in a neighborhood of (null) infinity.

This question dates back at least to early works of Papapetrou \cite{papapetrou:I,papapetrou:II,papapetrou:III}, see also \cite{papapetrou:translation}. 
His motivation for considering this question appears to be tied to the dynamical problem of motion of gravitating bodies in general relativity.
Indeed, the first derivation of the equations of motion to \emph{first post-Newtonian order} is due to Einstein, Infeld and Hoffman \cite{einstein-infeld-hoffmann}. For a two-body system of masses $m_1$, $m_2$ at positions $\vec{r}_1$, $\vec{r}_2$, the equations of motion, recast as an effective one-body problem for the displacement $\vec{r}=\vec{r}_1-\vec{r}_2$, reads:
\begin{equation}\label{eq:EIH}
  \ddot{\vec{r}}=-\frac{m}{r^2}\frac{\vec{r}}{r}-\frac{1}{c^2}\frac{m}{r^2}\biggl[\Bigl((1+3\eta)\lvert\dot{\vec{r}}\rvert^2-\frac{3}{2}\eta\dot{r}^2-2(2+\eta)\frac{m}{r}\Bigr)\frac{\vec{r}}{r}-2(2-\eta)\dot{r}\dot{\vec{r}}\biggr]
\end{equation}
where $r=\lvert \vec{r} \rvert$, and $m=m_1+m_2$, $\eta=m_1 m_2/(m_1+m_2)^2$ are mass parameters; c.f.~\cite{poisson:book}.
In the subsequent \cite{robertson} Robertson studied the equations \eqref{eq:EIH} derived in \cite{einstein-infeld-hoffmann}, paying precise attention to whether they admit solutions which are periodic in time (in analogy to elliptical orbits in Newtonian theory). He observed that such orbits do exist, namely \emph{circular} processions around the center of mass. As already noted in \cite{einstein-infeld-hoffmann}, this is in apparent contradiction to the wave nature of gravitation (already visible at the level of the linearized Einstein equations) which suggests that the motion of bodies should result in the emission of gravitational waves that carry energy towards infinity, thus causing in turn the two body system to lose energy and hence ruling out the possibility of periodic-in-time solutions. 
This issue was resolved in the setting of the slow-motion approximation for Einstein's equations, with the correct understanding of \emph{higher order} post-Newtonian approximations and their relation to the post-Minkowskian expansion of the metric; see the book of Poisson and Will \cite{poisson:book} for a comprehesive discussion, and also the work of Damour and Blanchet \cite{damour:review,damour:I,damour:II,blanchet:livingreview}.

This still leaves open the question of the existence of time-periodic solutions for the \emph{actual} Einstein equations; we take up this question here and show Theorem~\ref{thm:periodic:informal} below.
We note that in physics an argument is often put forward to rule out such solutions: namely that they should not exist, since time-periodicity, together with the finiteness of the total (ADM) energy should imply that any such solution cannot lose energy towards infinity, and must therefore by stationary. 
This reasoning has in fact been applied to the more general setting of space-times which merely do not emit gravitational radiation, and indeed underpins some of the central views on the generic long-time behaviour of Einstein's equations, see e.g.~Section~9.3 in \cite{hawking-ellis}. 
However, in spite of the wave nature of gravity, this is in fact a very subtle mathematical question: Indeed, while the simplest linear analogue of this prediction, namely that non-radiating free waves in Minkowski space-time are trivial is a classical result of Friedlander \cite{friedlander},  one can easily construct small perturbations of the free wave operator in Minkowski space for which the corresponding assertion is \emph{false}: There exist  wave operators $L:=\Box+V$ (for a smooth, small, compactly supported potentials  $V$)  which admit \emph{non-zero} solutions which decay fast towards null infinities (and thus in particular have vanishing radiation fields); see e.g.~\cite{alexakis-shao}. 
Our main result result\footnote{The precise version of this theorem is stated in Theorem~\ref{thm:stationary:periodic} below.} rests decisively on the precise (non-linear) form of Einstein's equations, to rule out the existence of (asymptotically flat) time-periodic solutions: 

\begin{thm}\label{thm:periodic:informal}
Any asymptotically flat solution $(\mathcal{M},g)$ to the Einstein vacuum equations arising from a regular initial data set which admits a discrete isometry (near null infinity) that maps any point into its chronological future, must be stationary near null infinity. 
\end{thm}

In other words, the discrete isometry of the space-time $(\mathcal{M},g)$ is in fact induced by a \emph{continuous} isometry. Thus the 
solution is not ``genuinely'' time-periodic. 

We remark that if one strengthens the regularity assumption to include null infinity, we can also prove that vacuum space-times that emit no radiation towards $\mathcal{I}^-$, $\mathcal{I}^+$ must also be stationary (near infinity); see Theorem~\ref{thm:stationary:non:radiating} below.

\subsection{Discussion}
We digress here to discuss earlier work on this subject.
The analogous question for \emph{spatially compact} (cosmological) space-times was settled in the affirmative by Galloway \cite{galloway:splitting}; see also \cite{tipler} who puts this question in a much broader context. The proof in \cite{galloway:splitting} is an application of a splitting theorem in Lorentzian geometry which relies on the dominant energy condition (and is thus applicable in the presence of  matter fields).
In the \emph{asymptotically flat} setting, in particular for \emph{asymptotically simple} space-times, all approaches to this problem \cite{papapetrou:I,papapetrou:II,papapetrou:III,gibbons-stewart,tod:I,tod:II} seek to relate the time-periodicity property to an analysis at null infinity. (An exception is the paper \cite{mihalis:periodic} of Dafermos who used the event horizons instead, to show that in \emph{spherical symmetry} time-periodic black hole space-times coupled to suitable matter models are \emph{static}.)

The main modern strategy to handle this problem has been to construct a vector field (near null infinity) which satisfies the Killing equation to any order, by virtue of the time-periodicity assumption.   This analysis is performed most cleanly in the recent work of Bi{\v{c}}{\'a}k, Scholtz, and Tod \cite{tod:I}, where also a suitably flexible notion of time-periodicity is introduced, which is the one we adopt here; c.f.~Def.~\ref{def:time-periodic} below. The paper \cite{tod:I} proceeds in two steps: First, assuming that the space-time admits a $\mathrm{C}^\infty$ smooth conformal compactification, they study the (conformally transformed) Einstein equations at an interior portion of future null infinity. The assumption of time-periodicity is employed to successive orders to derive the existence of a \emph{jet} of a vector field $T$ which satisfies the Killing equation to \emph{any} order. At this point the authors impose the additional assumption of analyticity on the metric $g$ \emph{up to and including} $\mathcal{I}^+$. This allows them to extend the formal jet of  $T$ to an \emph{actual} Killing vector field in $\mathcal{M}$. 

As noted in \cite{tod:I}, one would naturally want to remove the assumption of analyticity. In addition, in view of \cite{christodoulou:Ib}, one would also wish to remove the assumption of $\mathrm{C}^\infty$ smoothness at null infinity. In this paper we achieve both of these goals, and are able to infer the stationarity of smooth space-times using the Einstein equations alone.

Our proof follows the two-step strategy of \cite{tod:I}: In Section~\ref{sec:vanishing} we first derive the existence of a ``candidate'' Killing vector field which satisfies the Killing equation to all orders. We take special care to perform this analysis in the actual space-time; c.f.~Section~\ref{sec:asymptotics} for a precise description of the space-times considered, and the relevant limiting procedures. We work under a regularity assumption on the space-time which is imposed at the level of the initial data. At this point the assumption of time-periodicity is essential. In the second step we \emph{discard} the assumption of analyticity at $\mathcal{I}^+$,  using instead our recent unique continuation theorem for linear waves derived jointly with Shao in \cite{alexakis-schlue-shao}. The latter is in fact \emph{not directly} applicable as a uniqueness result to deduce the extension of symmetry, but is instead implemented at the underlying level of Carleman estimates; c.f.~\cite{alexakis-ionescu-klainerman}. The fact that we extend from infinity introduces additional difficulties, which are overcome in Section~\ref{sec:extension}.

\begin{figure}[tb]
  \centering
  \includegraphics[scale=0.7]{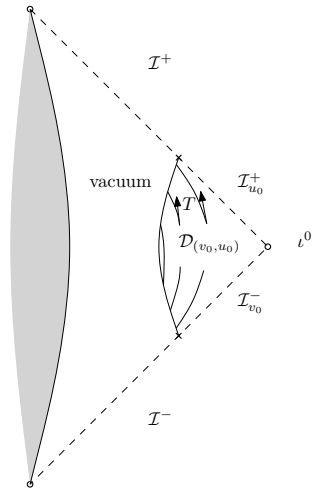}
  \caption{Domain of stationarity near infinity.}
  \label{fig:thm:domain}
\end{figure}

\begin{rmk}
%We remark that there is a natural generalization of our result that one can pursue, where one considers not only time-periodic solutions of the Einstein equations, but space-times for which no graviational energy enters or leaves the space-time, through past and future null infinities, respectively. 
We remark here that a more general question is the stationarity of \emph{non-radiative} space-times, namely solutions to Einstein's equations which are not necessarily time-periodic, yet have the property that no graviational energy enters or leaves the space-time, through suitable portions of past and future null infinities, respectively. 
A version of this problem was in fact also considered by Papapetrou, first for linear (electromagnetic) fields and  partially for the gravitational field \cite{papapetrou:nonradiative}, who argued that if a solution is stationary in the past of a characteristic hypersurface, and non-radiative towards the future, then it should be stationary everywhere. 

Note that this question naturally allows for a \emph{localized} version: One can inquire whether solutions with no incoming and no outgoing radiation on \emph{small} portions $\mathcal{I}_{v_0}^-, \mathcal{I}_{u_0}^+$ of $\mathcal{I}^-$, $\mathcal{I}^+$ (see Fig.~\ref{fig:thm:domain}) must be stationary \emph{near those portions}. (One cannot  expect the solution to be stationary throughout the space-time, by domain of dependence considerations, since imposing incoming radiation ``prior'' to $\mathcal{I}^-_{v_0}$  immediately  yields a non-stationary space-time far from $\mathcal{I}_{v_0}^-\cup \mathcal{I}_{u_0}^+$.)

The methods developed in this paper do in fact apply to this setting also, at least under strong regularity assumptions, see Section~\ref{sec:vanishing:leading}, and Theorem~\ref{thm:stationary:non:radiating} below.
A sufficient condition to derive the existence of a stationary Killing field near $\mathcal{I}^+_{u_0}$, $\mathcal{I}^-_{v_0}$ is essentially to assume a full asymptotic expansion of the metric near null infinity, which is ``well behaved'' in the limit towards space-like infinity; (this is at the same level as assuming the existence of a smooth conformal compactification of null infinity).
Moreover, for a local result of this kind, one expects that such a regularity assumption is \emph{necessary}; see for example the discussion for linear waves in \cite{alexakis-shao}. We stress again that in the setting of time-periodic solutions considered in Theorem~\ref{thm:periodic:informal}, one \emph{does not} need to assume regularity at null infinity, but essentially only on the initial data. We find this an appealing feature of the present result, since the infinite order regularity of null infinities is not expected to hold in general; see \cite{christodoulou:Ib, kroon:log}.
It is for this reason that we are careful to perform the analysis of the Einstein equations  in Section~\ref{sec:vanishing} in the physical space-time (as opposed to on the boundary of a conformally compactified space). 
\end{rmk}

\begin{rmk}
We mention that in the time-periodic setting, one would wish to extend the domain of existence of the stationary Killing vector field to the entire interior. This however is a formidable challenge in full generality, due to the possibility of \emph{trapped null geodesics}. A discussion of this obstacle in the context of the present methods can be found in \cite{alexakis-ionescu-klainerman:cmp}. 
\end{rmk}

\begin{rmk}
  We note that the problem addressed here, namely the extension of an asymptotic symmetry ``at infinity'' to a \emph{genuine} symmetry of the space-time metric also appears naturally in other problems in geometric partial differential equations. Notably the work of Brendle on steady Ricci solitons \cite{brendle} is one such example in Riemannian geometry, where the extension of an asymptotic (rotational) symmetry to the interior was a key step.  
\end{rmk}

\subsection{Definitions and main results}
\label{sec:intro:results}

Let $(\mathcal{M}^{3+1},g)$ be a solution to the Einstein vacuum equations:
\begin{equation}\label{eq:eve}
  \textrm{Ric}(g)=0\,.
\end{equation}
Here we shall be concerned with dynamical solutions arising from asymptotically flat initial data. 
In the seminal work \cite{christodoulou:I}, Christodoulou and Klainerman gave in particular a detailed description of the asymptotic behavior of solutions for a general class of asymptotically flat initial data.
We consider then an asymptotically flat space-time $(\mathcal{M},g)$ with the asymptotic behavior established in their work, and restrict our attention to a domain $\mathcal{D}$ outside the domain of a influence of a ball $B$ in the initial data set $\Sigma$, as sketched in Fig.~\ref{fig:domain:data}.\footnote{Note that the \emph{global} smallness assumption in \cite{christodoulou:I} is not a serious restriction if our attention is focused on a neighborhood of spacelike infinity $\iota^0$; see also the initial remarks in Section~\ref{sec:asymptotics} and \cite{christodoulou:Ib} for a discussion of the relevance of the asymptotics obtained in \cite{christodoulou:I} for the understanding of gravitational waves from potentially \emph{large} sources.} We recall the precise asymptotics of $(\mathcal{M},g)$ towards null infinity in Section~\ref{sec:asymptotics}.
\begin{figure}[tb]
  \centering
  \includegraphics{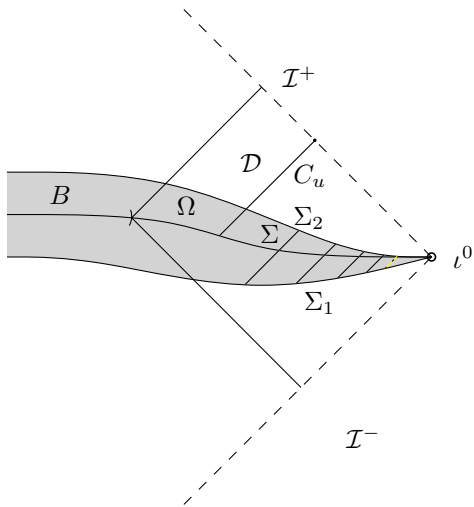}
  \caption{Penrose diagram of an asymptotically flat space-time arising from initial data on $\Sigma$; sketched are the fundamental domain $\Omega$, and the domain $\mathcal{D}$ extending to future and past null infinities $\mathcal{I}^+$, $\mathcal{I}^-$.}
  \label{fig:domain:data}
\end{figure}
In this paper we are primarily interested in solutions which are periodic in time in the following sense; c.f.~\cite{tod:I}:

\begin{defn}[Time-periodic spacetimes]\label{def:time-periodic}
  An asymptotically flat space-time $(\mathcal{M},g)$ is called \emph{time-periodic} if there exists a discrete isometry $\varphi$ with timelike orbits,
  \begin{equation}\label{eq:def:isometry}
    \varphi^\ast g=g\,,\qquad \varphi(p)\in I^+(p) \text{ for all } p\in\mathcal{M}\,.\\
  \end{equation}
  Moreover, we require that $\varphi$ extends smoothly to future and past null infinities as translations $\varphi^+$, $\varphi^-$  along the geodesic generators of future, and past null infinity, respectively, in affine parameter,
 \begin{gather}\label{eq:def:translation}
   \varphi^+(u,\xi)=(u+b,\xi)\qquad (u,\xi)\in\mathbb{R}\times\mathbb{S}^2\simeq\mathcal{I}^+\,,\\
      \varphi^-(v,\xi)=(v+b,\xi)\qquad (v,\xi)\in\mathbb{R}\times\mathbb{S}^2\simeq\mathcal{I}^-\,,
 \end{gather}
 for some fixed $b>0$. 

Furthermore we say that $(\mathcal{M},g)$ is \emph{time-periodic near infinity} if for some fixed advanced time $v=v_0$ and retarded time $u=u_0$, $0<v_0<\infty$, $0<u_0<\infty$, there exists $\varphi$ as in \eqref{eq:def:isometry} on a neighborhood of $\mathcal{I}^+_{u_0}\cup\mathcal{I}^-_{v_0}$,
 \begin{equation}
   \mathcal{I}^{+}_{u_0}=\{(u,\xi)\in\mathcal{I}^+\ :u\leq u_0\}\,,\qquad
   \mathcal{I}^{-}_{v_0}=\{(v,\xi)\in\mathcal{I}^-\ :v\geq v_0\}\,,
 \end{equation}
and $\varphi$ extends smoothly to translations only on $\mathcal{I}^{+}_{u_0}$, $\mathcal{I}^{-}_{v_0}$.
\end{defn}

Given the existence of a discrete isometry $\varphi$ in time-periodic space-times, it will be convenient to pass to the quotient of the manifold $\mathcal{M}$ by the action of the map $\varphi$, which yields the \emph{fundamental domain} $\Omega$ as sketched in Fig.~\ref{fig:domain:data}.

\begin{defn}[Fundamental domain]
  Let $(\mathcal{M},g)$ be an asymptotically flat space-time arising from initial data on $\Sigma$, which is time-periodic.
  We say $\Omega\subset\mathcal{M}$ is a \emph{fundamental domain} for $(\mathcal{M},g)$  if $\Omega$ is connected, $\Sigma\subset\Omega$, and $\partial\Omega$ consists of two smooth space-like hypersurfaces $\Sigma_1,\Sigma_2$ such that $\varphi(\Sigma_1)=\Sigma_2$; moreover, $\Omega$ has the property that for each point $p\in\Omega, \varphi(p)\notin\Omega$, and for each  $q\notin\Omega$ there exists a unique $n\in \mathbb{Z}$ so that $\varphi^{(n)}(q)\in\Omega$. 
\end{defn}

Since $\varphi$ is an isometry, a time-periodic space-time $(\mathcal{M},g)$ can be reconstructed from $g$ in fundamental domain $\Omega$.

For convenience, and to ensure that our methods are applicable more generally to non-radiating space-times, we will be working with the null structure equations near \emph{null infinity}. In particular, we will foliate null infinity by a family of \emph{cuts} $S_u^\ast$ which are the ``boundaries at infininty''  of a 1-parameter family of smooth outgoing null hypersurfaces $C_u$ in the space-time; see Section~\ref{sec:asymptotics}. We remark that for time-periodic space-times, the whole analysis then corresponds to the study of a family of ``broken'' null surfaces in the fundamental domain $\Omega$, c.f.~Fig~\ref{fig:domain:data}. 

This also illustrates that in the time-periodic setting any regularity assumptions made at null infinity really correspond to a regularity assumption towards infinity in the domain $\Omega$. %; this in turn corresponds to a regularity assumption only on the inital data  on $\Sigma$. 
Essentially, the regularity condition that we require follows by assuming a \emph{full asymptotic expansion} for all components of the metric and its first derivatives, in the fundamental domain $\Omega$. 

For definiteness, let us choose coordinates $(t,r,\vartheta^1,\vartheta^2)$ such that \[\Sigma\cap\mathcal{D}=\{(t,r,\vartheta^1,\vartheta^2):t=0,r\geq R\}\] and the initial data on $\Sigma$ have mass $m>0$, and angular momentum $L=ma$, $0<\lvert a\rvert< m$. 
(We also assume that in these coordinates the \emph{linear} momentum of $\Sigma$ vanishes.)
In particular, this fixes the metric $g$ on $\Sigma$ to \emph{highest order} (in $r$).  Beyond the leading order, we assume a full asymptotic expansion of the metric (in inverse powers of $r$); we refer the reader to Definition~\ref{def:expansion} in Section~\ref{sec:smoothness} for the definition of appropriate function spaces $\mathcal{O}^\infty_2(r^{-k})$.  This is an essential regularity assumption needed in this paper to \emph{derive} the existence of a vector field that satisfies the Killing equation to infinite order. In the time-periodic setting, we show that this assumption is \emph{sufficient}, c.f.~Section~\ref{sec:vanishing}. 

In summary, we require that on the \emph{initial} hypersurface $\Sigma\cap\mathcal{D}=\{t=0,r\geq R\}$, the first and second fundamental forms, $g\rvert_\Sigma$, and $k$, respectivly are given by
\begin{subequations}  \label{eq:initial:data}
\begin{align}
  g\rvert_\Sigma&=g_{(m,a)}^{\text{Kerr}}\rvert_{t=0}+g^\infty\rvert_{t=0}\,,\\
  k\rvert_\Sigma&=\partial_tg_{(m,a)}^{\text{Kerr}}\rvert_{t=0}+\partial_tg^\infty\rvert_{t=0}\,,
\end{align}
\end{subequations}
where $g^{\text{Kerr}}_{(m,a)}$ denotes the Kerr metric in Boyer-Lindquist coordinates $(t,r,\vartheta^1,\vartheta^2)$, and the tensor components $g^{\infty}_{\alpha\beta}$ are functions in $\mathcal{O}^\infty_2(r^{-k})$, with $k$ chosen depending on $\alpha,\beta\in\{t,r,\vartheta^1,\vartheta^2\}$ so that $g^\infty$ contains the full asymptotic expansion in $1/r$ beyond the leading orders of the Kerr metric; see Section~\ref{sec:data} for the precise numerology of the exponents.

\begin{rmk}
The existence of such an expansion is often assumed freely at the level of the initial data (both for the metric and the second fundamental form). 
Such an assumption corresponds to smoothness at (spatial) infinity of a suitably compactified metric. We refer to \cite{friedrich:smooth} for a thorough discussion of this issue.
\end{rmk}

\begin{rmk}
   The assumption that $g$ matches a member of the Kerr family to a given order (in particular capturing the first non-trivial order involving the angular momentum parameter $a$, c.f.~Section~\ref{sec:data}) in principle merely corresponds to picking out the Kerr metric with the same mass and angular momentum as our chosen space-time. The agreement of the two metrics to suitably high order should then be derivable, by  normalizing the coordinates on the initial data set near spatial infinity. A way to achieve this could be via conformal normal coordinates centered at spatial infinity, as in \cite{friedrich:smooth}.
\end{rmk}

In fact, we shall require that in the \emph{entire} domain $\Omega\cap\mathcal{D}$ the metric $g$ admits an expansion of the form (c.f.~Def.~\ref{def:expansion} for the definition of the function spaces $\mathcal{O}_2^\infty(r^{-k})$)
\begin{subequations}\label{eq:metric:expansion}
\begin{equation}
  g=-\ud t^2+\ud r^2+r^2\gammao_{AB}\ud \vartheta^A\vartheta^B+g^\infty
\end{equation}
\begin{multline}
 g^\infty=\mathcal{O}_2^\infty(r^{-1})\ud t^2+\mathcal{O}_2^\infty(r^{-1})\ud r^2+\mathcal{O}_2^\infty(r^{-1})\ud t\ud r\\
  +\sum_{A=1}^2\mathcal{O}_2^\infty(r^{-1})\ud t\ud \vartheta^A+\sum_{A=1}^2\mathcal{O}_2^\infty(r^{-1})\ud r\ud \vartheta^A+r^2\sum_{A,B=1}^2\mathcal{O}_2^\infty(r^{-1})\ud \vartheta^A\vartheta^B\,.
\end{multline}
\end{subequations}

\begin{defn}[Regularity at spatial infinity]\label{def:regular}
  Let $(\mathcal{M},g)$ be a time-periodic asymptotically flat spacetime and $\Omega\subset\mathcal{M}$ a fundamental domain. We say $(\mathcal{M},g)$ is \emph{regular} at spatial infinity if $g$ admits an asymptotic expansion of the form \eqref{eq:metric:expansion} in $\Omega$.
\end{defn}

\begin{rmk}
It would be favorable to replace the regularity assumption on $g$ in $\Omega$ by a corresponding assumption purely on the initial data on $\Sigma$.
We expect it would be possible to assume as initial data on $t=0$ the first and second fundamental form induced by \eqref{eq:metric:expansion} on $\Sigma$, and then derive that \eqref{eq:metric:expansion} holds for all $t\in[0,T]$, for $T>0$ arbitrarily large, but finite; in particular, large enough to be valid in the entire domain $\Omega\cap\mathcal{D}$. This requires a suitable ``preservation of regularity'' estimate (for the Einstein equations in harmonic gauge, for example) which, however, appears to be missing in the literature. 
%Thus, the assumption \eqref{eq:metric:expansion} should really be regarded as a regularity assumption on the level of the \emph{initial data} for the putative time-periodic space-time.  
\end{rmk}

The main result in this paper is:

\begin{thm}[Stationarity of time-periodic spacetimes near infinity]\label{thm:stationary:periodic}
Let $(\mathcal{M},g)$ be an asymptotically flat solution to the Einstein vacuum equations \eqref{eq:eve}.
Suppose $(\mathcal{M},g)$ is \emph{time-periodic near } $\mathcal{I}^+_{u_0}\cup\mathcal{I}^-_{v_0}$  as defined in Def.~\ref{def:time-periodic}, for some $u_0,v_0\in\mathbb{R}$, and $g$ is \emph{regular} in the fundamental domain $\Omega$ in the sense of Def.~\ref{def:regular}.
Moreover we assume that on $\Sigma\cap\mathcal{D}$ the initial data coincides to leading order with a Kerr solution, c.f.~\eqref{eq:initial:data:specific}.
Then $(\mathcal{M},g)$ is \emph{stationary near infinity}, i.e. there exists a vectorfield $T$ which is  timelike near infinity such that
\begin{equation}
  \mathcal{L}_T g=0\qquad :\text{ on } \mathcal{D}_{(v_0,u_0)}\,,
\end{equation}
where $\mathcal{D}_{(v_0,u_0)}\subset\mathcal{M}\cap\{u\leq u_0,v\geq v_0\}$ is neighborhood of $\mathcal{I}^+_{u_0}\cup\mathcal{I}^-_{v_0}$; c.f.~Figure~\ref{fig:thm:domain}.
\end{thm}

The proof, as we have already outlined above, proceeds in three steps: We shall first construct a ``candidate'' stationary vectorfield for a large class of asymptotically flat space-times in Section~\ref{sec:asymptotics}. Then we shall prove in Section~\ref{sec:vanishing} that by virtue of the time-periodicity assumption this vectorfield satisfies the Killing equation to all orders, namely its deformation tensor decays faster than any polynomial rate in the distance. Finally we shall apply in Section~\ref{sec:extension} the results in \cite{alexakis-schlue-shao,ionescu-klainerman:extension} to infer that the deformation tensor of the ``candidate'' vectorfield in fact vanishes identically in a neighborhood of infinity, using that by the Einstein vacuum equations the relevant Lie derivative of the curvature satisfies a wave equation.

It is thus only in the second step, in Section~\ref{sec:vanishing}, that the time-periodicity assumption is crucially used; (besides the reduction of the regularity assumptions to the initial data). Moreover, the time-periodicity assumption is used in a rather benign way, essentially only to rule out linear growth, c.f.~Section~\ref{sec:vanishing:leading}. It turns out that the proof of Prop.~\ref{prop:vanishing:induction} (which shows that our ``candidate'' vectorfield satisfies the Killing equation to all orders) applies more generally, in the setting of \emph{non-radiating} solutions, if one is willing to assume instead a strong regularity condition on null infinity, as discussed in the following theorem.

\begin{defn}\label{def:non:radiative}
  An asymptotically flat, dynamical solution to the Einstein vacuum equations is called \emph{non-radiating} (towards future null infinity), if the Bondi mass $M(u)$, as defined in Section~\ref{sec:asymptotics:limits}, c.f.~in particular (\ref{eq:hawking:mass:def}-\ref{eq:bondi:mass:loss}), is \emph{constant} as a function of ``retarded time'' $u$. (Similarly for past null infinity.)
\end{defn}

In view of the monotonicity of the Bondi mass, c.f.~\eqref{eq:bondi:mass:loss}, the non-radiating condition is strictly weaker than the time-periodicity assumption; see also Section~\ref{sec:vanishing:leading}.

\begin{thm}[Stationarity of non-radiative spacetimes under strong smoothness assumptions on null infinity]\label{thm:stationary:non:radiating}
  Let $(\mathcal{M},g)$ be a dynamical solution to the vacuum equations satisfying the asymptotics towards null infinity established in \cite{christodoulou:I} (c.f.~Section~\ref{sec:asymptotics:limits}).
  Let us assume that in addition to \eqref{eq:initial:data}, more specifically \eqref{eq:initial:data:specific},  the initial data on $\Sigma\subset\mathcal{M}$ is such that, with the notation of Section~\ref{sec:null:structure},
  \begin{equation*}
    \sup_{\Sigma} r^5\lvert \alpha \rvert_{\gs}<\infty\,,\qquad \sup_{\Sigma} r^4\lvert\beta\rvert_{\gs}<\infty\,,
  \end{equation*}
holds. Furthermore, let us assume that the space-time $(\mathcal{M},g)$ is smooth at null infinity in the sense of Def.~\ref{def:smooth:null:infinity},
and that all components of the curvature, c.f.~Section~\ref{sec:null:structure}, \[\kappa=\alpha,\beta,\rho,\sigma,\betab,\alphab\] admit and expansion (c.f.~Def.~\ref{def:expansion}) near future and past null infinities $\mathcal{I}^+$, $\mathcal{I}^-$,
\begin{equation}\label{eq:thm:non:radiating:expansion}
  \kappa\sim \sum_{l=0}^\infty \kappa^l_\pm(u_\pm,\xi)\, r^{k-l}\,,\qquad \mathcal{I}^\pm\simeq \{(u_\pm,\xi):u_\pm\in\mathbb{R},\xi\in\mathbb{S}^2\}
\end{equation}
with the property that 
\begin{equation}\label{eq:thm:non:radiating:assumption}
  \lim _{u_\pm\to\mp\infty}\lvert \kappa^l_\pm(u_\pm,\xi)\rvert <\infty\,,\qquad\forall l\in\mathbb{N}.
\end{equation}

If $(\mathcal{M},g)$ is \emph{non-radiating}, in the sense of Def.~\ref{def:non:radiative}, towards future and past null infinity,
then $(\mathcal{M},g)$ is \emph{stationary near infinity}, as concluded in Theorem~\ref{thm:stationary:periodic}.
\end{thm}

\begin{rmk}
We would like to remark that the \emph{localization} achieved in both Theorem~\ref{thm:stationary:periodic}, and \ref{thm:stationary:non:radiating}, namely that it is possible to assume the time-periodic, or more generally non-radiative property only on \emph{arbitrarily small} segments of future and past null infinity, relies crucially on the positivity of the mass of the spacetime. Indeed, the unique continuation results in our \cite{alexakis-schlue-shao} which are used for the extension of the stationary vectorfield from infinity, allow for such localisation to a neighborhood of spatial infinity only in the presence of a positive mass, c.f.~Section~\ref{sec:carleman}. The underlying relationship between the behavior of null geodesics in the space-time, and the positivity property of mass, which has been exploited quantitatively in \cite{alexakis-schlue-shao}, has previously been observed by Penrose, see e.g.~\cite{penrose-woolgar}, who sought to give an alternative proof of the space-time positive mass theorem \cite{schoen-yau} based purely on the behavior of null geodesics near infinity; see also Galloway and Chru\'sciel~\cite{galloway-chrusciel:poorman,chrusciel:nullgeodesics}.
\end{rmk}

%\begin{conj} \label{conj:non:radiating:stationary}   Any non-radiating black hole exterior solution to the vacuum equations with vanishing energy flux through the horizon is stationary. \end{conj}

%% file: asymptotics.tex
The asymptotics of a gravitational field have been described in \cite{christodoulou:I, christodoulou:II} in the fully non-linear dynamical regime; (see also \cite{klainerman-nicolo:I}).
In the following we shall assume in particular the asymptotic behaviour towards future and past null infinity of solutions to the Einstein vacuum equations arising from asymptotically flat initial data; while the results in \cite{christodoulou:I} were derived under a global smallness assumption on the initial data, the rigorous analysis of the asymptotic behaviour towards null infinty is also valid for gravitational waves from strong sources \cite{christodoulou:II}.

\subsection{Choice of gauge and construction of candidate Killing vectorfield}
\label{sec:gauge}

In \cite{christodoulou:II,christodoulou:I} all geometric quantities are decomposed with respect to a null frame defined in terms of one optical function $u$, the level sets of which are outgoing null hypersurfaces, and an affine function $s$. As we shall see this choice is particularly well adapted to the construction of a candidate time-like Killing vectorfield.

In \cite{christodoulou:II} future null infinity $\mathcal{I}^+$ is constructed as a limiting ingoing null hypersurface emanating from spacelike infinity $\iota^0$. 
More precisely, given an initial asymptotically flat spacelike  hypersurface $\Sigma_0$, we can define a family of null hypersurfaces $C_d^-$ relative to an exhaustion of $\Sigma_0$ by balls $B_d$, as the ingoing null hypersufaces emanating from $S_d=\partial B_d$. The intersection of the outgoing null hypersurface $C_0^+$ from a fixed sphere, say $S_0$, with the ingoing null hypersurface $C_\ast^-=C_{d^\ast}^-$, emanating from $S_{d^\ast}$ is a distinguished sphere $S_0^\ast$ at null infinity, as $d^\ast\to\infty$; see Figure~\ref{fig:foliation}.

\begin{figure}[tb]
  \centering
  \includegraphics{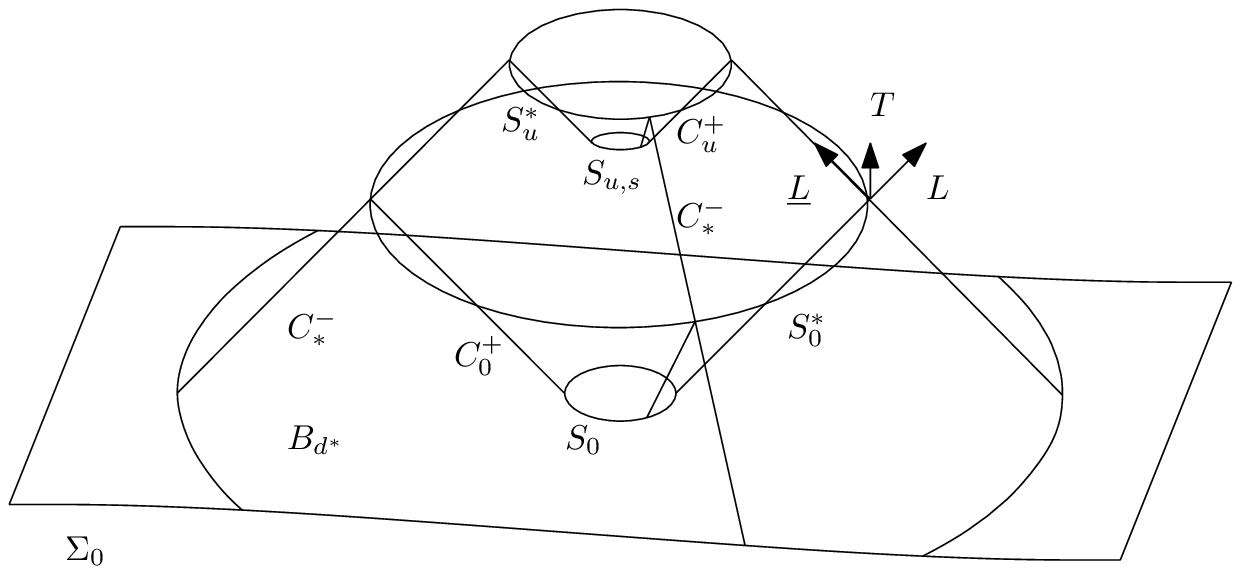}
  \caption{Construction of coordinates near null infinity.}
  \label{fig:foliation}
\end{figure}

Now define null normals $(L,\Lb)$ to the sphere $S_0^\ast$ at infinity by the conditions $g(L,L)=g(\Lb,\Lb)=0$, $g(L,\Lb)=-2$, and $g(L,X)=0$ for $X\in \mathrm{T}S_0^\ast$, which fixes the ingoing and outgoing null vectors $\Lb$, and $L$, respectively up to a rescaling
\begin{equation}
  L\mapsto a\,L\qquad\Lb\mapsto\frac{1}{a}\,\Lb
\end{equation}
where $a$ is a positive function on $S_0^\ast$.
The null second fundamental form $\chi$, and conjugate null second fundamental form $\chib$ defined relative to $L$, and $\Lb$ respectively, are then fixed up to rescalings
\begin{equation}
  \chi\mapsto a\chi\qquad\chib\mapsto\frac{1}{a}\chib\,.
\end{equation}
In view of asymptotic flatness, however,  $\tr\chi>0$, and $\tr\chib<0$, and we can fix a gauge, i.e. a function $a$ on $S_0^\ast$ such that
\begin{equation}
  \tr\chi+\tr\chib=0\qquad :\text{on }S_0^\ast\,.
\end{equation}
Now define
\begin{equation}
  T=\frac{1}{2}\Bigl(L+\Lb\Bigr)\qquad :\text{on }S_0^\ast\,;
\end{equation}
the vectorfield $T$ thus defined coincides with the binormal to $S_0^\ast$, and has the property that the first variation of the area of $S_0^\ast$ along timelike geodesics generated by $T$ vanishes. 

Next the null normals $(L,\Lb)$ are extended to $C_*^-$ as follows: Define $\Lb$ to be the null geodesic vectorfield coinciding with $\Lb$ defined on $S_0^\ast$, and generating null geodesic segments of $C_\ast^-$,
\begin{equation}
  \nabla_{\Lb}\Lb=0\qquad :\text{on }C_\ast^-\,.
\end{equation}
Note that as $d^\ast\to\infty$ we obtain the null geodesic generators $\Lb$ of future null infinity $\mathcal{I}^+$ as limits of null vectorfields along $C^-_\ast$.
With the auxiliary affine distance function $t$ from $S_0^\ast$, $\Lb\cdot t=1$, $t=0$ on $S_0^\ast$ we define a \emph{retarded time} function $u$ on $C^-_\ast$ to have the value
\begin{equation}\label{eq:u:C}
  u(t)=2(r_0^\ast-r_t^\ast)
\end{equation}
on the level sets $S_t^\ast$ of the function $t$, then simply denoted by $S^\ast_u$; 
here $r^\ast_t$ is the area radius of the sphere $S_t^\ast$, c.f.~\eqref{eq:def:arearadius} below.
Note that $u<0$ in the past of $S_0^\ast$, and $u>0$ to the future of $S_0^\ast$.
Then $L$ is defined to be the conjugate null normal to $\Lb$ on each $S_u^\ast$: $g(L,\Lb)=-2$, $g(L,X)=0$ for $X\in\mathrm{T}S_u^\ast$.

Finally we extend the null frame to the past of the null hypersurface $C_\ast^-$ by first defining $u$ to be a solution to the eikonal equation
\begin{equation}
  g^{\alpha\beta}\partial_\alpha u\partial_\beta u=0
\end{equation}
with the above prescribed values on $C^-_\ast$. The level sets $C_u^+$ of the optical function $u$ are then outgoing null hypersurfaces that intersect $C_\ast^-$ in $S_u^\ast$, and we can extend $L$ to be the null geodesic vectorfield on $C_u^+$ coinciding with $L$ previously constructed on $S_u^\ast$:
\begin{equation}\label{eq:propagate:L}
  \nabla_L L =0\,.
\end{equation}
We denote by $s$ the affine parameter along $C_u^+$,
\begin{equation}
  L\cdot s=1\,,\qquad s\rvert_{S_u^\ast}=r_u^\ast\,,
\end{equation}
and denote by $S_{u,s}$ the level sets of $s$ on $C_u^+$. 
Finally $\Lb$ is defined along $C_u^+$ to be the conjugate null normal to $L$ on $S_{u,s}$: $g(L,\Lb)=-2$, $g(L,X)=0$, $X\in\mathrm{T}S_{u,s}$.

We shall now define a coordinate chart in the past of $C_\ast^-$. First choose coordinates $(\vartheta^1,\vartheta^2)$ on $S_0^\ast$, and fix a diffeomorphism $\Phi_0$ of $\mathbb{S}^2$ onto $S_0^\ast$. Then define
\begin{equation}
  \Phi_{u,s}:\mathbb{S}^2\to S_{u,s}
\end{equation}
such that $\Phi_{u,s}\circ\Phi_0^{-1}:S_0^\ast\to S_{u,s}$ is the group of diffeomorphisms generated by the flow along null geodesics, first along the integral curves of $\Lb$ on $C_\ast^-$, then along the integral curves of $L$ on $C_u^+$. 
Now we can assign to any point $p\in S_{u,s}$ the coordinates $(u,s,\vartheta^1,\vartheta^2)$, where $\Phi_{u,s}^{-1}(p)$ has coordinates $(\vartheta^1,\vartheta^2)$ on $S_0^\ast$.

This choice of coordinates is closely related to our construction of the ``candidate'' Killing vectorfield. 
In fact, in these coordinates
\begin{equation}
  L=\frac{\partial}{\partial s}
\end{equation}
and we now \emph{define} everywhere
\begin{equation}
  T=\frac{\partial}{\partial u}\,.
\end{equation}
As we shall see in Lemma~\ref{lemma:T:lead} the coordinate vectorfield $T$ thus constructed is timelike in a neighborhood of null infinity; 
it commutes with the geodesic vectorfield $L$ by construction:
\begin{equation}\label{eq:L:T}
  [L,T]=0
\end{equation}
In Section~\ref{sec:vanishing} we shall prove under the time-periodicity assumption of Def.~\ref{def:time-periodic} that $T=\partial_u$ satisfies the Killing equation to infinite order (in $r$), while \eqref{eq:L:T} is essential in Section~\ref{sec:extension} to deduce further that $T$ in fact generates an isometry.

Note that by construction 
\begin{equation}
  [L,X_a]=0\,,\qquad X_a=\frac{\partial}{\partial \vartheta^a}\,,\quad a=1,2\,,
\end{equation}
it will be useful to work with an orthonormal frame $E_A:A=1,2$ on the spheres $S_{u,s}$, $g(E_A,E_B)=\delta_{AB}$ which is propagated along $C_u^+$ according to
\begin{equation}\label{eq:propagate:EA}
  \nabla_LE_A=-\zeta_A L\,,
\end{equation}
where $\zeta$ is the torsion recalled below. This complements the pair of null vectors to a null frame $(L,\Lb,E_1,E_2)$ for the spacetime.

Let $\gs$ be the induced metric on $S_{u,s}$ and define the area radius function $r$ by
\begin{equation} \label{eq:def:arearadius}
  4\pi r^2(u,s)=\int_{S_{u,s}}\dm{\gs}\,.
\end{equation}

\subsection{Bianchi and null structure equations }
\label{sec:null:structure}

The present null frame $(L,\Lb,E_1,E_2)$ introduced above obeys the frame relations:
\footnote{We refer the reader to Chapter~7.3 of \cite{christodoulou:I} for the basic definitions of null frames, and the associated decompositions of the connection and curvature.}
(Here and in the following $\nablas:=\Pi\nabla$, where $\Pi^{\mu\nu}=g^{\mu\nu}+\frac{1}{2}(\Lb^\mu L^\nu+\Lb^\nu L^\mu)$ is the projection to the tangent space of $S_{u,s}$, denotes the induced connection on the spheres $S_{u,s}$.)
\begin{subequations}\label{eq:frame} 
\begin{gather}
  \nabla_AE_B=\nablas_A E_B+\frac{1}{2}\chi_{AB}\Lb+\frac{1}{2}\chib_{AB}L\\
  \nabla_LE_A=-\zeta_A L\qquad
  \nabla_{\Lb}E_A=\nablas_{\Lb}E_A+\zeta_A\Lb-\lambda_A L\\
  \nabla_{A}\Lb=\chib_A^{\sharp B}E_B+\zeta_A\Lb\qquad
  \nabla_A L=\chi_A^{\sharp B}E_B-\zeta_AL\\
  \nabla_{\Lb}L=2\zeta^{\sharp A}E_A-\omegab L \qquad\nabla_L\Lb=-2\zeta^{\sharp A}E_A \label{eq:propagate:Lb}\\
  \nabla_L L=0\qquad\nabla_{\Lb}\Lb=\omegab\Lb-2\lambda^{\sharp A}E_A
\end{gather}
\end{subequations}
where $\chi$, and $\chib$ denote the null second fundamental forms
\begin{equation}
  \chi_{AB}=g(\nabla_A L,E_B)\qquad \chib_{AB}=g(\nabla_A\Lb, E_B)\,,
\end{equation}
and $\zeta$ is the torsion 1-form,
\begin{equation}
  \zeta\cdot X=\frac{1}{2}g(\nabla_X L,\Lb)\qquad X\in \mathrm{T}S_{u,s}\,;
\end{equation}
here we have also introduced the Ricci rotation coefficients
\begin{equation}
  \omegab=-\frac{1}{2}g(\nabla_{\Lb}\Lb,L)\qquad \lambda_A=-\frac{1}{2}g(\nabla_{\Lb}\Lb,E_A)\,.
\end{equation}
We shall also denote the trace-free parts of $\chi$, and $\chib$ by $\chih$, and $\chibh$, respectively.

The null decomposition of the Einstein equations is presented in full generality in Chapter 7 of \cite{christodoulou:I}.
We quote here in particular that with the following decomposition of the curvature tensor
\begin{subequations}
  \begin{gather}
    \alphab_{AB}=R(E_A,\Lb,E_B,\Lb)\qquad \alpha_{AB}=R(E_A,L,E_B,L)\\
    \betab_A=\frac{1}{2}R(E_A,\Lb,\Lb,L)\qquad \beta_A=\frac{1}{2}R(A,L,\Lb,L)\\
    \rho=\frac{1}{4}R(\Lb,L,\Lb,L)\qquad \sigma=\frac{1}{4}{}^\ast R(\Lb,L,\Lb,L)
  \end{gather}
\end{subequations}
the Bianchi equations on a Ricci flat spacetime manifold take the form, c.f.~Proposition 7.3.2 in \cite{christodoulou:I},
\begin{subequations}
\begin{gather}
  \nablas_{\Lb}\alpha+\frac{1}{2}\tr\chib\alpha+2\omegab\alpha=\nablas\otimesh\beta+5\zeta\otimesh\beta-3\chih\rho-3\ld\chih\sigma\label{eq:bianchi:alpha}\\
  \nablas_L\alphab+\frac{1}{2}\tr\chi\alphab=-\nablas\otimesh\betab+5\zeta\otimesh\betab-3\chibh\rho+3\ld\chibh\sigma\label{eq:bianchi:alphab:L}\displaybreak[0]\\  
  \nablas_L\beta+2\tr\chi\beta=\divs\alpha+\zeta^\sharp\cdot\alpha\label{eq:bianchi:beta:L}\\
  \nablas_{\Lb}\betab+2\tr\chib\betab-\omegab\betab=-\divs\alphab+\zeta^\sharp\cdot\alphab+3\lambda\,\rho-3\ld\lambda\,\sigma\\
  \nablas_{\Lb}\beta+\tr\chib\beta+\omegab\beta=\nablas\rho+\ld\nablas\sigma+3\zeta\rho+3\ld\zeta\sigma+2\chih^\sharp\cdot\betab-\lambda^\sharp\cdot\alpha\label{eq:bianchi:beta:Lb}\\
  \nablas_L\betab+\tr\chi\betab=-\nablas\rho+\ld\nablas\sigma+3\zeta\rho-3\ld\zeta\sigma+2\chibh^\sharp\cdot\beta\label{eq:bianchi:betab:L}\displaybreak[0]\\
  L\rho+\frac{3}{2}\tr\chi\rho=\divs\beta-(\zeta,\beta)-\frac{1}{2}(\chibh,\alpha)\label{eq:bianchi:rho:L}\\
  \Lb\rho+\frac{3}{2}\tr\chib\rho=-\divs\betab-(\zeta,\betab)-2(\lambda,\beta)-\frac{1}{2}(\chih,\alphab)\label{eq:bianchi:rho:Lb}\\
  L\sigma+\frac{3}{2}\tr\chi\sigma=-\curls\beta+\zeta\wedge\beta+\frac{1}{2}\chibh\wedge\alpha\label{eq:bianchi:sigma:L}\\
  \Lb\sigma+\frac{3}{2}\tr\chib\sigma=-\curls\betab-\zeta\wedge\betab+2\lambda\wedge\beta-\frac{1}{2}\chih\wedge\alphab\label{eq:bianchi:sigma:Lb}  
\end{gather}
\end{subequations}
Moreover, we recall the null structure equations from Proposition 7.4.1 in \cite{christodoulou:I}:
\begin{subequations}
  \begin{gather}
    \curls\zeta=-\frac{1}{2}\chih\wedge\chibh+\sigma\\ \curls\lambda=-2\lambda\wedge\zeta \displaybreak[0]\\
    \nablas_{\Lb}\chibh+\tr\chib\chibh-\omegab\chibh=-\nablas\otimesh \lambda+2\zeta\otimesh\lambda-\alphab\\
    \nablas_L\chibh+\frac{1}{2}\tr\chi\chib=-\nablas\otimesh\zeta-\frac{1}{2}\tr\chib\chih+\zeta\otimesh\zeta\label{eq:structure:chibh:L}\displaybreak[0]\\
    \nablas_{\Lb}\tr\chib+\frac{1}{2}\tr\chib\tr\chib-\omegab\tr\chib=-2\divs\lambda+4(\lambda,\zeta)-(\chibh,\chibh)\label{eq:structure:trchib:Lb}\\
    \nablas_L\tr\chib+\frac{1}{2}\tr\chi\tr\chib=-2\divs\zeta-(\chih,\chibh)+2(\zeta,\zeta)+2\rho\displaybreak[0]\label{eq:structure:trchib:L}\\
    \nablas_{\Lb}\zeta=\nablas\omegab-\chib^\sharp\cdot\zeta-\chi^\sharp\cdot\lambda-\betab\label{eq:structure:zeta:Lb}\\
    \nablas_L\zeta=-2\chi^\sharp\cdot\zeta-\beta\label{eq:structure:zeta:L}\displaybreak[0]\\
    \nablas_{\Lb}\chih+\frac{1}{2}\tr\chib\chih+\omegab\chih=\nablas\otimesh\zeta-\frac{1}{2}\tr\chi\chibh+\zeta\otimesh\zeta\\
    \nablas_{L}\chih+\tr\chi\chih=-\alpha\label{eq:structure:chih:L}\displaybreak[0]\\
    \nablas_{\Lb}\tr\chi+\frac{1}{2}\tr\chib\tr\chi+\omegab\tr\chi=2\divs\zeta-(\chih,\chibh)+2(\zeta,\zeta)+2\rho\label{eq:structure:trchi:Lb}\\
    \nablas_L\tr\chi+\frac{1}{2}\tr\chi\tr\chi=-(\chih,\chih)\label{eq:structure:trchi:L}\displaybreak[0]\\
    \nablas_L\lambda-\nablas_{\Lb}\zeta=2\chib^\sharp\cdot\zeta+\betab\label{eq:structure:lambda}\\
    \nablas_L\omegab=-6(\zeta,\zeta)-2\rho\label{eq:structure:omegab}
  \end{gather}
\end{subequations}
Note that the Bianchi and null structure equations in this setting are slightly different from those of a double null foliation used so successfully in \cite{klainerman-nicolo:I, klainerman-nicolo:II, christodoulou:III}. The present frame, which is based on a $u$, $s$ foliation, where the level sets of $u$ are null hypersurfaces but those of $s$ are timelike, is used primarily because it is well-adapted to the construction of the vectorfield $T$, and allows us to refer directly to the presentation \cite{christodoulou:II} of the results on the asymptotics obtained in \cite{christodoulou:I}; these we recall in the next section.

The above equations are stated in terms of covariant derivatives. Alternatively the Bianchi and null structure equations can be written using Lie derivatives, for the purpose of which we define
\begin{equation}\label{def:D}
  D_L\theta=\Pi \mathcal{L}\theta\,,\qquad\theta\text{ : p-form on }S_{u,s}\,
\end{equation}
where $\Pi$ denotes the projection to the spheres $S_{u,s}$.
Its relation to the covariant derivative is summarised in Lemma~\ref{lemma:T:Lie} below.

\subsection{Asymptotic quantities and limiting equations}
\label{sec:asymptotics:limits}

We now recall several conclusions of \cite{christodoulou:I} on the asymptotics of dynamical vacuum solutions towards null infinity which shall be viewed as assumptions in the present setting.

In the geometric setting of Section~\ref{sec:gauge} limits at null infinity are obtained by taking $d^\ast\to\infty$ thus sending $C_\ast^-$ to an asymptotic incoming null hypersurface at infinity. We follow \cite{christodoulou:II} and let every point on $S_0^\ast$, $\phi_0(\vartheta^1,\vartheta^2)$, trace a generator of the outgoing null hypersurface $C_0^+$. We denote for simplicity by
\begin{equation}
  \Phi_u=\Phi_{u,r_u^\ast}:\mathbb{S}^2\to S_u^\ast\,.
\end{equation}
For any $p$-covariant tensorfield $w$ on $(S_u^\ast,\gs)$ of order 
\begin{equation}
  \lvert w\rvert_{\gs}=\mathcal{O}(r^{-q})
\end{equation}
we denote the limit of the pull-back of $w$ to $(\mathbb{S}^2,\gammao)$ by $\Phi_u$ as $d^\ast\to\infty$, in so far it exists, by
\begin{equation}
  \lim_{u;r\to\infty}\Phi_u^\ast r^{-p+q}w_{A_1A_2\ldots A_p}=\lim_{d^\ast\to\infty}r^q w\bigl(r^{-1}\ud \Phi_u\cdot e_{A_1},\ldots, r^{-1}\ud\Phi_u\cdot e_{A_p}\bigr)
\end{equation}
where $e_A:A=1,2$ is an orthonormal frame on $\mathbb{S}^2$, such that $r E_A=\ud\Phi_u\cdot e_A$.

For the spacetimes under consideration one has, according to Chapter~17 of \cite{christodoulou:I} and \cite{christodoulou:II}, 
\begin{equation}
    \limd\Phi^\ast_u r^{-2}\gs=\gammao\qquad
  \limd\Phi_u^\ast K[r^{-2}\gs]=1 \label{eq:asymptotics:metric}
\end{equation}
in other words $S_u^\ast$ become round spheres, and $\gs$ converges to the standard metric on $\mathbb{S}^2$.
Moreover,
\begin{equation}\label{eq:asymptotics:trchi}
    \limd\Phi_u^\ast(r\tr\chi)=2\qquad
  \limd\Phi_u^\ast(r\tr\chib)=-2
\end{equation}
for each fixed $u$,
and to next order
\begin{subequations}
\begin{gather}
  \limd\Phi_u^\ast \Bigl(r(r\tr\chi-2)\Bigr)=H\\
  \limd \Phi_u^\ast\Bigl(r(r\tr\chib+2)\Bigr)=\underline{H}
\end{gather}
\end{subequations}
are functions on $\mathbb{S}^2$, where $H$ is \emph{independent} of $u$, and of vanishing mean, 
and $\underline{H}$ satisfies:
\begin{equation}\label{eq:asymptotics:Hb}
  \partial_u\underline{H}=-\frac{1}{2}\lvert\Xi\rvert^2\,.
\end{equation}
Here the limits of the trace-free parts of the null second fundamental forms are
\begin{subequations}
\begin{gather}\label{eq:asymptotics:chih}
  \lvert\chih\rvert_{\gs}=\mathcal{O}(r^{-2})\qquad \limd \Phi_u^\ast\chih=\Sigma\qquad\tr_{\gammao}\Sigma=0\\
  \lvert\chibh\rvert_{\gs}=\mathcal{O}(r^{-1})\qquad \limd \Phi_u^\ast r^{-1}\chibh=\Xi\qquad\tr_{\gammao}\Xi=0\,,
\end{gather}
\end{subequations}
and are related by, c.f.~\eqref{eq:asymptotics:zeta},
\begin{subequations}\label{eq:asymptotics:Sigma}
\begin{gather}
  \partial_u\Sigma=-\frac{1}{2}\Xi\\
  \divsc\Sigma=\frac{1}{2}\nablasc H+Z\,,
\end{gather}
\end{subequations}
(where $\nablasc$, and $\divsc$ denote respectively the gradient and the divergence operator on the unit sphere $\mathbb{S}^2$).
It is shown in \cite{christodoulou:I} that the Hawking mass enclosed by each sphere $S_{u,s}$,
\begin{equation}\label{eq:hawking:mass:def}
  m(u,s)=\frac{r(u,s)}{2}\Bigl(1+\frac{1}{16\pi}\int_{S_{u,s}}\tr\chi\tr\chib\dm{\gs}\Bigr)
\end{equation}
has a limit for each $u$, called the \emph{Bondi mass}, which plays a central role in the theory of gravitational radiation:
\begin{equation}\label{eq:bondi:mass:def}
  M(u)=\limd m(u,r_u^\ast)\,.
\end{equation}
In particular it is a monotone quantity and satisfies the Bondi mass loss formula
\begin{equation}\label{eq:bondi:mass:loss}
  \partial_u M=-\frac{1}{32\pi}\int_{\mathbb{S}^2}\lvert\Xi\rvert^2\dm{\gammao}\,.
\end{equation}
For this reason $\lvert\Xi\rvert^2(\xi)/32\pi$ is interpreted as \emph{gravitational power radiated to infinity} in a given direction $\xi\in\mathbb{S}^2$ per unit solid angle.
Moreover,
\begin{subequations}\label{eq:asymptotics:Xi}
\begin{gather}
  \partial_u\Xi=-\frac{1}{2}\underline{A}\\
  \divsc\Xi=\Bb
\end{gather}
\end{subequations}
where $\underline{A}$, $\Bb$ are the leading order components of the curvature.
The main theorem in \cite{christodoulou:I} yields in particular the following asymptotics and limits for the curvature components:
\begin{subequations}\label{eq:asymptotics:curvature}
  \begin{gather}
    \lvert\alphab\rvert_{\gs}=\mathcal{O}(r^{-1})\qquad\limd\Phi_u^\ast r^{-1}\alphab=\Ab\\
    \lvert\betab\rvert_{\gs}=\mathcal{O}(r^{-2})\qquad\limd\Phi_u^\ast r\betab=\Bb\\
    \lvert\rho\rvert=\mathcal{O}(r^{-3})\qquad\limd\Phi_u^\ast r^3\rho=P\\
    \lvert\sigma\rvert=\mathcal{O}(r^{-3})\qquad\limd\Phi_u^\ast r^3\sigma=Q
  \end{gather}
\end{subequations}
Moreover, the remaining curvature components are of lower order,
\begin{equation}\label{eq:asymptotics:alphabeta}
  \lvert\alpha\rvert_{\gs}=\mathcal{O}(r^{-\frac{7}{2}})\,, \qquad \lvert\beta\rvert_{\gs}=\mathcal{O}(r^{-\frac{7}{2}})\,.
\end{equation}
\begin{rmk}
  A strengthening of the assumption \eqref{eq:asymptotics:alphabeta} to the existence of the limits
\begin{subequations}\label{eq:limit:alphabeta}
\begin{gather}
  \lvert\beta\rvert_{\gs}=\mathcal{O}(r^{-4})\qquad  \limd\Phi_u^\ast r^3\beta=B \label{eq:limit:beta}\\
  \lvert\alpha\rvert_{\gs}=\mathcal{O}(r^{-5})\qquad   \limd\Phi_u^\ast r^3\alpha=A \label{eq:limit:alpha}
\end{gather}
\end{subequations}
would correspond to the presence of \emph{peeling}; such decay properties are e.g.~implied by the existence of a smooth compactification of the spacetime at null infinity. Moreover \eqref{eq:limit:beta} has in fact been \emph{derived} under strong decay and regularity assumptions on the initial data in \cite{klainerman-nicolo:II}.

While \eqref{eq:limit:alphabeta} as an assumption on future null infinity would not be a major assumption in the context of this paper the main purpose of which is to remove the \emph{analyticity} assumption in previous approaches, it is still undesirable from the dynamical point of view.
Thus in this paper we do \emph{not} impose \eqref{eq:limit:alphabeta} as an assumption on future null infinity.
In the context of Theorem~\ref{thm:stationary:periodic}, the argument for time-periodic solutions only relies on the existence of the limits
\begin{equation}
  \partial_u A=\limd\partial_u\Phi_u^\ast r^3\alpha\,,\qquad\partial_u B=\limd\partial_u\Phi_u^\ast r^3\beta\,.
\end{equation}
which we prove in Section~\ref{sec:vanishing:leading}; c.f.~\cite{christodoulou:Ib}.\footnote{The reason that these limits exist, while \eqref{eq:limit:alphabeta} may not, is that the coefficients to the generically present logarithmic terms are time-independent, that is to say if in fact
\begin{equation*}
  \beta=\frac{B_\ast\log r}{r^4}+\frac{B}{r^4}+o(r^{-4})
\end{equation*}
then $\partial_u B_\ast=0$, see (5) and (6) in \cite{christodoulou:Ib}.}
However, in the context of Theorem~\ref{thm:stationary:non:radiating}, for solutions which are merely non-radiative we have to impose the corresponding assumption on the initial data, such that, c.f.~\eqref{eq:non:radiating:data:limits},
\begin{equation}
  \lim_{u\to-\infty}\lvert A(u,\xi)\rvert <\infty\,,\qquad \lim_{u\to-\infty}\lvert B(u,\xi)\rvert <\infty\,.
\end{equation}

\end{rmk}
Finally the torsion has a limit
\begin{equation}\label{eq:asymptotics:zeta}
  \lvert\zeta\rvert_{\gs}=\mathcal{O}(r^{-2})\qquad\limd \Phi_u^\ast r\zeta=Z\,,
\end{equation}
and satisfies the Hodge system
\begin{subequations}\label{eq:asymptotics:zeta:Hodge}
  \begin{gather}
    \curlsc Z = Q-\frac{1}{2}\Sigma\wedge\Xi\\
    \divsc Z=\underline{N}+P-\frac{1}{2}\Sigma\cdot \Xi
  \end{gather}
\end{subequations}
where $\Nb$ is the limit of a ``mass-aspect function'' satisfying
\begin{subequations}
\begin{gather}
  \partial_u\Nb=-\frac{1}{4}\lvert\Xi\rvert^2\label{eq:asymptotics:Nb}\\
  \frac{1}{4\pi}\int_{\mathcal{S}^2}\Nb\dm{\gammao}=2M\,.\label{eq:asymptotics:Nb:average}
\end{gather}
\end{subequations}

Regarding the remaining Ricci coefficients we assume
\begin{subequations}
  \begin{gather}
   \limd r\omegab=0 \label{eq:asymptotics:omegab}\\
   \limd \Phi_u^\ast\lambda=0\,; \label{eq:asymptotics:lambda}
  \end{gather}
\end{subequations}
note that while the coefficients decay by construction since the vectorfield $\Lb$ is geodesic at infinity, the existence of the limits follows from \eqref{eq:structure:lambda},\eqref{eq:structure:omegab}, and \eqref{eq:structure:zeta:Lb}.

\subsection{Asymptotic expression for  $T$} 

In the following it will be will be important to have an asymptotic expression for the candidate Killing vectorfield $T$ in terms of null frame $(\Lb,L,E_A:A=1,2)$ constructed in Section~\ref{sec:gauge}.

\begin{lemma}\label{lemma:T:lead}
  In terms of the null frame $(E_\mu:\mu=0,\ldots,3)=(\Lb,L,E_1,E_2)$ the coordinate vectorfield $T=\frac{\partial}{\partial u}$ is given by
  \begin{equation}
    T=\frac{1}{2}\bigl(L+\Lb\bigr)\quad+\sum_{\mu=0}^3\mathcal{O}\bigl(\frac{1}{r^2}\bigr)E_\mu\,.
  \end{equation}
\end{lemma}

\begin{proof}

We write
\begin{equation}
  \frac{\partial}{\partial u}=T=T^{\Lb}\Lb+T^L L+T^A E_A
\end{equation}
and can determine the coefficients as solutions to O.D.E.'s using that $T$ is Lie transported by $L$ by \eqref{eq:L:T}.
Indeed, on one hand, using \eqref{eq:propagate:L}, \eqref{eq:propagate:EA}, and \eqref{eq:propagate:Lb},
\begin{equation}
  \nabla_LT=\frac{\ud T^{\Lb}}{\ud s}\Lb+\Bigl(\frac{\ud T^L}{\ud s}-\zeta_A T^A\Bigr)L+\Bigl(\frac{\ud T^A}{\ud s}-2\zeta^A T^{\Lb}\Bigr)E_A\,.
\end{equation}
On the other hand, using \eqref{eq:L:T}, namely $\nabla_LT=\nabla_T L$, we have
\begin{subequations}
\begin{gather}
  g(\nabla_L T,L)=g(\nabla_T L,L)=\frac{1}{2}T g(L,L)=0\,,\\
  g(\nabla_L T,\Lb)=2\omegab T^{\Lb}+2\zeta_A T^A\,,
\end{gather}
\end{subequations}
where have used \eqref{eq:frame}.
By comparison we obtain the following propagation equations for the coefficients:
\begin{subequations}\label{eq:propagate:Talpha}
\begin{gather}
  \frac{\ud T^{\Lb}}{\ud s}=0\,,\\
  \frac{\ud T^L}{\ud s}=-\omegab T^{\Lb}\,,\\
  \frac{\ud T^A}{\ud s}=4\zeta_A T^{\Lb}+\chi_{AB}T^B\,.
\end{gather}
\end{subequations}
We determine the initial values of the coefficients on $C_\ast^-$ from the geometric construction as follows.
By definition $T$ is tangent to the curves $u\mapsto(u,s,\vartheta^1,\vartheta^2)$.
Now we have by construction $s=r$ on $C_\ast^-$, and
\begin{equation}
  \frac{\partial r}{\partial s}=\frac{r}{2}\overline{\tr\chi}\to 1\qquad\text{ as }d^\ast\to\infty\,,
\end{equation}
so that in a neighborhood of $C^-_\ast$ the level sets of $s$ and $r$ coincide as $d^\ast\to\infty$, and $T$ on $C^-_\ast$ is tangential to the level sets of $r$.
But in view of the asymptotics discussed in Section~\ref{sec:asymptotics:limits} the null expansions with respect to $L$, and $\Lb$ are equal and opposite in sign to leading order along  all of future null infinity,  c.f.~\eqref{eq:asymptotics:trchi},
\begin{equation}
  \lim_{u;r\to\infty} r\bigl( \tr\chi+\tr\chib)=0\,,
\end{equation}
and thus also $L+\Lb$ is asymptotically tangential to the level sets of $r$ on $C_\ast^-$ as $d^\ast\to\infty$.
Therefore $T$ and $L+\Lb$  are asymptotically colinear on $C_\ast^-$ as $d^\ast\to\infty$,
while the normalization is fixed by \eqref{eq:u:C}:
\begin{equation}
  \partial_u=\frac{1}{2}(L+\Lb)\qquad\text{ on }C^-_\ast\text{ , as }d^\ast\to\infty\,.
\end{equation}
The initial conditions are thus
\begin{equation}
  T^{\Lb}\rvert_{s=r_u^\ast}=\frac{1}{2}\,,\qquad 
  T^{L}\rvert_{s=r_u^\ast}=\frac{1}{2}\,,\qquad
  T^A\rvert_{s=r_u^\ast}=0\,,\qquad \text{  as } d^\ast\to\infty\,.
\end{equation}
Lastly by the decay properties of $\zeta$, $\omegab$, and $\chi$ recalled below, we have from \eqref{eq:propagate:Talpha} that 
\begin{equation}
    \frac{\ud T^{\Lb}}{\ud s}\rvert_{s=r_u^\ast}=0\,,\quad
  s\frac{\ud T^L}{\ud s}\rvert_{s=r_u^\ast}\to0\,,\quad
  s\frac{\ud T^A}{\ud s}\rvert_{s=r_u^\ast}\to0\,,\qquad \text{  as } d^\ast\to\infty\,,
\end{equation}
which proves the formula of the Lemma.% by Taylor expansion in $s^{-1}$.
\end{proof}

\subsection{Metric expression} 

In the coordinates introduced in Section~\ref{sec:gauge} the spacetime metric takes the form
\begin{equation}\label{eq:gauge:metric}
  g=-\ud u\ud s-l \ud u^2+\gs_{ab}\bigl(\ud\vartheta^{a}-b^a\ud u\bigr)\bigl(\ud\vartheta^b-b^b\ud u\bigr)
\end{equation}
and the null normals $(L,\Lb)$ are given by
\begin{subequations}
  \begin{gather}
    L=\frac{\partial}{\partial s}\qquad\\
    \Lb=2\frac{\partial}{\partial u}-l\frac{\partial}{\partial s}+2 b^a\frac{\partial }{\partial \vartheta^a}\,.
  \end{gather}
\end{subequations}
Note that since by \eqref{eq:propagate:Lb},
\begin{equation}
    [L,\Lb]=\nabla_L\Lb-\nabla_{\Lb}L=-4\zeta^{a}\frac{\partial}{\partial \vartheta^a}+\omegab L
    %=-\frac{\partial l}{\partial s}\frac{\partial }{\partial s}+2\frac{\partial b^A}{\partial s}\frac{\partial}{\partial\vartheta^A}
\end{equation}
we have
\begin{equation} \label{eq:propagate:lb}
  \frac{\partial l}{\partial s}=-\omegab\qquad\frac{\partial  b^a}{\partial s}=-2\zeta^a\,.
\end{equation}
Moreover,
\begin{equation}
  \limd l=1\qquad\limd b^a=0\,.
\end{equation}

\subsection{Regularity properties}
\label{sec:smoothness}

In Section~\ref{sec:intro:results} we have assumed that in the fundamental domain the metric $g$ is of the form
\begin{equation}\label{eq:metric:ginfty}
  g=-\ud t^2+\ud r^2+r^2\gammao_{AB}\ud \vartheta^A\vartheta^B+g^\infty
\end{equation}
where each component $g^\infty_{\alpha\beta}$ of the tensor $g^\infty$ in the coordinates $(t,r,\vartheta^1,\vartheta^2)$ is an element of $\mathcal{O}_2^\infty(r^{-k})$, for a suitably chosen $k\in\mathbb{N}$. We now define the function spaces $\mathcal{O}_m^\infty(r^k)$.

\begin{defn}[Asymptotic expansions]\label{def:expansion}
We say a function $f:[R,\infty)\times\mathbb{S}^2\to\mathbb{R}, (r,\vartheta^1,\vartheta^2)\mapsto f(r,\vartheta^1,\vartheta^2)$ belongs to the class $\mathcal{O}^\infty_m(r^k)$, for a fixed $k\in\mathbb{Z}$, $m\in\mathbb{N}$, if there exist $\mathrm{C}^m$-functions $f^l:\mathbb{S}^2\to\mathbb{R}$, $l\in\mathbb{N}$ such that: 
\begin{subequations}
\begin{gather}
\label{eq:defn:expansion:sim}
f(r,\vartheta^1,\vartheta^2)\sim \sum_{l=0}^\infty f^l(\vartheta^1,\vartheta^2) r^{k-l}\,,\\
\partial_{\vartheta}^\alpha f(r,\vartheta^1,\vartheta^2)\sim \sum_{l=0}^\infty \partial_{\vartheta}^\alpha f^l(\vartheta^1,\vartheta^2) r^{k-l}\,,\\\text{where} \quad\partial^\alpha_\vartheta:=\frac{\partial ^{\alpha_1}}{\partial\vartheta^{A_1}}\ldots\frac{\partial ^{\alpha_n}}{\partial\vartheta^{A_n}}\,,\alpha_1+\ldots+\alpha_n=m\,,A_j=1,2,\notag\\
\partial_r^n f(r,\vartheta^1,\vartheta^2)\sim \sum_{l=0}^\infty \prod_{j=0}^{n-1}(k-l-j)f^l(\vartheta^1,\vartheta^2) r^{k-l-n}\,,\ 1\leq n\leq m\,.
\end{gather}
\end{subequations}
Analogously we define the class $\mathcal{O}^\infty_m(r^k)$ for functions $f:\mathbb{R}\times [R,\infty)\times \mathbb{S}^2\to\mathbb{R}$, $(t,r,\vartheta^1,\vartheta^2)\mapsto f(t,r,\vartheta^1,\vartheta^2)$ (with each $f^l$ now depending on $(t,\vartheta^1,\vartheta^2)$) by  requiring in addition that: 
\begin{equation}
\partial_t^n\partial_\vartheta^\alpha f(t,r,\vartheta^1,\vartheta^2)\sim \sum_{l=0}^\infty \partial_t^n\partial_\vartheta^\alpha f^l(t, \vartheta^1,\vartheta^2) r^{k-l}\,,\text{whenever } n+\vert\alpha\rvert\leq m\,.
\end{equation}
Here  $f\sim\sum_{l=0}^\infty f^l r^{k-l} $ means that for all $N\in\mathbb{N}$ there exists a $C_N>0$ so that 
\begin{equation}
\lvert f- \sum_{l=0}^N f^lr^{k-l}\rvert \leq C_N r^{k-N-1}\,;
\end{equation}
the above is assumed to hold for all values of $r\geq R$, $(\vartheta^1,\vartheta^2)\in\mathbb{S}^2$ (and $t\in\mathbb{R}$).
\end{defn}

The following smoothness properties at null infinity are essential for the argument of Section~\ref{sec:vanishing}.
In the time-periodic setting these smoothness properties are automatically inherited from the regularity in the fundamental domain.

Consider a $p$-covariant tensorfield on $S_{u,s}$ as in the previous section. Let us suppose that $w$ is a geometric component of the null decomposition recalled in Section~\ref{sec:null:structure}, either of the curvature, the connection, or any derivative thereof, and suppose that it is already known that
\begin{equation}\label{eq:smoothness:vanishing}
  \limd \Phi_u^\ast(r^{q} w)=0\,.
\end{equation}

\begin{defn}[Smoothness at null infinity] \label{def:smooth:null:infinity}
  An asymptotically flat solution to the vacuum equations is called \emph{smooth at null infinity} if for any geometric quantity $w$ that vanishes in the limit to a given order $q$ in the sense of \eqref{eq:smoothness:vanishing}, we have 
\begin{equation}\label{eq:smoothness:r}
  \Phi_{u,s}^\ast(r^{q}w)=\mathcal{O}(r^{-1})\,,
\end{equation}
and the limit of $w$ exists to order $q+1$, in the sense that
\begin{equation}\label{eq:smoothness:L}
  \partial_s\Phi_{u,s}^\ast (r^{q+1}w)=\Phi_{u,s}^\ast\mathcal{L}\!\!\!/{}_L(r^{q+1}w)=\mathcal{O}(r^{-2})\,.
\end{equation}
Moreover, as part of the smoothness property, we require a mild angular regularity assumption that states that if the limit of a geometric quantity vanishes in the sense of \eqref{eq:smoothness:vanishing} then so do its angular derivatives:
\begin{equation}
  \label{eq:smoothness:angular}
  \limd \nablasc\Phi_u^\ast(r^q w)=\limd \Phi_u^\ast (r^{q}\nablas w)=0\,.
\end{equation}
\end{defn}

\begin{lemma}
  Let $(\mathcal{M},g)$ an asymptotically flat space-time, which is time-periodic and $\Omega\subset\mathcal{M}$ a fundamental domain.
  If $g$ is regular in $\Omega$, then $(\mathcal{M},g)$ is smooth at future null infinity.
\end{lemma}

\begin{proof}
  Let $s\mapsto \gamma(s)$ be an affinely parametrized null generator of a null hypersurface $C_u$. We show $s$ is comparable to $r$ in the fundamental domain. The regularity statement then follows immediately from the Definition~\ref{def:regular}. W.l.o.g.~we can assume that $\gamma(s)\in S_{u,s}$. By time-periodicity there exists $n\in\mathbb{Z}$ such that $\varphi^{(n)}(\gamma(s))\in\Omega$. This defines a map $\phi:C_u\to\Omega$, the image of $\gamma$ thus being a ``broken null curve'' in $\Omega$, c.f.~Fig.~\ref{fig:domain:data}. Since by \eqref{eq:asymptotics:trchi},
  \begin{equation}
    \limd L\cdot r=\limd \frac{r}{2}\overline{\tr\chi}=1
  \end{equation}
  and by definition $L=\dot{\gamma}$, we obtain that there exists a constant $c$ such that $\phi(\gamma(s))\in\{(t,r,\vartheta^1,\vartheta^2)\in\Omega:c^{-1}s\leq r\leq c s\}$. Then the limits in Def.~\ref{def:smooth:null:infinity} follow from \eqref{eq:metric:ginfty} with the components of $g^\infty$ in $\mathcal{O}_2^\infty(r^{-k})$ as defined in Def.~\ref{def:expansion}. 
\end{proof}

\begin{rmk}
In cases where the geometric quantity $w$ satisfies a propagation equation in the outgoing null direction, the existence of the limit to next order can in fact be inferred if the forcing terms in the o.d.e. are known to be integrable; see Lemma~\ref{lemma:smoothness:ode} which we shall employ frequently below. However, for quantities where no such equation is available from the vacuum equations (such as the curvature component $\alpha$), this a genuine smoothness assumption on the spacetime geometry at infinity. %The angular regularity assumption can be proven quite generally for small data problems for non-linear p.d.e.'s if assumed on the level of the initial data, but we shall here not attempt to weaken the assumptions further. 
\end{rmk}

\begin{lemma}
  \label{lemma:smoothness:ode}
  Let $f:[s_0,\infty)\to\mathbb{R}$ be a bounded function solving the o.d.e.
  \begin{equation}
    \partial_sf(s)-\frac{k}{s}f(s)=g(s)
  \end{equation}
  where $k>0$, and $g=\mathcal{O}(s^{-2})$. Then
  \begin{equation}
    f=\mathcal{O}(s^{-1})\,.
  \end{equation}
\end{lemma}

\begin{proof}
  Multiplying the equation by $s^{-k}$ we derive
  \begin{equation}
    \partial_s(s^{-k}f)=s^{-k}g(s)=\mathcal{O}(s^{-k-2})
  \end{equation}
  which upon integrating yields in view of the boundedness of $f$,
  \begin{equation}
    \lvert s^{-k}f(s)\rvert\leq C s^{-k-1}
  \end{equation}
  for some constant $C>0$, proving the claim after multiplying by $s^k$.
\end{proof}

%% file: vanishing.tex
In this section we will prove that under the assumption of time-periodicity the spacetimes described in Section~\ref{sec:asymptotics} are ``stationary to all orders'', i.e.~the deformation tensor of the time-like vectorfield constructed in Section~\ref{sec:gauge} vanishes to all orders at null infinity.

\begin{prop}
  \label{prop:time:periodic:killing:all:orders}
  Let $(\mathcal{M},g)$ be a solution to the vacuum equations satisfying the asymptotics of Section~\ref{sec:asymptotics:limits} towards future null infinity and the regularity assumptions of Definition~\ref{def:smooth:null:infinity}. If $(\mathcal{M},g)$ is time-periodic, then
  \begin{equation}\label{eq:prop:killing:all:orders}
    \limd r^k\mathcal{L}_Tg=0\qquad\forall k\in\mathbb{N}\cup\{0\}\,,
  \end{equation}
  where the time-like vectorfield $T$ is constructed as in Section~\ref{sec:gauge}.
\end{prop}

In the spirit of \cite{tod:I} this is achieved by an induction argument, which had been employed therein in the case where the spacetime can be conformally compactified.
In Section~\ref{sec:vanishing:leading} we first discuss the general asymptotics of \emph{non-radiating} spacetimes, and prove the stationarity of time-periodic spacetimes ``to leading order''. 
In Section~\ref{sec:vanishing:induction} we then prove the main Proposition~\ref{prop:vanishing:induction} for time-periodic solutions showing that all components of the curvature and connection, and thus of the metric, are time-independent to all orders at future null infinity.

\subsection{Time-periodicity and leading order behaviour}
\label{sec:vanishing:leading}

In view of the Bondi mass loss formula \eqref{eq:bondi:mass:loss} which states in particular that the Bondi mass $M(u)$ is a \emph{monotone} function of time $u$ we have that for any time-periodic spacetime the Bondi mass is \emph{constant},
\begin{equation}\label{eq:vanishing:Xi}
  \partial_u M=0\,,\qquad \Xi=0\,.
\end{equation}
A dynamical spacetime with the property \eqref{eq:vanishing:Xi} \emph{(not necessarily time-periodic)} is called \emph{non-radiating}.
While such spacetimes are of separate interest, we state here a few consequences of \eqref{eq:vanishing:Xi} which will be used below.

\begin{prop}\label{prop:non:radiating}
  Consider a solution to the vacuum equations with the asymptotics towards future null infinity discussed in Section~\ref{sec:asymptotics:limits}.
  If the spacetime is non-radiating in the sense of \eqref{eq:vanishing:Xi} then
  \begin{gather*}
    \Ab=0\,,\qquad\Bb=0\,,\\
    \partial_u P=0\,,\qquad\partial_uQ=0\,,\qquad \partial_uZ=0\,,\qquad \partial_u\Sigma=0,
  \end{gather*}
  and
  \begin{equation*}
    \limd r\partial_u\gammao=0\,,\qquad  \partial_uH=0\,,\qquad  \partial_u\underline{H}=0\,.
  \end{equation*}
  Moreover, if the solution is smooth at future null infinity in sense of Definition~\ref{def:smooth:null:infinity} then in fact
  \begin{equation*}
    \lvert\alphab\rvert_{\gs}=\mathcal{O}(r^{-3})\,,\qquad \lvert\betab\rvert_{\gs}=\mathcal{O}(r^{-3})\,.
  \end{equation*}
\end{prop}

\begin{proof}
It follows immediately from the asymptotic equations \eqref{eq:asymptotics:Xi} that the no-radiation condition \eqref{eq:vanishing:Xi} implies the vanishing of the leading order curvature components:
\begin{subequations}\label{eq:vanishing:AB}
\begin{gather}
  \limd\Phi_u^\ast r^{-1}\alphab=\Ab=-2\partial_u\Xi=0\\
  \limd\Phi_u^\ast r\betab=\Bb=\divsc\Xi=0
\end{gather}
\end{subequations}
Already the next curvature component in the hierarchy \eqref{eq:asymptotics:curvature}, the limit of $\rho$, cannot be expected to vanish identically\footnote{The value of $P=\lim r^3 \rho$ for the \emph{static} spherically symmetric Schwarzschild solution is $P=-2M$.}; however, it follows directly from the limit of the Bianchi equations that $P$, and $Q$ must be time-independent by virtue of \eqref{eq:vanishing:AB}.
Indeed, adding \eqref{eq:bianchi:rho:L}, \eqref{eq:bianchi:rho:Lb}, and multiplying by $r^3$, 
\begin{multline}
    r^3(L+\Lb)\rho+\frac{3}{2}r(\tr\chi+\tr\chib)r^2\rho=\divs r^3\beta-r^2(r\zeta,\beta)-\frac{1}{2}r^4(r^{-1}\chibh,\alpha)\\
    -\divs r^3\betab-r^2(r\zeta,\betab)-2r^2(\lambda,r\beta)-\frac{1}{2}r^4(\chih,r^{-1}\alphab)
\end{multline}
we obtain after taking the limit using the assumptions of Section~\ref{sec:asymptotics:limits}, in particular \eqref{eq:asymptotics:trchi}, \eqref{eq:asymptotics:chih}, and \eqref{eq:asymptotics:curvature},
\begin{equation}\label{eq:DT:P}
   T P=-\frac{1}{2}\divsc\underline{B}-\frac{1}{4}(\Sigma,\underline{A})\,,
\end{equation}
where we have used the leading order identity $2T=\Lb+L$ of Lemma~\ref{lemma:T:lead}, and the result that $T$ commutes with $r$ to leading order, c.f.~Lemma~\ref{lemma:T:commute},
\begin{equation}
  \limd T\cdot r=  \limd \frac{r}{2}\frac{1}{2}\bigl(\overline{\tr\chi}+\overline{\tr\chib}\bigr)=0\,.
\end{equation}
Similarly, we obtain from \eqref{eq:bianchi:sigma:L} and \eqref{eq:bianchi:sigma:Lb} the limiting equation for $\sigma$,
\begin{equation}\label{eq:DT:Q}
  T Q=-\frac{1}{2}\curlsc\Bb-\frac{1}{4}\Sigma\wedge\Ab\,.
\end{equation}
Therefore, by \eqref{eq:vanishing:AB},
\begin{equation}\label{eq:vanishing:PQ}
  \partial_uP=0\,,\qquad\partial_uQ=0\,.
\end{equation}
Moreover, by \eqref{eq:vanishing:Xi} the Hodge system for the limit of the torsion \eqref{eq:asymptotics:zeta:Hodge} simplifies to
\begin{equation}\label{eq:Hodge:Z:Xi}
      \curlsc Z = Q\qquad    \divsc Z=\underline{N}+P
\end{equation}
and thus implies by \eqref{eq:vanishing:PQ}, and \eqref{eq:asymptotics:Nb}, that
\begin{equation}
      \curlsc \partial_uZ = \partial_uQ=0\qquad    \divsc \partial_uZ=-\frac{1}{4}\lvert\Xi\rvert^2+\partial_uP=0\,,
\end{equation}
which yields, this being a div-curl system on the unit sphere\footnote{We appeal here to the basic $\mathrm{L}^2$ theory of Hodge systems on compact $2$-dimensional manifolds, as presented e.g. in Chapter 2 of \cite{christodoulou:I}, c.f.~Lemma~2.2.2 therein.}, the time-independence of the torsion
\begin{equation}\label{eq:vanishing:Z}
  \partial_uZ=0\,.
\end{equation}  

Finally as a direct consequence of $\Xi=0$ the metric itself is time-independent, in fact beyond the leading order, 
\begin{equation}
  \limd r \partial_u\gammao=\limd \Phi_u^\ast rD_T(r^{-2}\gs)=\frac{1}{2}\limd \Phi_u^\ast r^{-1}(\chib+\chi)=\Xi=0\,.
\end{equation}
Also, by \eqref{eq:asymptotics:Sigma},
\begin{equation}
    \partial_u\Sigma=-\frac{1}{2}\Xi=0\,,
\end{equation}
and moreover, 
using \eqref{eq:structure:trchi:Lb},\eqref{eq:structure:trchi:L} and \eqref{eq:structure:trchib:Lb}, \eqref{eq:structure:trchib:L},
or simply \eqref{eq:asymptotics:Hb}
we verify that
\begin{subequations}
\begin{gather}
  \frac{\partial H}{\partial u}=\limd r^2\partial_u\tr\chi=\frac{1}{2}\limd r^2\Lb\tr\chi+\frac{1}{2}\limd r^2 L\tr\chi=0\\
  \frac{\partial \underline{H}}{\partial u}=\limd r^2\partial_u\tr\chib=-\frac{1}{2}\lvert\Xi\rvert^2=0\,.
\end{gather}
\end{subequations}

While \eqref{eq:vanishing:AB} already proves, for smooth solutions, that
\begin{equation}
  \lvert\alphab\rvert_{\gs}=\mathcal{O}(r^{-2})\,,\qquad\lvert\betab\rvert_{\gs}=\mathcal{O}(r^{-3})\,,
\end{equation}
it remains to improve the order of vanishing of $\alphab$ beyond these leading order asymptotics.
We return to the Bianchi equations for $\alphab$, and use Lemma~\ref{lemma:T:Lie} to rewrite \eqref{eq:bianchi:alphab:L} as follows:
\begin{equation}
  D_L\alphab-\frac{1}{2}\tr\chi\alphab=-\nablas\otimesh\betab+\frac{1}{2}(\chih,\alphab)\gs+5\zeta\otimesh\betab-3\chibh\rho+3\ld\chibh\sigma
\end{equation}
Since, by angular regularity, 
\begin{subequations}\label{eq:vanishing:Ab:angular}
\begin{gather}
  \limd \Phi_u^\ast(r\nablas\otimesh\betab)=\nablasc\otimesh \Bb=0\\
  \limd \Phi_u^\ast\Bigl( r(\chih,\alphab)\gs\Bigr)=(\Sigma,\Ab)\gammao=0
\end{gather}
\end{subequations}
and since in view of the smoothness assumptions the vanishing of the limits imply that the tensors on the left hand side of \eqref{eq:vanishing:Ab:angular} are in fact order $\mathcal{O}(r^{-1})$,
we obtain the propagation equation
\begin{equation}
    \partial_s\Phi_{u,s}^\ast\alphab-\frac{1}{2}\tr\chi\Phi_{u,s}^\ast\alphab=\mathcal{O}(r^{-2})\\
\end{equation}
which implies using \eqref{eq:vanishing:AB}, and Lemma~\ref{lemma:smoothness:ode}, that
\begin{equation}
  \limd\Phi_u^\ast\alphab=0\,,\qquad\lvert\alphab\rvert_{\gs}=\mathcal{O}(r^{-3})\,.
\end{equation}
This completes the proof of the proposition.
\end{proof}

\begin{rmk}
  The time-independence expressed in \eqref{eq:vanishing:PQ} and \eqref{eq:vanishing:Z} are related to the preservation of linear and angular momentum in the system, c.f.~\cite{rizzi:I}.
\end{rmk}

In Proposition~\ref{prop:non:radiating} we have thus discussed general properties of non-radiating solutions to the vacuum equations.
We now proceed under the \emph{stronger} assumption of \emph{time-periodicity}.
According to Definition~\ref{def:time-periodic} a time-periodic spacetime is endowed with a discrete isometry $\varphi$ which extends as a translation to future null infinity. Therefore any geometric quantity $w$, i.e.~tensors derived from the metric, such as components of the curvature tensor, which have a limit at future null infinity,
\begin{equation}
  \limd \Phi_u^\ast w=W(u,\xi)\,,\qquad u\in\mathbb{R}\,,\quad\xi\in\mathbb{S}^2\,,
\end{equation}
are periodic functions of retarded time $u$, if the spacetime is assumed to be time-periodic in sense of Definition~\ref{def:time-periodic}.
%This applies in particular to the Bondi mass~$M(u)$, which is also a monotone quantity, whence constant \eqref{eq:vanishing:Xi}.

As a first consequence of the stronger time-periodicity assumption we state here the time-independence of the lower order curvature components $\alpha$, $\beta$, which cannot be inferred from the weaker no-radiation condition alone. The argument exploits time-periodicity as in \cite{tod:I}.

\begin{prop}\label{prop:time:periodic:leading}
  Consider a solution to the vacuum equations as in Prop.~\ref{prop:non:radiating}, which in particular satisfies the smoothness assumptions of Def.~\ref{def:smooth:null:infinity}. If the spacetime is \emph{time-periodic}  then in addition to the conclusions of Prop.~\ref{prop:non:radiating} we have that
  \begin{gather*}
    Q=0\,,\qquad P=-2M\,,\\
    \partial_u A=\limd \partial_u\Phi_u^\ast(r^3\alpha)=0\,,\qquad\partial_u B=\limd \partial_u\Phi_u^\ast r^3\beta=0\,.
  \end{gather*}
Moreover, we have in fact
\begin{equation*}
  \lvert\betab\rvert_{\gs}=\mathcal{O}(r^{-4})\,.
\end{equation*}
\end{prop}

\begin{proof}

The Bianchi equations for $\beta$ are \eqref{eq:bianchi:beta:L} and \eqref{eq:bianchi:beta:Lb}, which we add and multiply by $r^3$ to obtain
\begin{multline}
    r^3(\nablas_L\beta+\nablas_{\Lb}\beta)+2r\tr\chi\, r^2\beta+r\tr\chib\,r^2\beta+r\omegab\,r^2\beta=r^3\divs\alpha+r^3\zeta^\sharp\cdot\alpha\\
  +r^3\nablas\rho+r^3\ld\nablas\sigma+3\zeta\, r^3\rho+3\ld\zeta \,r^3\sigma+2r^3\chih^\sharp\cdot\betab-r^3\lambda^\sharp\cdot\alpha\,.
\end{multline}
The limiting equation, under the assumptions of Section~\ref{sec:asymptotics:limits}, in particular \eqref{eq:asymptotics:alphabeta}, reads
\begin{equation}\label{eq:DT:B}
  D_TB=\nablasc P+\ld\nablasc Q+2\Sigma\cdot \Bb
\end{equation}
where $D$ denotes (the projection of) the Lie derivative, 
\begin{equation}
  D_TB=\limd\Phi_u^\ast\mathcal{L}\!\!\!/\,{}_T(r^3\beta)\,,
\end{equation}
and we have used Lemma~\ref{lemma:T:Lie} relating the covariant and Lie derivatives of a 1-form with respect to $T$.
While we have already shown that $\Bb=0$, and that $P(u,\xi)$, $Q(u,\xi)$ are time-independent, there is a priori no reason why $P$ and $Q$ should be spherically symmetric. The main idea to exploit time-periodicity --- as it is also done in \cite{tod:I} ---  is now to differentiate \eqref{eq:DT:B} another time, and use \eqref{eq:vanishing:PQ}:
\begin{equation}
  D_T^2B=\nablasc (T\cdot P)+\ld\nablasc (T\cdot Q)=0
\end{equation}
In fact, it is evident from \eqref{eq:DT:P} and \eqref{eq:DT:Q} that the last equality requires the vanishing of second order \emph{angular} derivatives of $\Bb$, and $\Ab$; this is ensured by $\Bb=\Ab=0$ in view of the smoothness assumptions imposed in Def.~\ref{def:smooth:null:infinity}.
Therefore, $B$ must be a linear function in $u$,
\begin{equation}\label{eq:B:linear}
  B(u_2,\xi)-B(u_1,\xi)=B_0(\xi) (u_2-u_1)\,;
\end{equation}
however, since $B$ is also periodic in $u$, we must have $B_0(\xi)=0$, hence
\begin{equation}
  \label{eq:vanishing:B}
  \partial_uB=0\,.
\end{equation}
We may now return to \eqref{eq:DT:B} to infer that
\begin{equation}
  \nablasc P+\ld\nablasc Q=0\,,
\end{equation}
which in view of the (Hodge) structure of this equation implies that $P$, and $Q$ are separately spherically symmetric.
But then integrating \eqref{eq:Hodge:Z:Xi} on the unit sphere implies
\begin{equation}\label{eq:vanishing:PQ:values}
  Q=\overline{Q}=0\qquad P=\overline{P}=-2M
\end{equation}
where $\overline{Q}$ denotes the average value of $Q$ on $\mathbb{S}^2$, and we used \eqref{eq:asymptotics:Nb:average}.
Similarly to \eqref{eq:DT:B} we derive the limiting equation for $\alpha$ from the Bianchi equation \eqref{eq:bianchi:alpha} using the asymptotics of Section~\ref{sec:asymptotics:limits}:
\begin{equation}\label{eq:DT:A}
  D_TA=\frac{1}{2}\nablasc\otimesh B-\frac{3}{2}\Sigma\,P-\frac{3}{2}\ld\Sigma Q
\end{equation}
Therefore, using \eqref{eq:vanishing:B},\eqref{eq:vanishing:PQ:values} and \eqref{eq:asymptotics:Sigma},
\begin{equation}
  D_T^2A=\frac{1}{2}\nablasc\otimesh D_TB-\frac{3}{2}M\Xi=0\,,
\end{equation}
where in view of \eqref{eq:DT:B}, \eqref{eq:vanishing:PQ:values} we invoke again the smoothness assumption of Def.~\ref{def:smooth:null:infinity} for the angular derivatives.
As in \eqref{eq:B:linear} it then follows from the time-periodicity assumption that
\begin{equation}
  \partial_uA=0\,.
\end{equation}

It remains to improve the order of vanishing of the curvature component $\betab$.
We return to the Bianchi equation \eqref{eq:bianchi:betab:L} which we may rewrite using Lemma~\ref{lemma:T:Lie} as follows:
\begin{equation}
    D_L\betab+\frac{1}{2}\tr\chi\betab=-\nablas\rho+\ld\nablas\sigma+\chih^\sharp\cdot\betab+3\zeta\rho-3\ld\zeta\sigma+2\chibh^\sharp\cdot\beta
\end{equation}
Since, by \eqref{eq:vanishing:PQ:values} and angular regularity,
\begin{subequations}
\begin{gather}
  \limd \Phi_u^\ast(r^3\nablas\rho,r^3\nablas\sigma)=(\nablasc P,\nablasc Q)=0\\
  \limd \Phi_u^\ast (r^3\chih^\sharp\cdot\betab)=\Sigma\cdot\Bb=0
\end{gather}
\end{subequations}
we obtain, in view of the smoothness assumptions, the propagation equation
\begin{equation}
    \partial_s\Phi_{u,s}^\ast(r^2\betab)-(\overline{\tr\chi}-\frac{1}{2}\tr\chi)\Phi_u^\ast(r^2\betab)=\mathcal{O}(r^{-2})
\end{equation}
which implies by Lemma~\ref{lemma:smoothness:ode},
\begin{equation}
  \limd \Phi_u^\ast(r^2\betab)=0\,,\qquad\lvert\betab\rvert_{\gs}=\mathcal{O}(r^{-4})\,.
\end{equation}
\end{proof}

\begin{rmk}
  One may proceed to derive a complete description of the leading order asymptotics in this setting. 
In fact, using the ideas in \cite{rizzi:I} one can exploit the freedom in the choice of $H$ (related to the gauge freedom in the choice of $S_0^\ast$) to arrange for $\Sigma=0$, since $Q=0$, and $\Xi=0$, c.f. (6-8) in \cite{rizzi:I}.  Returning to \eqref{eq:DT:A} we then obtain
\begin{equation}\label{eq:vanishing:B:conformal}
  \nablasc\otimesh B=0\,,
\end{equation}
which says that $B$ is a conformal Killing vectorfield on the unit sphere.
Moreover, by \eqref{eq:Hodge:Z:Xi} the torsion $Z$ is a gradient vectorfield on the sphere, here in view of \eqref{eq:asymptotics:Sigma},
\begin{equation}
  Z=\nablasc \phi^\prime\qquad \phi^\prime=\phi-\frac{1}{2}H
\end{equation}
where $\phi^\prime$ is a solution of vanishing mean to
\begin{equation}
  \stackrel{\circ}{\triangle\!\!\!\!/}\phi^\prime=\Nb-2M\,.
\end{equation}
%The remaining freedom in \eqref{eq:vanishing:B:conformal} corresponds to the choice of an axis of rotation.
\end{rmk}

\begin{rmk}
  We point out that the time-periodicity assumption is only used in the proof of Prop.~\ref{prop:time:periodic:leading} to deduce the time-independence of $B$, \eqref{eq:vanishing:B}, from the linearity of $B$ in $u$, \eqref{eq:B:linear}, and similarly for $A$.
  Clearly the time-periodicity assumption can be dropped if instead we assume the existence of the limits
  \begin{subequations}\label{eq:non:radiating:data:limits}
    \begin{gather}
    A(\xi)=\lim_{u\to-\infty}A(u,\xi)=\lim_{u\to-\infty}\limd \Phi_u^\ast(r^3\alpha)\\
    B(\xi)=\lim_{u\to-\infty}B(u,\xi)=\lim_{u\to-\infty}\limd \Phi_u^\ast(r^3\beta)\,.  
    \end{gather}
  \end{subequations}
  For then taking the limit $u_1\to-\infty$ in \eqref{eq:B:linear},
\begin{equation}
  B(u_2,\xi)-B(u_1,\xi)=B_0(\xi) (u_2-u_1)\,,
\end{equation}
forces $B_0(\xi)=0$, which then again implies $\partial_u B=0$; similarly for $A$.
The condition \eqref{eq:non:radiating:data:limits} requires in particular that \emph{on the level of the initial data}, i.e.~on a spacelike Cauchy hypersurface $\Sigma$,
  \begin{equation}\label{eq:non:radiating:data}
    \sup_{\Sigma} r^5\lvert \alpha \rvert_{\gs}<\infty\,,\qquad \sup_{\Sigma} r^4\lvert\beta\rvert<\infty\,.
  \end{equation}
This yields the following statement whose proof is identical to that given for Prop.~\ref{prop:time:periodic:leading}.
\end{rmk}

\begin{prop}\label{prop:non:radiating:data}
Let $(\mathcal{M},g)$ be a dynamical solution to the vacuum equations satisfying the asymptotics of Section~\ref{sec:asymptotics:limits} towards null infinity, and the smoothness assumptions of Def.~\ref{def:smooth:null:infinity}. If the spacetime arises from initial data such that \eqref{eq:non:radiating:data:limits} and \eqref{eq:non:radiating:data} hold, and if the spacetime is non-radiating, i.e.~the Bondi mass is constant along future null infinity, then all conclusions of Propositions~\ref{prop:non:radiating},~\ref{prop:time:periodic:leading} hold true, in particular
  \begin{gather*}
    \lvert\alphab\rvert_{\gs}=\mathcal{O}(r^{-3})\,,\qquad\lvert\betab\rvert_{\gs}=\mathcal{O}(r^{-4})\,,\quad\lvert\sigma\rvert=\mathcal{O}(r^{-4})\,,\\
    \partial_u\rho=0\,,\quad\partial_uB=0\,,\quad\partial_uA=0\,.
  \end{gather*}
\end{prop}

Note that the time-independence of $A$ and $B$ imply that if $A$ and $B$ are finite on the level of the initial data, then these bounds are propagated along future null infinity.

\subsection{Induction}
\label{sec:vanishing:induction}

We have seen that as a consequence of time-periodicity all leading order asymptotic quantities are time-independent. Now we will prove that the argument can be iterated yielding the statement that the vectorfield $T$ is Killing to all orders at infinity.

\begin{prop}
  \label{prop:vanishing:induction}
  Let $(\mathcal{M},g)$ be a solution to the vacuum equations satisfying the asymptotics of Section~\ref{sec:asymptotics:limits} towards future null infinity and the regularity assumptions of Definition~\ref{def:smooth:null:infinity}. If $(\mathcal{M},g)$ is time-periodic, then
  \begin{subequations}\label{eq:vanishing:induction}
  \begin{gather}
    \limd  \Phi_u^\ast r^k D_T(\alphab,r^2\betab, r^{3}\rho,r^3\sigma, r^3\beta,r^3\alpha)=0\label{eq:prop:vanishing}\\
    \limd \Phi_u^\ast r^k D_T(r\zeta,\chih,r^2\tr\chi,r^{-1}\chibh,r\tr\chib,r\omegab,\lambda)=0\qquad\forall k\in\mathbb{N}\cup\{0\}\\
    \limd \Phi_u^\ast r^{k+1} D_T(r^{-2}\gs,r^{-1}l,r^{-1}b)=0\,.
  \end{gather}
\end{subequations}
\end{prop}

Here $D$ denotes the Lie derivative defined in \eqref{def:D}, c.f.~Lemma~\ref{lemma:T:Lie}; see also the commutation relations of Lemma~\ref{lemma:T:commute} that are relevant for the proof.

\bigskip
\noindent\emph{Proof by induction.} 
We have by Propositions~\ref{prop:non:radiating},~\ref{prop:time:periodic:leading} that \eqref{eq:vanishing:induction} holds for $k=0$.
Let us now assume that \eqref{eq:vanishing:induction} holds for some $k\in\mathbb{N}$.
We will prove that \eqref{eq:vanishing:induction} then holds for $k+1$.

\smallskip
\noindent\textbf{Step 1a: Connection Coefficients.}
In a first step we shall show that the inductive assumptions on the curvatures $\alpha$, and $\beta$, namely
  \begin{equation}
    \limd \Phi_u^\ast r^k D_T(r^3\beta,r^3\alpha)=0\,, \label{induction:DT:alpha}
  \end{equation}   
allow us to improve the order of vanishing of the null second fundamental form $\chi$, and torsion $\zeta$,
\begin{equation}
    \limd \Phi_u^\ast r^{k+1} D_T(r\zeta,\chih,r^2\tr\chi)=0\,.
  \end{equation}

\subsubsection*{Propagation equations}

While the leading order argument presented in Section~\ref{sec:vanishing:leading} mainly draws consequences from the equations along future null infinity, the induction argument relies also on the use of propagation equations along the \emph{outgoing} null hypersurfaces $C_u^+$. 
These equations are expressed in Section~\ref{sec:null:structure} in terms of a null frame $(\Lb,L;E_A)$, where $E_A:A=1,2$ is an orthonormal frame on each sphere $S_{u,s}$ transported according to the equation
\begin{equation}
  \nabla_L E_A=-\zeta_A L\,.
\end{equation}
Since we are interested in ``time-derivatives'', namely Lie derivatives with respect to $T=\partial_u$, it shall be more convenient to use coordinate vectorfields, i.e.~Jacobi fields $X_a:a=1,2$ along the null generators of $C_u^+$:
\begin{equation}
  L=\frac{\partial}{\partial s}\qquad
  X_a=\frac{\partial}{\partial y^a}\qquad
  [L,X]=0
\end{equation}
The change from covariant to Lie derivatives is readily faciliated using Lemma~\ref{lemma:T:Lie}.

% We then find easily that the components of $X_a$ in the orthonormal frame $E_A$, $X_a=X_a^AE_A$, satisfy the propagation equation:
% \begin{equation}  \frac{\partial X_a^A}{\partial s}=\chi_{AB}X^B_a\,. \end{equation}

\subsubsection*{Null second fundamental form}

We can write \eqref{eq:structure:chih:L} using  Lemma~\ref{lemma:T:Lie} as
\begin{equation}
  D_L\chih=(\chih,\chih)\gs-\alpha\,.
\end{equation}
Since $L$ and $T$ commute by construction, we have in the above coordinates,
\begin{equation}
  \frac{\partial(\partial_u\chih_{ab})}{\partial s}=\lvert\chih\rvert^2\partial_u\gs_{ab}+2\chih^{cd}\partial_u\chih_{cd}\gs_{ab}-2(\partial_u\gs_{cd})\chih^{cf}\chih^{d}_{\phantom{d}f}\gs_{ab}-\partial_u\alpha_{ab}\,.
\end{equation}
Now by the inductive assumption \eqref{eq:vanishing:induction},
\begin{subequations}
\begin{gather}
  r^{k+1}\chih^{cd}\partial_u\chih_{cd}\gs_{ab}=\mathcal{O}(r^{-2})\\
  r^{k+1}(\lvert\chih\rvert^2 \partial_u\gs_{ab},\chih^{cf}\chih^{d}_{\phantom{d}f}\gs_{ab}\partial_u\gs_{cd})=\mathcal{O}(r^{-2})
\end{gather}
\end{subequations}
and in particular by \eqref{induction:DT:alpha},
\begin{equation}
     r^{k+1}\partial_u\alpha_{ab}=\mathcal{O}(r^{-2})\,,  
\end{equation}
we derive that
\begin{equation}
  \frac{\partial (r^{k+1}\partial_u\chih_{ab})}{\partial s}-\frac{k+1}{2}r^{k+1}\overline{\tr\chi}\partial_u\chih_{ab}=\mathcal{O}(r^{-2})
\end{equation}
which implies in view of Lemma~\ref{lemma:smoothness:ode} that
\begin{equation}
  \limd r^{k+1}\partial_u\chih_{ab}=0\,.\label{eq:induction:chih}
\end{equation}

Similarly for the trace part we have by the Raychadhuri equation \eqref{eq:structure:trchi:L},
\begin{equation}
    \partial_s\partial_u\tr\chi=-2\chih^{ab}\partial_u\chih_{ab}+(\partial_u \gs_{ab})\chih^{ac}\chih^{b}_{\phantom{b}c}-\tr\chi\partial_u\tr\chi
\end{equation}
  and the inductive assumptions,
  \begin{subequations}
  \begin{gather}
    r^{3}\chih^{ab}r^k\partial_u\chih_{ab}=\mathcal{O}(r^{-2})\\
    r^{k-1}(\partial_u \gs_{ab})r^4\chih^{ac}\chih^{b}_{\phantom{b}c}\,,
  \end{gather}    
  \end{subequations}
that 
\begin{equation}
  \frac{\partial (r^{k+3}\partial_u\tr\chi)}{\partial s}-\Bigl(\frac{k+3}{2}\overline{\tr\chi}-\tr\chi\Bigr)r^{k+3}\partial_u\tr\chi=\mathcal{O}(r^{-2})
\end{equation}
and hence
\begin{equation}
  \limd r^{k+3}\partial_u\tr\chi=0\,.\label{eq:induction:trchi}
\end{equation}

\subsubsection*{Torsion}
The equation \eqref{eq:structure:zeta:L} can be written using Lemma~\ref{lemma:T:Lie} as
\begin{equation}
  D_L\zeta+\chi^\sharp\cdot \zeta=-\beta
\end{equation}
We multiply by $r$ and decompose $\chi$ into its trace and trace-free part to obtain the equation that guarantees the existence of the limit \eqref{eq:asymptotics:zeta}:
\begin{equation}
  \partial_s(r\zeta_a)+\frac{1}{2}(\tr\chi-\overline{\tr\chi})r\zeta_a=-r\chih_a^b\zeta_b-r\beta_a
\end{equation}
Now differentiate, noting that $T$ and $L$ commute by construction, and multiply by $r^{k+1}$ to obtain,
\begin{multline}
  \frac{\partial( r^{k+1}\partial_u(r\zeta_a))}{\partial s}-\frac{k+1}{2}\overline{\tr\chi}\:r^{k+1}\partial_u(r\zeta_a)\\+\frac{1}{2}(\tr\chi-\overline{\tr\chi})r^{k+1}\partial_u(r\zeta_a)+\chih_a^br^{k+1}\partial_u(r\zeta_b)=\\=-\frac{1}{2}r^{k+1}\partial_u(\tr\chi-\overline{\tr\chi})r\zeta_a+r^{k+2}\partial_u\gs^{bc}\chih_{ab}\zeta_c-r^{k+1}(\partial_u\chih_a^b)r\zeta_b-r^{k+1}\partial_u(r\beta_a)\,.
\end{multline}
Therefore, in view of the inductive assumptions, and smoothness assumptions in the form,
\begin{subequations}
  \begin{gather}
    r^{k+1}\partial_u(r\zeta_b)=\mathcal{O}(1)\label{eq:zeta:nextorder}\\
    (r^{k+1}\partial_u\tr\chi\,,r^{k-3}\partial_u\gs_{ab},r^{k-1}\partial_u\chih_{ab}\,,r^{k+2}\partial_u\beta_a)=\mathcal{O}(r^{-2})\,,
  \end{gather}
\end{subequations}
the function defined by the left hand side of \eqref{eq:zeta:nextorder} is a solution to an o.d.e. satisfying the assumptions of Lemma~\ref{lemma:smoothness:ode}, which yields
\begin{equation}
  \limd r^{k+2}\partial_u\zeta_b=0\,.\label{eq:induction:torsion}
\end{equation}

\smallskip
\noindent\textbf{Step 1b:}
 In a next step we show that the inductive assumption on $\rho$, namely
  \begin{equation}
    \limd \Phi_u^\ast r^k D_T(r^3\rho)=0\,, \label{induction:DT:rho}
  \end{equation}   
allows us to improve the order of vanishing of the \emph{conjugate} null second fundamental form $\chib$,
\begin{equation}
    \limd \Phi_u^\ast r^{k+1} D_T(r^{-1}\chibh,r\tr\chib)=0\,.
  \end{equation}

\subsubsection*{Conjugate null second fundamental form}
The propagation equation \eqref{eq:structure:chibh:L} for $\chibh$ reads in view of Lemma~\ref{lemma:T:Lie},
\begin{equation}
      \partial_s\chibh_{ab}-\frac{1}{2}\tr\chi\chibh_{ab}=\chih^{cd}\chih_{cd}\gs_{ab}-(\nablas\otimesh\zeta)_{ab}-\frac{1}{2}\tr\chib\chih_{ab}+(\zeta\otimesh\zeta)_{ab}
\end{equation}
and thus
\begin{multline}
      \partial_s\partial_u\chibh_{ab}-\frac{1}{2}\tr\chi\partial_u\chibh_{ab}=\frac{1}{2}(\partial_u\tr\chi)\chibh_{ab}\\+\chih^{cd}\chih_{cd}\partial_u\gs_{ab}+2\chih^{cd}(\partial_u\chih_{cd})\gs_{ab}+2(\partial_u\gs_{cd})\chih^{cf}\chih^d_{\phantom{d}f}\gs_{ab}\\
      -\nablas_a\partial_u\zeta_b-\nablas_b\partial_u\zeta_a+\partial_u\gs_{ab}\divs\zeta+\gs_{ab}\divs\partial_u\zeta 
      -\frac{1}{2}(\partial_u\tr\chib)\chih_{ab}-\frac{1}{2}\tr\chib\partial_u\chih_{ab}\\+2(\partial_u\zeta_a)\zeta_b+2\zeta_a(\partial_u\zeta_b)-(\partial_u\gs_{cd})\zeta^c\zeta^d\gs_{ab}-2\zeta^c\partial_u\zeta_c\gs_{ab}-\zeta^c\zeta_c\partial_u\gs_{ab}\,.
\end{multline}
Note that the inductive assumption on $\partial_u\zeta_a$, and our angular regularity assumptions imply
\begin{equation}
  \limd \Phi_u^\ast(r^{k+3-2}\nablas D_T\zeta)=0\,,\qquad r^k\nablas_a\partial_u\zeta_b=\mathcal{O}(r^{-2})\,.
\end{equation}
Therefore we obtain
\begin{equation}
  \frac{\partial (r^k\partial_u\chibh_{ab})}{\partial s}-\Bigl(\frac{k}{2}\overline{\tr\chi}+\frac{1}{2}\tr\chi\Bigr)r^k\partial_u\chibh_{ab}=\mathcal{O}(r^{-2})
\end{equation}
which implies in view of Lemma~\ref{lemma:smoothness:ode} that
\begin{equation}
  \limd r^{k+1}\partial_u(r^{-1}\chibh_{ab})=0\,. \label{eq:induction:chibh}
\end{equation}

Now for the trace of $\chib$ we have \eqref{eq:structure:trchib:L} which yields
\begin{multline}
      \partial_s\partial_u\tr\chib+\frac{1}{2}\tr\chi\partial_u\tr\chib=-\frac{1}{2}(\partial_u\tr\chi)\tr\chib\\
      +2(\partial_u\gs_{ab})\nablas^a\zeta^b-2\divs\partial_u\zeta-\chih^{ab}\partial_u\chibh_{ab}-\chibh^{ab}\partial_u\chih_{ab}-2(\partial_u\gs_{ab})\chibh^{ac}\chih^{b}_{\phantom{b}c}\\
      +4\zeta^a\partial_u\zeta_a+2(\partial_u\gs_{ab})\zeta^a\zeta^b+2\partial_u\rho\,.
\end{multline}
Since in particular by the inductive assumption 
\begin{equation}
  r^{k+2}\partial_u\rho=\mathcal{O}(r^{-2})
\end{equation}
we have using the assumptions in this step,
\begin{equation}
  \frac{\partial (r^{k+2}\partial_u\tr\chib)}{\partial s}-\Bigl(\frac{k+2}{2}\overline{\tr\chi}-\frac{1}{2}\tr\chi\Bigr)r^{k+2}\partial_u\tr\chib=\mathcal{O}(r^{-2})\,
\end{equation}
which yields as desired
\begin{equation}
  \limd r^{k+1}\partial_u(r\tr\chib)=0\,.\label{eq:induction:trchib}
\end{equation}

\smallskip
\noindent\textbf{Step 1c:} The remaining connection coefficients are $\omegab$ and $\lambda$ which would vanish in a null frame constructed from two optical functions. By \eqref{eq:structure:omegab} we have
\begin{equation}
      \partial_s\partial_u\omegab=-6(\partial_u\gs_{ab})\zeta^a\zeta^b-12 \zeta^a\partial_u\zeta_a-2\partial_u\rho \label{eq:structure:omegab:u}
\end{equation}
which implies in view of the inductive assumptions that
\begin{equation}
  \limd r^{k+2}\partial_u\omegab=0\,. \label{eq:induction:omegab:improved}
\end{equation}
Note also that directly integrating \eqref{eq:structure:omegab:u} justifies the assumption \eqref{eq:asymptotics:omegab}.

For the coefficient $\lambda$ we combine \eqref{eq:structure:lambda} and \eqref{eq:structure:zeta:Lb} to obtain in view of Lemma~\ref{lemma:T:Lie}:
\begin{equation}
  D_L\lambda=\chib^\sharp\cdot \zeta+\nablas\omegab
\end{equation}
Therefore
\begin{multline}
  \partial_s\partial_u\lambda_a=-(\partial_u\gs_{bc})\chibh^c_{\phantom{c}a}\zeta^b+\gs^{bc}(\partial_u\chibh_{ca})\zeta_b+\chibh^b_{\phantom{b}a}\partial_u\zeta_b\\
  +\frac{1}{2}(\partial_u\tr\chib)\zeta_a+\tr\chib\partial_u\zeta_a+\nablas_a\partial_u\omegab
\end{multline}
and given that we have already proven \eqref{eq:induction:omegab:improved} we obtain by angular regularity
\begin{equation}
  r^{k+1}\nablas_a\partial_u\omegab=\mathcal{O}(r^{-2})\,.
\end{equation}
The inductive assumptions then imply
\begin{equation}
  \frac{\partial(r^{k+1}\partial_u\lambda_a)}{\partial s}-\frac{k+1}{2}\overline{\tr\chi}r^{k+1}\partial_u\lambda_a=\mathcal{O}(r^{-2})
\end{equation}
which shows that as desired
\begin{equation}
  \limd r^{k+1}\partial_u\lambda_a=0\,.
\end{equation}

\smallskip
\noindent\textbf{Step 2: Metric.}
By the very definition of the null second fundamental form
\begin{equation}
  D_L\gs=2\chi \label{eq:metric:L}
\end{equation}
we derive the following propagation equation for the metric components which ensures the existence of the limit \eqref{eq:asymptotics:metric}:
\begin{equation}
  \frac{\partial (r^{-2}\gs_{ab})}{\partial s}=2 r^{-2}\chih_{ab}+\Bigl(\tr\chi-\overline{\tr\chi}\Bigr)r^{-2}\gs_{ab}
\end{equation}
Here we multiply \eqref{eq:metric:L} by $r^k$, and use that $[L,T]=0$ to derive
\begin{equation}
  \frac{\partial(r^k\partial_u\gs_{ab})}{\partial s}-\Bigl(\frac{k}{2}\overline{\tr\chi}+\tr\chi\Bigr)r^k\partial_u\gs_{ab}
  =2r^k\partial_u\chih_{ab}+r^k(\partial_u\tr\chi)\gs\,.\label{eq:induction:metric:propagation}
\end{equation}
Given that we have already improved the inductive assumption on $\chi$, c.f.~\eqref{eq:induction:chih} and \eqref{eq:induction:trchi},
\begin{equation}
  r^{k+1}\partial_u\chih_{ab}=\mathcal{O}(r^{-1})\,,\qquad r^{k+1}(\partial_u\tr\chi)\gs=r^{k+3}(\partial_u\tr\chi)\gammao=\mathcal{O}(r^{-1})\,,
\end{equation}
we conclude that the right hand side of \eqref{eq:induction:metric:propagation} is $\mathcal{O}(r^{-2})$, hence
\begin{equation}
  \limd r^k\partial_u\gs_{ab}=0\,.
\end{equation}

The remaining metric components $l$ and $b$ in \eqref{eq:gauge:metric} satisfy the propagation equations \eqref{eq:propagate:lb} which after multiplying by $r^{k+1}$ yield
\begin{subequations}\label{eq:induction:lb}
  \begin{gather}
    \frac{\partial(r^{k+1}\partial_u l)}{\partial s}-\frac{k+1}{2}\overline{\tr\chi}r^{k+1}\partial_u l=-r^{k+1}\partial_u\omegab\\
    \frac{\partial(r^{k+1}\partial_u b^a)}{\partial s}-\frac{k+1}{2}\overline{\tr\chi}r^{k+1}\partial_u b^a=2\gs^{ac}\gs^{db}r^{k+1}\partial_u\gs_{cd}\zeta_b-2\gs^{ab}r^{k+1}\partial_u\zeta_b\,.
  \end{gather}
\end{subequations}
In view of the above obtained \eqref{eq:induction:omegab:improved}, and the inductive assumption \eqref{eq:vanishing:induction} we conclude that the right hand side of \eqref{eq:induction:lb} is $\mathcal{O}(r^{-2})$. Therefore by Lemma~\ref{lemma:smoothness:ode},
\begin{subequations}
  \begin{gather}
    \limd r^{k+1}\partial_u l=0\\
    \limd r^{k+1}\partial_u b^a=0\,.
  \end{gather}
\end{subequations}

\smallskip
\noindent\textbf{Step 3: Curvature.}
While Step 1 \& 2 mainly relied on the inductive assumptions \eqref{eq:vanishing:induction} used in conjuction with the propagation equations for the connection coefficients, the argument for the curvature components exploits the time-periodicity assumption similarly to the leading order argument of Section~\ref{sec:vanishing:leading}.

Recall here the leading order expression for the candidate vectorfield $T$ given in Lemma~\ref{lemma:T:lead}, and the commutation property of Lemma~\ref{lemma:T:commute} below.

\subsubsection*{Curvature components $\rho$, $\sigma$.}

The Bianchi equations \eqref{eq:bianchi:rho:L} and \eqref{eq:bianchi:rho:Lb} imply
  \begin{multline}
    2T(r^{3+k+1}\rho)+\frac{3}{2}r\bigl(\tr\chi-\overline{\tr\chi}\bigr)r^{3+k}\rho+\frac{3}{2}r\bigl(\tr\chib-\overline{\tr\chib}\bigr)r^{3+k}\rho=\\
    =r^{3+k+1}\Bigl[\divs\beta-(\zeta,\beta)-\frac{1}{2}(\chibh,\alpha)-\divs\betab-(\zeta,\betab)-\frac{1}{2}(\chih,\alphab)\Bigr]
  \end{multline}
  We observe that the terms on the right hand side are time-independent by the inductive assumptions. In fact, 
  \begin{subequations}
  \begin{gather}
    \limd  \partial_u\Bigl(r^{3+k+1}(\chih,\alphab)_{\gs}\Bigr)= \limd \bigl(r^{k} \partial_u\chih,\alphab\bigr)_{\gammao}+\limd \bigl(\chih,r^{k} \partial_u\alphab\bigr)_{\gammao}=0\\
    \limd \partial_u\Bigl(r^{3+k+1}(\zeta,\betab)_{\gs}\Bigr)=\limd (r^{k+1}\partial_u\zeta,r\betab)_{\gammao}+\limd (r\zeta,r^{k+1}\partial_u\betab)_{\gammao}=0\\
    \limd \partial_u\Bigl(r^{3+k+1}\divs\betab\Bigr)=\limd r^{k+4}\divs \partial_u\betab=0\,,\label{eq:vanishing:Bb:div}
  \end{gather}
  \end{subequations}
  where the last limit \eqref{eq:vanishing:Bb:div} follows from the assumption \eqref{eq:vanishing:induction} on the curvature, and the angular regularity assumptions, c.f.~Def.~\ref{def:smooth:null:infinity}.
  Moreover, by \eqref{eq:vanishing:induction},
  \begin{subequations}
    \begin{multline}
          \limd \partial_u\Bigl(r^{3+k+1}(\chibh,\alpha)_{\gs}\Bigr)=\\=\limd r^{-2}\bigl(r^{k-1}\partial_u\chibh,r^3\alpha\bigr)_{\gammao} +\limd r^{-2}\bigl(r^{-1}\chibh,r^{k+3}\partial_u\alpha\bigr)_{\gammao}=0
    \end{multline}
    \begin{multline}
         \limd \partial_u\Bigl(r^{3+k+1}(\zeta,\beta)_{\gs}\Bigr)=\\=\limd r^{-2}\bigl(r^{k+1}\partial_u\zeta, r^3\beta\bigr)_{\gammao}+\limd r^{-2}\bigl(r\zeta,r^{k+3}\partial_u\beta\bigr)_{\gammao}=0
    \end{multline}
  \begin{equation}
    \limd \partial_u\Bigl(r^{3+k+1}\divs\beta\Bigr)=\limd r^{k+1}\divs \partial_u(r^3\beta)=0\,.
  \end{equation}
  \end{subequations}
  Similarly on the left hand side,
  \begin{multline}
   \limd \partial_u\bigl(r(\tr\chi-\overline{\tr\chi})r^{3+k}\rho\bigr)=\\=\limd r^{k+1}\partial_u\bigl(\tr\chi-\overline{\tr\chi}\bigr)\,r^3\rho+\limd r(\tr\chi-\overline{\tr\chi})r^{k+3}\partial_u\rho=0\,.
  \end{multline}
  Therefore,
  \begin{equation}\label{eq:induction:DDT:rho}
    \limd \partial_u^2\bigl(r^{3+k+1}\rho\bigr)=0\,,
  \end{equation}
  which implies \emph{by time-periodicity} that
  \begin{equation}\label{eq:induction:DT:rho}
    \limd r^{k+1}\partial_u(r^3\rho)=0\,;
  \end{equation}
  time-periodicity is here exploited as in the proof of Prop.~\ref{prop:non:radiating} above, and as in the argument given in \cite{tod:I}.
  In complete analogy to the above, we prove
  \begin{equation}
    \limd r^{k+1}\partial_u(r^3\sigma)=0\,. \label{eq:induction:DT:sigma}
  \end{equation}

  \begin{rmk}
    Instead of using the time-periodicity assumption in this step, the statement \eqref{eq:induction:DT:rho} can be inferred from \eqref{eq:induction:DDT:rho} more generally by virtue of a smoothness assumption on null infinity. Indeed, in the context of Theorem~\ref{thm:stationary:non:radiating}, we have in view of \eqref{eq:thm:non:radiating:expansion} that by \eqref{eq:induction:DDT:rho}
    \begin{equation}
      \partial_u^2 \rho_+^{k+1}(u,\xi)=0\,.
    \end{equation}
    Therefore $\partial_u\rho_+^{k+1}$ is independent of $u$, and 
    \begin{equation}
      \rho_+^{k+1}(u_2)-\rho_+^{k+1}(u_1)=\partial_u\rho_+(\xi)(u_2-u_1)
    \end{equation}
    which implies after taking $u_1\to -\infty$ that by the key assumption \eqref{eq:thm:non:radiating:assumption}:
    \begin{equation}
      \partial_u\rho_+^{k+1}=0\,.
    \end{equation}
    Similarly for the curvature components $\sigma$, and $\beta$, $\alpha$ below.
  \end{rmk}

\subsubsection*{Curvature component $\beta$.}

We add the Bianchi equations \eqref{eq:bianchi:beta:L} and \eqref{eq:bianchi:beta:Lb} after multiplying by $r^3$, and use Lemma~\ref{lemma:T:Lie} to obtain:
  \begin{multline}
    2D_T(r^3\beta)+\frac{3}{2}\bigl(\tr\chi-\overline{\tr\chi}\bigr)r^3\beta+\frac{3}{2}\bigl(\tr\chib-\overline{\tr\chib}\bigr)r^3\beta\\-r\tr\chib\, r^2\beta-r^3\chibh^\sharp\cdot\beta-r^3\chih^\sharp\cdot\beta+r\omegab\, r^2\beta=\\
    =r^3\divs\alpha+r\zeta^\sharp\cdot r^2\alpha+\nablas( r^3\rho)+\ld\nablas( r^3\sigma)+3r\zeta\,r^2\rho+3r\ld\zeta\,r^2\sigma+2 r^3\chih\cdot\betab-r^3\lambda^\sharp\cdot\alpha \label{eq:beta:DT}
  \end{multline}
Multiplying the equation further by $r^{k+1}$ and differentiating in $u$, we prove that all terms on the right hand side are $u$-independent in the limit.
In particular, by virtue of \eqref{eq:induction:DT:rho} and \eqref{eq:induction:DT:sigma}, and angular regularity, we have
\begin{equation}
\limd \partial_u\Bigl(r^{k+1}\nablas_a(r^3\rho,r^3\sigma)\Bigr)=\limd r^{3+k+1}\nablas_a(\partial_u\rho,\partial_u\sigma)=0\,.
\end{equation}
Moreover, by the inductive assumptions,
    \begin{multline}
          \limd \partial_u\Bigl(r^{3+k+1}\gs^{bc}\chih_{ab}\betab_c\Bigr)=\\=\limd \gammao^{bc}r^k\partial_u\chih_{ab}\, r^2\betab_c+\limd \gammao^{bc}\chih_{ab}\, r^{k+2}\partial_u\betab_c=0
    \end{multline}
and similarly for all remaining terms, including on the left hand side where we have
\begin{subequations}
  \begin{multline}
    \limd \partial_u\Bigl( r^{3+k+1}\gs^{bc}\chibh_{ab}\beta_c\Bigr)=\\=\limd \gammao^{bc} r^{k-1}\partial_u\chibh_{ab}\,r^3\beta+\limd \gammao^{bc}r^{-1}\partial_u\chibh_{ab}\,r^{k+3}\partial_u\beta_c=0
  \end{multline}
  \begin{equation}
    \limd \partial_u\Bigl(r^{3+k+1}\omegab\, \beta_a\Bigr)=\limd r^{k+1}\partial_u\omegab\, r^3\beta_a+\limd r\omegab \,r^{k+3}\partial_u\beta_a=0\,.
  \end{equation}
\end{subequations}
 Therefore
 \begin{equation}
   \limd \partial_u^2(r^{3+k+1}\beta_a)=0\,,
 \end{equation}
 hence, by time-periodicity as above,
 \begin{equation}
   \limd r^{3+k+1}\partial_u\beta_a=0\,. \label{eq:induction:DT:beta}
 \end{equation}

\subsubsection*{Curvature component $\alpha$.}

Recall the Bianchi equation \eqref{eq:bianchi:alpha} in the form
\begin{multline}
  D_{\Lb}(r^3\alpha)-\frac{1}{2}\bigl(\tr\chib+3\overline{\tr\chib}\bigr)r^3\alpha-r^3(\chibh,\alpha)\gs+2\omegab r^3\alpha=\\
  =r^3\nablas\otimesh\beta+5r^3 \zeta\otimesh\beta-3r^3 \chih\rho-3r^3 \ld\chih\sigma\,.\label{eq:bianchi:alpha:D}
\end{multline}
Since by inductive assumption
\begin{equation}
  \limd r^{3+k} D_T\alpha=0
\end{equation}
 we here invoke the smoothness assumption of Section~\ref{sec:smoothness} to infer that
 \begin{equation}
   r^{3+k} D_LD_T\alpha=\mathcal{O}(r^{-2})
 \end{equation}
or simply
\begin{equation}
  \limd r^{3+k+1}\partial_u\partial_s\alpha_{ab}=0\,.\label{eq:limit:alpha:smooth}
\end{equation}
After multiplying \eqref{eq:bianchi:alpha:D} by $r^{k+1}$ and differentiating in $u$, the limit thus yields a statement for $D_T^2\alpha$ in view of \eqref{eq:limit:alpha:smooth}. Now given that we have already improved $D_T\chih$, c.f.~\eqref{eq:induction:chih}, and $D_T(\rho,\sigma)$, c.f.~\eqref{eq:induction:DT:rho}, \eqref{eq:induction:DT:sigma} we obtain
\begin{multline}
  \limd D_T\Bigl(r^{3+k+1} \chih(\rho,\sigma))\Bigr)=\\=\limd r^{k+1}D_T\chih \,(r^3\rho,r^3\sigma)+\limd \chih\,r^{k+1}D_T(r^3\rho,r^3\sigma)=0\,.
\end{multline}
Moreover by \eqref{eq:induction:DT:beta} and angular regularity we have
\begin{equation}
  \limd r^{3+k+1}\nablas\otimesh D_T\beta=0\,,
\end{equation}
and all other terms in \eqref{eq:bianchi:alpha:D} are $u$-independent in the limit by the inductive assumptions.
Therefore
\begin{equation}
  \limd r^{k+1}D_T^2(r^3\alpha)=0
\end{equation}
and again by time-periodicity
\begin{equation}
  \limd r^{3+k+1}\partial_u\alpha_{ab}=0\,.
\end{equation}

\smallskip
\noindent\textbf{Step 3b.}
Finally we improve the order of vanishing of $D_T\alphab$, and $D_T\betab$ with an argument that is similar to the treatment of the connection coefficients. Indeed, we can view the Bianchi equation \eqref{eq:bianchi:alphab:L} as a propagation equation:
\begin{multline}
  D_L(r^{-1}\alphab)-\frac{1}{2}(\tr\chi-\overline{\tr\chi})r^{-1}\alphab-(\chih,r^{-1}\alphab)\gs=\\=-r^{-1}\nablas\otimesh\betab+5r^{-1}\zeta\otimesh\betab-3r^{-1}\chibh\rho-3r^{-1}\ld\chibh\sigma
\end{multline}
Therefore by the inductive assumptions \eqref{eq:vanishing:induction} we have after differentiating in $u$ and multiplying by $r^{k+2}$,
\begin{equation}
  \partial_s\Bigl(r^{k+2}\partial_u(r^{-1}\alphab_{ab})\Bigr)-\frac{k+2}{2}\overline{\tr\chi}r^{k+2}\partial_u(r^{-1}\alphab_{ab})=\mathcal{O}(r^{-2})
\end{equation}
and thus
\begin{equation}
  \limd \Phi_u^\ast\Bigl( r^{k+1}D_T\alphab \Bigr)=0\,.
\end{equation}

Similarly we rewrite \eqref{eq:bianchi:betab:L} as
\begin{equation}
  D_L(r\betab)+\frac{1}{2}(\tr\chi-\overline{\tr\chi})r\betab-\chih^\sharp\cdot r\betab
  =-r\nablas\rho+r\ld\nablas\sigma+3\zeta r\rho-3\ld\zeta r\sigma+2r\chibh^\sharp\cdot\beta
\end{equation}
and observe that after differentiating in $u$, and multiplying by $r^{k+2}$ the right hand side is $\mathcal{O}(r^{-2})$ by the inductive assumptions \eqref{eq:vanishing:induction}. Moreover, by \eqref{eq:induction:chih} and \eqref{eq:induction:trchi} also the remaining terms on the left hand side satisfy
\begin{equation}
  r^{k+2}\partial_u\Bigl(\gs^{bc}\chih_{ab} r\betab_c\Bigr)=\mathcal{O}(r^{-2})\,.
\end{equation}
Therefore, as above, 
\begin{equation}
  \limd \Phi_u^\ast\Bigl(r^{k+2}D_T(r\betab)\Bigr)=0\,.
\end{equation}

\bigskip
We have thus proven \eqref{eq:vanishing:induction} holds for $k+1$, and hence completed the induction step, thus proving Proposition~\ref{prop:vanishing:induction}.\qed

\begin{proof}[Proof of Proposition~\ref{prop:time:periodic:killing:all:orders}]
Proposition~\ref{prop:vanishing:induction} in particular proves the statement of Proposition~\ref{prop:time:periodic:killing:all:orders} because in $(u,s,\vartheta^1,\vartheta^2)$ coordinates, 
\begin{subequations}
  \begin{gather}
    (\mathcal{L}_T g)_{ab}=\partial_u\gs_{ab}\\
    (\mathcal{L}_T g)_{au}=-2(\partial_u\gs_{ac})b^c-2\gs_{ac}(\partial_u b^c)\\
    (\mathcal{L}_T g)_{uu}=-\partial_u l+(\partial_u\gs_{ac})b^ab^c+2\gs_{ac}(\partial_ub^a )b^c\,.
  \end{gather}
\end{subequations} 
\end{proof}

\begin{lemma}
  \label{lemma:T:commute}
  The vectorfield $T$ commutes to leading order with factors in $r$, namely
  \begin{equation}
    \limd \partial_u r=0\,.
  \end{equation}
\end{lemma}
\begin{proof}
  Since $r(u,s)$ is defined by \eqref{eq:def:arearadius} as the area radius of $S_{u,s}$ and
\begin{equation}
  \frac{\partial}{\partial u}\sqrt{\det\gs_{s,u}}=\frac{1}{2}\gs^{ab}\partial_u \gs_{ab}\sqrt{\det\gs_{u,s}}
\end{equation}
  we have by Lemma~\ref{lemma:T:lead} that
  \begin{equation}
    \limd \partial_u r=\limd T\cdot r=\frac{1}{8}\limd r(\overline{\tr\chi}+\overline{\tr\chib})=0
  \end{equation}
  where the overline denotes the average on the sphere $S_{u,s}$.
\end{proof}

\begin{lemma}
  \label{lemma:T:Lie}
  Let $\theta$ be a $S_{u,s}$ 1-form (i.e.~$\theta(L)=\theta(\Lb)=0$), and $\omega$ a symmetric trace-free $S_{u,s}$ 2-form (i.e.~$\omega(L,\cdot)=\omega(\Lb,\cdot)=0$), then
  \begin{subequations}\label{eq:lemma:L:Lie}
    \begin{gather}
      D_L\theta=\nablas_L\theta+\chih^\sharp\cdot\theta+\frac{1}{2}\tr\chi \theta\\
        D_L\omega=\nablas_L\omega+(\chih,\omega)\gs+\tr\chi\omega\,.
    \end{gather}
  \end{subequations}
  Moreover, to leading order in $r$, 
  \begin{subequations}\label{eq:T:Lie}
    \begin{gather}
          D_T\theta=\nablas_T\theta+\frac{1}{2}\chibh^\sharp\cdot\theta+\frac{1}{4}\tr\chib\theta+\frac{1}{2}\chih^\sharp\cdot\theta+\frac{1}{4}\tr\chi\theta\\
  D_T\omega=\nablas_T\omega+\frac{1}{2}(\chibh,\omega)\gs+\frac{1}{2}(\chih,\omega)\gs+\frac{1}{2}\tr\chib\omega+\frac{1}{2}\tr\chi\omega\,.
    \end{gather}
  \end{subequations}
\end{lemma}

  \begin{proof}
    The identities \eqref{eq:lemma:L:Lie} are discussed in Chapter 1 of \cite{christodoulou:III}.
Moreover, by definition,   $D_T\theta=\Pi\mathcal{L}_T\theta$ where $\Pi$ is the projection to the $S_{u,v}$-spheres. Therefore
\begin{equation}
  (D_T\theta)_A=(\nablas_T\theta)_A+\theta\cdot\nabla_AT\,,
\end{equation}
and the formula follows from Lemma~\ref{lemma:T:lead}.
Similarly,
\begin{equation}
  (D_T\omega)_{AB}=(\nablas_T\omega)_{AB}+\omega(\nabla_{e_A}T,e_B)+\omega(e_A,\nabla_{e_B}T)
\end{equation}
which gives to leading order by Lemma~\ref{lemma:T:lead},
\begin{equation}
  D_T\omega=\nablas_T\omega+\frac{1}{2}\chib\times\omega+\frac{1}{2}\omega\times\chib+\frac{1}{2}\chi\times\omega+\frac{1}{2}\omega\times\chi\,.
\end{equation}
Since $\omega$ is trace-free symmetric this simplifies to \eqref{eq:T:Lie} using the formulas of Chapter~1 in \cite{christodoulou:III}.
  \end{proof}

%% file: extension.tex
In Section~\ref{sec:vanishing} we have proven that the time-like vectorfield $T$ constructed in Section~\ref{sec:gauge} generates an isometry at infinity to all orders, c.f.~Proposition~\ref{prop:time:periodic:killing:all:orders}.
We shall now prove that this vectorfield is in fact a Killing vectorfield in a neighborhood of infinity.

\begin{prop}\label{prop:vanishing:main}
  Let $(\mathcal{M},g)$ be a time-periodic solution to the vacuum equations satisfying the assumptions of Section~\ref{sec:smoothness}.
  Let $T$ be the time-like vectorfield constructed in Section~\ref{sec:gauge}. Recall that $T$ has the asymptotic form towards future null infinity given in Lemma~\ref{lemma:T:lead}, and satisfies
  \begin{equation}
    [L,T]=0\,,\qquad \nabla_L L=0\,,\qquad g(L,L)=0\,,
  \end{equation}
  in a neighborhood $\mathcal{D}_\omega$ of infinity, $\omega>0$; see \eqref{eq:D:omega} for the precise definition below. 
  Then
  \begin{equation}
    \mathcal{L}_T g=0\qquad\text{: on }\mathcal{D}_{\omega^\prime}
  \end{equation}
  for some $0<\omega^\prime<\omega$, i.e. $(\mathcal{M},g)$ is stationary in a neighborhood of infinity.

\end{prop}

The proof of Prop.~\ref{prop:vanishing:main} relies crucially on our unique continuation from infinity results for linear waves on asymptotically flat spacetimes proven in collaboration with A.~Shao in \cite{alexakis-schlue-shao}; in fact it employs the Carleman estimates developed therein (rather than the uniqueness theorem \emph{per se}) and combines them with the general framework developed by Ionescu and Klainerman in \cite{ionescu-klainerman:extension}.

We note that the latter provided an alternative (purely tensorial) to the method originally developed in \cite{alexakis-ionescu-klainerman}
on the problem of  extending  Killing vectorfields in Ricci flat manifolds using Carleman estimates. Interestingly, although for the purpose of extending Killing fields across a finite boundary both the method in \cite{ionescu-klainerman:extension} and the earlier in \cite{alexakis-ionescu-klainerman} are applicable,\footnote{The advantage of the former being that the Killing fields need not be tangential to the boundary.} in the case at hand where we seek to extend from infinity, \emph{only} the newer method in \cite{ionescu-klainerman:extension} is applicable. This is due to its tensorial nature, which allows us to evaluate the resulting tensorial equation against any frame. In particular, we use an asymptotically Cartesian frame, which has the advantage that the (connection) coefficients in the resulting wave equation decay fast enough towards infinity for the theorem in \cite{alexakis-schlue-shao} to apply. \footnote{If one employed the method in \cite{alexakis-ionescu-klainerman} one would be forced to use a null frame, for which the resulting coefficients decay too slowly.}

In Section~\ref{sec:ik} we derive the relevant equations from \cite{ionescu-klainerman:extension} in the present setting. Then in Section~\ref{sec:carleman} we restate the Carleman estimate of \cite{alexakis-schlue-shao} in \emph{physical space}; (this estimate was first proven in a conformally inverted space; see also \cite{ionescu-klainerman:ill-posed}).
Finally, in Section~\ref{sec:uniqueness} we prove the unique continuation theorem, and complete the proof of Theorem~\ref{thm:stationary:periodic}.

\subsection{Ionescu-Klainerman system of tensorial equations}
\label{sec:ik}

In Section~\ref{sec:asymptotics} we have defined the time-like vectorfield $T$ away from future null infinity according to
\begin{equation}
  [L,T]=0
\end{equation}
where $L$ is the geodesic generator of the outgoing null hypersurfaces $C_u^+$,
\begin{equation}
  \nabla_L L=0\,,\qquad g(L,L)=0\,.
\end{equation}
This implies of course that $T$ is a solution to the Jacobi equation
\begin{equation}
  \nabla_L\nabla_LT=R(L,T)\cdot L\,.
\end{equation}

Let $\pi$ denote the deformation tensor of $T$, $\pi=\mathcal{L}_T g$, 
and following \cite{ionescu-klainerman:extension} let
\begin{equation}
  B=\frac{1}{2}\bigl(\pi+\omega\bigr)
\end{equation}
where $\omega$ is an anti-symmetric 2-form defined by the transport equation
\begin{equation} \label{eq:omega}
  \nabla_L\omega_{\alpha\beta}=\pi_{\alpha\rho}\nabla_\beta L^\rho- \pi_{\beta\rho}\nabla_\alpha L^\rho\,;
\end{equation}
here and in the remainder of this subsection the components are expressed relative to an \emph{arbitrary} frame.

Moreover, define as in \cite{ionescu-klainerman:extension} the modified Lie derivative of the curvature $R$ by
\begin{subequations}
\begin{gather}
  W:=\hat{\mathcal{L}}_TR:=\mathcal{L}_T R-B\odot R\,,\\
    (B\odot R)_{\alpha_1\alpha_2\alpha_3\alpha_4}=\sum_{j=1}^4 B_{\alpha_j}^{\phantom{\alpha}\rho}R_{\alpha_1\ldots\rho\ldots\alpha_4} \label{eq:B.R}
\end{gather}
\end{subequations}
and also a tensor $\Pi$ algebraically similar to the Christoffel symbols (formerly denoted by $\Gamma$ in \cite{ionescu-klainerman:extension}), namely
\begin{equation}
  \Pi_{\alpha\beta\mu}:=\frac{1}{2}\Bigl(\nabla_\alpha \pi_{\beta\mu}+\nabla_\beta \pi_{\alpha\mu}-\nabla_\mu \pi_{\alpha\beta}\Bigr)
\end{equation}
and
\begin{equation}
  P_{\alpha\mu\beta}:=\Pi_{\alpha\beta\mu}-\nabla_\beta B_{\alpha\mu}\,.
\end{equation}
It is then proven in \cite{ionescu-klainerman:extension} Proposition~2.7 that $B$ and $P$ satisfy the transport equations
\begin{subequations}\label{eq:transport:BP}
\begin{gather}
  \nabla_L B_{\alpha\beta}=P_{\rho\beta\alpha}L^\rho-B_{\rho\beta} \nabla_\alpha L^\rho\label{eq:transport:B}\\
  \nabla_L P_{\alpha\beta\mu}=W_{\alpha\beta\mu\nu} L^\nu+R_{\alpha\beta\rho\nu} B_\mu^{\phantom{\mu}\rho} L^\nu-P_{\alpha\beta\rho}\nabla_\alpha L^\rho\label{eq:transport:P}\,.
\end{gather}
\end{subequations}

These transport equations complement a covariant wave equation satisfied by $W$.
Indeed, as a consequence of the Bianchi equations on Ricci flat spacetimes, the curvature satisfies a covariant wave equation:
\begin{equation}\label{eq:wave:R:covariant}
  \begin{split}    
  \Box R_{\alpha_1\alpha_2\alpha_3\alpha_4}&=+R_{\sigma\rho\alpha_3\alpha_4}R^{\sigma\phantom{\alpha_1\alpha_2}\rho}_{\phantom{\sigma}\alpha_1\alpha_2}+R_{\sigma\alpha_2\rho\alpha_4} R^{\sigma\phantom{\alpha_1\alpha_3}\rho}_{\phantom{\sigma}\alpha_1\alpha_3}+R_{\sigma\alpha_2\alpha_3\rho} R_{\phantom{\sigma}\alpha_1\alpha_4}^{\sigma\phantom{\alpha_1\alpha_4}\rho}\\
  &\quad -R_{\sigma\rho\alpha_3\alpha_4}R^{\sigma\phantom{\alpha_2\alpha_1}\rho}_{\phantom{\sigma}\alpha_2\alpha_1}-R_{\sigma\alpha_1\rho\alpha_4}R^{\sigma\phantom{\alpha_2\alpha_3}\rho}_{\phantom{\sigma}\alpha_2\alpha_3}-R_{\sigma\alpha_1\alpha_3\rho}R^{\sigma\phantom{\alpha_2\alpha_4}\rho}_{\phantom{\sigma}\alpha_2\alpha_4}
  \end{split}
\end{equation}
or for brevity,
\begin{equation}
  \Box R=R\odot R\,.
\end{equation}
Now using the commutation properties of covariant and Lie derivatives (c.f.~\cite{ionescu-klainerman:extension} Lemma 2.2) we have
\begin{multline}
  \Box (\mathcal{L}_T R)_{\alpha_1\alpha_2\alpha_3\alpha_4}=\mathcal{L}_T(\Box R_{\alpha_1\alpha_2\alpha_3\alpha_4})+\Pi_{\sigma\phantom{\sigma}\rho}^{\phantom{\sigma}\sigma}\nabla^\rho R_{\alpha_1\alpha_2\alpha_3\alpha_4}\\+\sum_{j=1}^4\Pi_{\alpha_j\phantom{\sigma}\rho}^{\phantom{\alpha}\sigma}\nabla_\sigma R_{\alpha_1\ldots\alpha_4}^{\phantom{\ldots}\rho}
  +\sum_{j=1}^4(\nabla^\sigma\Pi_{\alpha_j\sigma\rho})R_{\alpha_1\ldots\alpha_4}^{\phantom{\ldots}\rho}+\sum_{j=1}\Pi_{\alpha_j\sigma \rho}\nabla^\sigma R_{\alpha_1\ldots\alpha_4}^{\phantom{\ldots}\rho}
\end{multline}
In view of the presence of metric contractions in \eqref{eq:wave:R:covariant} we have, schematically,
\begin{equation}
  \mathcal{L}_T\Box R=\pi\odot R\odot R+R\odot \mathcal{L}_T R\,,
\end{equation}
or, more precisely,
\begin{equation}
  \begin{split}
  \mathcal{L}_T\Box R_{\alpha_1\alpha_2\alpha_3\alpha_4}
&=+\mathcal{L}_T R_{\sigma\rho\alpha_3\alpha_4}R^{\sigma\phantom{\alpha_1\alpha_2}\rho}_{\phantom{\sigma}\alpha_1\alpha_2}+\mathcal{L}_T R_{\sigma\alpha_2\rho\alpha_4} R^{\sigma\phantom{\alpha_1\alpha_3}\rho}_{\phantom{\sigma}\alpha_1\alpha_3}+\mathcal{L}_T R_{\sigma\alpha_2\alpha_3\rho} R_{\phantom{\sigma}\alpha_1\alpha_4}^{\sigma\phantom{\alpha_1\alpha_4}\rho}\\
  &\quad -\mathcal{L}_T R_{\sigma\rho\alpha_3\alpha_4}R^{\sigma\phantom{\alpha_2\alpha_1}\rho}_{\phantom{\sigma}\alpha_2\alpha_1}-\mathcal{L}_T R_{\sigma\alpha_1\rho\alpha_4}R^{\sigma\phantom{\alpha_2\alpha_3}\rho}_{\phantom{\sigma}\alpha_2\alpha_3}-\mathcal{L}_T R_{\sigma\alpha_1\alpha_3\rho}R^{\sigma\phantom{\alpha_2\alpha_4}\rho}_{\phantom{\sigma}\alpha_2\alpha_4}  \\  
&\quad+R_{\sigma\rho\alpha_3\alpha_4}\mathcal{L}_T R^{\sigma\phantom{\alpha_1\alpha_2}\rho}_{\phantom{\sigma}\alpha_1\alpha_2}+R_{\sigma\alpha_2\rho\alpha_4} \mathcal{L}_T R^{\sigma\phantom{\alpha_1\alpha_3}\rho}_{\phantom{\sigma}\alpha_1\alpha_3}+R_{\sigma\alpha_2\alpha_3\rho} \mathcal{L}_T R_{\phantom{\sigma}\alpha_1\alpha_4}^{\sigma\phantom{\alpha_1\alpha_4}\rho}\\
  &\quad -R_{\sigma\rho\alpha_3\alpha_4}\mathcal{L}_T R^{\sigma\phantom{\alpha_2\alpha_1}\rho}_{\phantom{\sigma}\alpha_2\alpha_1}-R_{\sigma\alpha_1\rho\alpha_4}\mathcal{L}_T R^{\sigma\phantom{\alpha_2\alpha_3}\rho}_{\phantom{\sigma}\alpha_2\alpha_3}-R_{\sigma\alpha_1\alpha_3\rho}\mathcal{L}_T R^{\sigma\phantom{\alpha_2\alpha_4}\rho}_{\phantom{\sigma}\alpha_2\alpha_4} \\
&\quad-2\pi_\sigma^{\phantom{\sigma}\lambda}R_{\lambda\rho\alpha_3\alpha_4}R^{\sigma\phantom{\alpha_1\alpha_2}\rho}_{\phantom{\sigma}\alpha_1\alpha_2}-\pi_\sigma^{\phantom{\sigma}\lambda} R_{\lambda\alpha_2\rho\alpha_4} R^{\sigma\phantom{\alpha_1\alpha_3}\rho}_{\phantom{\sigma}\alpha_1\alpha_3}-\pi_\sigma^{\phantom{\sigma}\lambda}R_{\lambda\alpha_2\alpha_3\rho} R_{\phantom{\sigma}\alpha_1\alpha_4}^{\sigma\phantom{\alpha_1\alpha_4}\rho}\\
&\quad-2\pi_\rho^{\phantom{\rho}\lambda}R_{\sigma\lambda\alpha_3\alpha_4}R^{\sigma\phantom{\alpha_1\alpha_2}\rho}_{\phantom{\sigma}\alpha_1\alpha_2}-\pi_\rho^{\phantom{\rho}\lambda} R_{\sigma\alpha_2\rho\alpha_4} R^{\sigma\phantom{\alpha_1\alpha_3}\rho}_{\phantom{\sigma}\alpha_1\alpha_3}-\pi_\rho^{\phantom{\rho}\lambda}R_{\sigma\alpha_2\alpha_3\lambda} R_{\phantom{\sigma}\alpha_1\alpha_4}^{\sigma\phantom{\alpha_1\alpha_4}\rho}\\
 &\quad\qquad +\pi_\sigma^{\phantom{\sigma}\lambda}R_{\lambda\alpha_1\rho\alpha_4}R^{\sigma\phantom{\alpha_2\alpha_3}\rho}_{\phantom{\sigma}\alpha_2\alpha_3}+\pi_\sigma^{\phantom{\sigma}\lambda}R_{\lambda\alpha_1\alpha_3\rho}R^{\sigma\phantom{\alpha_2\alpha_4}\rho}_{\phantom{\sigma}\alpha_2\alpha_4}\\
 & \quad\qquad +\pi_\rho^{\phantom{\rho}\lambda}R_{\sigma\alpha_1\lambda\alpha_4}R^{\sigma\phantom{\alpha_2\alpha_3}\rho}_{\phantom{\sigma}\alpha_2\alpha_3}+\pi_\rho^{\phantom{\rho}\lambda}R_{\sigma\alpha_1\alpha_3\lambda}R^{\sigma\phantom{\alpha_2\alpha_4}\rho}_{\phantom{\sigma}\alpha_2\alpha_4}
  \end{split}
\end{equation}
where
\begin{equation}
  \pi_\sigma^{\phantom{\sigma}\lambda}=\nabla_\sigma T^\lambda+\nabla^\lambda T_\sigma\,.
\end{equation}
Furthermore, by \eqref{eq:B.R} we have
\begin{multline}
  \Box (B\odot R)_{\alpha_1\alpha_2\alpha_3\alpha_4}=\sum_{j=1}^4 (\Box B_{\alpha_j}^{\phantom{\alpha}\rho})R_{\alpha_1\ldots\rho\ldots\alpha_4}\\+2\sum_{j=1}^4 \nabla^\sigma B_{\alpha_j}^{\phantom{\alpha}\rho}\nabla_\sigma R_{\alpha_1\ldots\rho\ldots\alpha_4}+\sum_{j=1}^4 B_{\alpha_j}^{\phantom{\alpha}\rho}\Box R_{\alpha_1\ldots\rho\ldots\alpha_4}
\end{multline}
and therefore
\begin{multline}
  \Box W_{\alpha_1\alpha_2\alpha_3\alpha_4}=\Box (\mathcal{L}_T R)_{\alpha_1\alpha_2\alpha_3\alpha_4}-\Box (B\odot R)_{\alpha_1\alpha_2\alpha_3\alpha_4}\\
   =\mathcal{L}_T(\Box R_{\alpha_1\alpha_2\alpha_3\alpha_4})+\Pi_{\sigma\phantom{\sigma}\rho}^{\phantom{\sigma}\sigma}\nabla^\rho R_{\alpha_1\alpha_2\alpha_3\alpha_4}+\sum_{j=1}^4\Pi_{\alpha_j\phantom{\sigma}\rho}^{\phantom{\alpha}\sigma} \nabla_\sigma R_{\alpha_1\ldots\alpha_4}^{\phantom{\ldots}\rho}\\
  +\sum_{j=1}^4(\nabla^\sigma P_{\alpha_j\rho\sigma})R_{\alpha_1\ldots\alpha_4}^{\phantom{\ldots}\rho}+\sum_{j=1}\Pi_{\alpha_j\sigma \rho}\nabla^\sigma R_{\alpha_1\ldots\alpha_4}^{\phantom{\ldots}\rho}\\
-2\sum_{j=1}^4 \nabla^\sigma B_{\alpha_j}^{\phantom{\alpha}\rho}\nabla_\sigma R_{\alpha_1\ldots\rho\ldots\alpha_4}-\sum_{j=1}^4 B_{\alpha_j}^{\phantom{\alpha}\rho}\Box R_{\alpha_1\ldots\rho\ldots\alpha_4}
\end{multline}
where we have used that
\begin{equation}\label{eq:wave:W:pi}
  \nabla^\sigma\Pi_{\alpha_j\sigma\rho}-\Box B_{\alpha_j\rho}=\nabla^\sigma\bigl(\Pi_{\alpha_j\sigma\rho}-\nabla_\sigma B_{\alpha_j\rho})=\nabla^\sigma P_{\alpha_j\rho\sigma}\,.
\end{equation}
In conclusion, we have schematically,
\begin{equation}
  \Box W=R\odot \mathcal{L}_T R+(R\odot R)\odot\pi+\nabla R\odot\nabla \pi+R\odot \nabla P+\nabla R\odot\nabla B+(R\odot R)\odot B
\end{equation}
or equivalently, since $\omega$ is anti-symmetric and
\begin{equation}
  \pi_{\alpha\beta}=B_{\alpha\beta}+B_{\beta\alpha}\,
\end{equation}
we can substitute $B$ in place of $\pi$ in \eqref{eq:wave:W:pi} and obtain, finally,
\begin{equation}\label{eq:wave:W}
  \Box W=R\odot W+\nabla R\odot\nabla B+R^2\odot B+R\odot\nabla P\,.
\end{equation}
The above can be viewed as a wave equation for $\mathcal{L}_T R$ with fast decaying coefficients, given that in this setting the curvature $R$ is  $\mathcal{O}(r^{-3})$; c.f.~Prop.~\ref{prop:non:radiating}. While this is morally the reason why the uniqueness theorems of \cite{alexakis-schlue-shao} are applicable, we have to revisit the underlying Carleman estimate because \eqref{eq:wave:W} is coupled to the differential equations \eqref{eq:transport:BP}.

\subsection{Carleman esimates in physical space}
\label{sec:carleman}

The results in \cite{alexakis-schlue-shao} concern linear wave equations on asymptotically flat spacetimes $(M,g)$.
The class of space-times $(\mathcal{M},g)$ of particular relevance here have \emph{positive mass} in the sense that
\begin{multline}\label{eq:metric:positivemass}
  g=g_{uu}\ud u^2-4K\ud u\ud v+g_{vv}\ud v^2+\sum_{A,B=1}^2 r^2 \gamma_{AB}\ud y^A\ud y^B\\ +\sum_{A=1}^2 (g_{Au}\,\ud u\,\ud y^A+g_{Av}\,\ud v\,\ud y^A)
\end{multline}
where
\begin{equation}
  K=1-\frac{2m}{r}\,,\qquad m\geq m_{\text{min}}>0\,,
\end{equation}
and the differential of $r$ satisfies
\begin{equation}\label{eq:dr}
  \Bigl(1+\frac{2m}{r}\Bigr)\ud r=\Bigl(1+\mathcal{O}(r^{-2})\Bigr)\ud v-\Bigl(1+\mathcal{O}(r^{-2})\Bigr)\ud u+\sum_{A=1}^2\mathcal{O}(r^{-1})\ud y^A\,.
\end{equation}
For the precise assumptions on the remaining metric components we refer the reader to Section~2.2 of \cite{alexakis-schlue-shao}. 
We recall here that while these assumptions include
\begin{equation}
  g_{uu}, g_{vv}=\mathcal{O}_1(r^{-3})\,,
\end{equation}
the cross-terms are allowed to behave like
\begin{equation}
  g_{Au}, g_{Av}=\mathcal{O}_1(r^{-1})\,.
\end{equation}
In particular, it is shown in Appendix A in \cite{alexakis-schlue-shao}, that the above class of metrics includes the Kerr family.
In Section~\ref{sec:data} we will use this fact to show that the space-times considered in Theorem~\ref{thm:stationary:periodic}, \ref{thm:stationary:non:radiating} are in the class of positive mass spacetimes in this sense of \cite{alexakis-schlue-shao}, c.f.~Lemma~\ref{lemma:data:ass} below.

In these coordinates we define
\begin{equation}\label{eq:f:F}
  f=-\frac{1}{uv}\,,\qquad F(f)=\log f-f^{2\delta}\,,
\end{equation}
for some fixed $\delta>0$. In view of the freedom of choice of the constant of integration in \eqref{eq:dr}, we can set $(u=0,v=0)$ on a $(u,v)$-level set where $r$ is arbitrarily large; in other words, by the choice of the $u=0$, $v=0$ level sets, the domain
\begin{equation}\label{eq:D:omega}
  \mathcal{D}_\omega=\bigl\{ (u,v,y^1,y^2): 0<f(u,v)<\omega\bigr\}
\end{equation}
is an \emph{arbitrarily small} neighborhood of spatial infinity, which extends to small portions of future and past null infinity $\mathcal{I}^+$, $\mathcal{I}^-$. (It is useful to keep in mind here the Schwarzschild geometry, where the freedom in the choice of the $u=0$ and $v=0$ hypersurfaces corresponds to the choice of a radius $r$ where $r^\ast:=v-u=0$, c.f.~Section~2.2 in \cite{alexakis-schlue-shao}.)

For convenience we introduce the weight function $\mathcal{W}$, and associated weighted norms, for any $\lambda>0$ and domain $\mathcal{D}=\mathcal{D}_\omega$, $\omega>0$,
\begin{equation}\label{eq:carleman:notation}
  \mathcal{W}=e^{-\lambda F}f^\frac{1}{2}\,,\qquad \lVert \cdot \rVert_{\mathcal{W}} = \lVert \mathcal{W} \cdot \rVert_2\,,\qquad \lVert \phi\rVert_2^2= \int_{\mathcal{D}}\phi^2\dm{g}
\end{equation}
relevant for the formulation and use of the Carleman estimates that follow.

\begin{thm}[Carleman estimate near infinity for linear waves, \cite{alexakis-schlue-shao}]\label{thm:carleman}
  Let $(M,g)$ be an asymptotically flat spacetime with positive mass $m\geq m_{\text{min}}>0$, and $\mathcal{D}_\omega$ a neighborhood of infinity of the form \eqref{eq:D:omega} for some $\omega>0$.
Let $\delta>0$, and let $\phi$ be a smooth function on $\mathcal{D}_\omega$ that vanishes to all orders at infinity, in the sense that for each $N\in\mathbb{N}$ there exists an exhaustion $(\mathcal{D}_k)$ of $\mathcal{D}_\omega$ such that
\begin{equation}
  \lim_{k\to\infty}\int_{\partial\mathcal{D}_k}r^N(\phi^2+\lvert\partial\phi\rvert^2)=0\,.
\end{equation}
Then, for $\omega>0$ sufficiently small and $\lambda>0$ sufficiently large,
\begin{equation}\label{eq:carleman:wave}
  \lambda^3\VertW{ f^\delta\phi }+\lambda \VertW{  f^{-\frac{1}{2}}\Psi^\frac{1}{2}\nabla \phi } \lesssim \VertW{ f^{-1}\Box\phi } \,,
\end{equation}
where $\Psi$ is defined by \footnote{As discussed in \cite{alexakis-schlue-shao}, the function  $\Psi$ measures the strength of the pseudo-convexity of the level sets of $f$.}
\begin{equation}
  \Psi:=\frac{m_{\text{min}}\log r}{r}\,.
\end{equation}
\end{thm}

\begin{proof}
  We translate the Carleman estimate established in \cite{alexakis-schlue-shao} Proposition~4.1 into physical space.
The relevant estimate (4.12) therein,\footnote{We suppress for ease of presentation the common additional factor $e^{2\lambda f^{2\delta}}$ in all integrals.}
\begin{multline}\label{eq:carleman:inverted}
  \int_{\mathcal{D}}f^{-2\lambda+1}\lvert \Box_{\overline{g}}\overline{\phi}\rvert^2\dm{\overline{g}}\gtrsim
  \lambda^3\int_{\mathcal{D}}f^{-2\lambda+1}f^{-2+2\delta}\overline{\phi}^2\dm{\overline{g}}\\
  +\lambda\int_{\mathcal{D}} f^{-2\lambda+1}f^{2\delta-1}\lvert E_N\overline{\phi}\rvert^2\dm{\overline{g}}+\lambda\int_{\mathcal{D}} f^{-2\lambda+1} f^{-1}\Psi\sum_{A=1}^2\lvert E_A\overline{\phi}\rvert^2\dm{\overline{g}}\,,
\end{multline}
(where $E_A:A=1,2$ is tangential, and $E_N$ is normal to the level sets of $f$, and unit with respect to $\overline{g}$)
 is expressed for functions $\overline{\phi}$ in the inverted spacetime $(\mathcal{D},\overline{g})$ which are related to the corresponding functions $\phi$ in physical space $(\mathcal{D},g)$ via the transformations
\begin{subequations}
\begin{gather}
  \overline{g}=\Omega^2 g\qquad \Omega^2=K^{-1}f^2\qquad \dm{\overline{g}}=\Omega^4\dm{g}\\
  \overline{\phi}=\Omega^{-1}\phi\,.
\end{gather}
\end{subequations}
In particular, c.f.~Section~5.2 in \cite{alexakis-schlue-shao},
\begin{equation}\label{eq:Box:conformal}
  \Box_{\overline{g}}\overline{\phi} = \Omega^{-3}\Box_g\phi+\frac{1}{6}\Bigl(\Omega^{-2}R_g-R_{\overline{g}}\Bigr)\Omega^{-1}\phi
\end{equation}
where $R_g$, and $R_{\overline{g}}$ are the scalar curvatures of $g$, and $\overline{g}$ respectively, and we proved, c.f.~Section~5.4 therein, that
\begin{equation}
  \lvert \Omega^{-2} R_g-R_{\overline{g}}\rvert \lesssim f^{-1+\delta}\,.
\end{equation}

Firstly, we rewrite the zeroth order term on the right hand side of \eqref{eq:carleman:inverted},
\begin{equation}
  \int f^{-2\lambda+1}f^{-2+2\delta}\overline{\phi}^2\dm{\overline{g}}=\int f^{-2\lambda+1}K^{-2}f^{2\delta}\phi^2\dm{g}\,. \label{eq:carleman:zero}
\end{equation}
Therefore, by rewriting the principal term in \eqref{eq:carleman:inverted} using \eqref{eq:Box:conformal},
\begin{multline}
  \int_{\mathcal{D}} f^{-2\lambda+1}\lvert \Box_{\overline{g}}\overline{\phi}\rvert^2\dm{\overline{g}}\lesssim\\
  \lesssim\int_{\mathcal{D}} f^{-2\lambda+1}\Bigl\{\Omega^{-6}\lvert\Box_g\phi\rvert^2+\lvert \Omega^{-2}R_{\overline{g}}-R_g\rvert^2\Omega^{-2}\phi^2\Bigr\}\Omega^4\dm{g}\\
  \lesssim \int_{\mathcal{D}} f^{-2\lambda+1} K^2 f^{-2}\lvert \Box_g\phi\rvert^2\dm{g}+\int_{\mathcal{D}} f^{-2\lambda+1}K^{-2} f^{2\delta}\phi^2\dm{g}\,,
\end{multline}
we see that the scalar curvature term coming from the conformal inversion can be absorbed on the right hand side.

Secondly, we treat all derivatives uniformly (retaining only the weaker weight of the tangential derivatives), and write
\begin{multline}
  \int_{\mathcal{D}} f^{-2\lambda+1}f^{-1}\Psi\lvert\overline{\nabla}\overline{\phi}\rvert^2_{\overline{g}}\dm{\overline{g}}=\\=\int_{\mathcal{D}} f^{-2\lambda+1}f^{-1}\Psi\Omega^2\Bigl\{\Omega^{-4}\lvert\nabla\Omega\rvert^2\phi^2-2\Omega^{-3}\phi \nabla\Omega\cdot\nabla\phi+\Omega^{-2}\lvert\nabla\phi\rvert^2\Bigr\}\dm{g}\,;
\end{multline}
we observe that
\begin{equation}
  \lvert\nabla\Omega\rvert^2\lesssim K^{-3}\lvert\nabla K\rvert^2 f^2+K^{-1}\lvert\nabla f\rvert^2\lesssim f^3
\end{equation}
so the first term is bounded by 
\begin{equation}
  \int f^{-2\lambda+1}f^{-1}\Psi\Omega^2\Omega^{-4}\lvert\nabla\Omega\rvert^2\phi^2\dm{g} \lesssim \int f^{-2\lambda+1}\Psi\phi^2\dm{g}
\end{equation}
and is thus bounded by the zeroth order term \eqref{eq:carleman:zero} above, as long as $4\delta<1$.
We can thus also use the zeroth order term to bound the mixed term by Cauchy's inequality; in conclusion we have
\begin{multline}
    \int_{\mathcal{D}} f^{-2\lambda+1}f^{-1}\Psi\lvert\overline{\nabla}\overline{\phi}\rvert^2_{\overline{g}}\dm{\overline{g}}\gtrsim\\\gtrsim \int_{\mathcal{D}} f^{-2\lambda+1} f^{-1}\Psi\lvert\nabla\phi\rvert^2_g\dm{g}+\int_{\mathcal{D}} f^{-2\lambda+1}f^{2\delta}\phi^2\dm{g}\,.
\end{multline}
With the notation \eqref{eq:carleman:notation} we have thus obtained the physical space Carleman estimate,
\begin{equation}
  \lambda^3\lVert e^{-\lambda F}f^\frac{1}{2}f^\delta\phi\rVert_2+\lambda \lVert e^{-\lambda F}\Psi^\frac{1}{2}\nabla \phi\rVert_2 \lesssim \lVert e^{-\lambda F}f^\frac{1}{2}f^{-1}\Box\phi\rVert_2\,,
\end{equation}
which is precisely the statement \eqref{eq:carleman:wave}.
\end{proof}

Recalling the o.d.e.'s derived in Section~\ref{sec:ik} we also need a Carleman estimate for propagation equations of the form
\begin{equation}
  L\cdot \phi=\Phi\,,\qquad L=\frac{\partial}{\partial v}\,.
\end{equation}
We can readily adapt the proof of \cite{alexakis-ionescu-klainerman} Lemma~A.3 to obtain a Carleman estimate from infinity under the infinite order vanishing assumption.
\begin{lemma}\label{lemma:carleman:ode}
  Let $(M,g)$ be an asymptotically flat spacetime, and $L=\partial_v$ the outgoing null vectorfield in the coordinates \eqref{eq:metric:positivemass}.
Let $\phi$ be a smooth function on $\mathcal{D}_\omega$ that vanishes to all orders at infinity, in the sense that
for any $N\in\mathbb{N}$ there is an exhaustion $(\mathcal{D}_k)$ of $\mathcal{D}_\omega$ such that
\begin{equation}\label{eq:carleman:ode:assumption}
  \lim_{k\to\infty}\int_{\partial\mathcal{D}_k}r^N\phi^2=0\,.
\end{equation}
Then for any $q\geq 1$, and $\lambda>0$ sufficiently large,
\begin{equation}\label{eq:carleman:ode}
  \lambda \VertW{ \frac{1}{r}f^{-1}r^{-q}\phi } \lesssim \VertW{ f^{-1}r^{-q} L\cdot \phi }\,.
\end{equation}
\end{lemma}

\begin{proof}
  Given a smooth functions $\phi$ and $F$ on $\mathcal{D}$, consider
\begin{equation}
  \psi=e^{-\lambda F}\phi\,, 
\end{equation}
then for any vectorfield $L$,
\begin{equation}
  e^{-\lambda F}L\cdot \phi=\lambda (L\cdot F)\psi+L\cdot \psi
\end{equation}
and so
\begin{multline}
  \int_{\mathcal{D}} e^{-\lambda F}(L\cdot \phi) (L\cdot F)\psi=\int_{\mathcal{D}}\lambda(L\cdot F)^2\psi^2+\frac{1}{2}(L\cdot F)L\cdot \psi^2\\
  =\int_{\mathcal{D}} \lambda(L\cdot F)^2\psi^2-\frac{1}{2}(L\cdot LF)\psi^2-\frac{1}{2}(L\cdot F)(\nabla_\alpha L^\alpha)\psi^2\,.
\end{multline}
as long as there is no contribution from the boundary term,
\begin{equation}
  \int_{\partial\mathcal{D}} (L\cdot F)\psi^2 =0\,.
\end{equation}

With the choices \eqref{eq:f:F}, and $L=\partial_v$, we have
\begin{equation}
  (L\cdot F)^2\simeq\frac{1}{v^2}\simeq L\cdot (LF)
\end{equation}
and in the asymptotically flat setting of \eqref{eq:metric:positivemass},
\begin{equation}
  \nabla_\alpha L^\alpha=\mathcal{O}(r^{-1})\,;
\end{equation}
moreover, by assumption \eqref{eq:carleman:ode:assumption},
\begin{equation}
  \lim_{k\to\infty}\int_{\partial \mathcal{D}_k}(L\cdot F)\psi^2=\lim_{k\to\infty}\int_{\partial\mathcal{D}_k}\frac{1}{v}f^{-2\lambda}\phi^2=0\,,\qquad \lambda\gg 1\,.
\end{equation}
Therefore
\begin{equation}
  \int e^{-\lambda F}(L\cdot \phi) (L\cdot F)\psi\gtrsim \lambda \int \frac{1}{v^2}\psi^2
\end{equation}
and hence by Cauchy-Schwarz,
\begin{equation}\label{eq:carleman:ode:prelim}
  \lVert e^{-\lambda F}L\cdot \phi\rVert_2\gtrsim \lambda \lVert \frac{1}{v}\psi\rVert_2 
\end{equation}
or, since $r\gtrsim v$,
\begin{equation}
  \lambda \lVert e^{-\lambda F}\frac{1}{r}\phi\rVert_2\lesssim \lVert e^{-\lambda F}L\cdot \phi\rVert_2\,.
\end{equation}

Finally, we can apply \eqref{eq:carleman:ode:prelim} to the function $f^{-\frac{1}{2}}r^{-q}\phi$; since
\begin{equation}
  L\cdot f^{-\frac{1}{2}}r^{-q}=\mathcal{O}(v^{-1}f^{-\frac{1}{2}}r^{-q})
\end{equation}
we have under the same assumptions
\begin{equation}
  \lambda \lVert e^{-\lambda F}\frac{1}{r}f^{-\frac{1}{2}}r^{-q}\phi\rVert_2\lesssim \lVert e^{-\lambda F} f^{-\frac{1}{2}}r^{-q}L\cdot \phi\rVert_2
\end{equation}
which is precisely the statement of \eqref{eq:carleman:ode}.
\end{proof}

\subsection{Construction of space-time domain from initial data}
\label{sec:data}

\input{data-lemma}

\subsection{Unique continuation of time-like Killing vectorfields from infinity}
\label{sec:uniqueness}

We shall now prove a unique continuation result from infinity for Killing vectorfields based on the Carleman estimates of Theorem~\ref{thm:carleman} and Lemma~\ref{lemma:carleman:ode}.

\begin{prop}[Extension of Killing vectorfields from infinity]\label{prop:stationary:conditional}
Let $(M,g)$ be an asymptotically flat spacetime with positive mass (c.f.~Section~\ref{sec:carleman}) and rapidly decaying curvature in the sense that uniformly with respect to an orthonormal frame
\begin{equation}\label{eq:stationary:condition:curvature}
  \lvert R_{\alpha\beta\gamma\delta}\rvert =\mathcal{O}(r^{-3})\,,\qquad \vert \nabla_\epsilon R_{\alpha\beta\gamma\delta}\rvert =\mathcal{O}(r^{-4})\,.
\end{equation}
Let $T$ be a time-like vectorfield on $\mathcal{D}_\omega$, $\omega>0$, chosen to be the binormal on future null infinity, and extended according to (c.f.~Section~\ref{sec:gauge})
\begin{equation}
  [L,T]=0\,,
\end{equation}
where $L$ is the affine outgoing null geodesic vectorfield.
Then we have: \\If $T$ satisfies the Killing equation to all orders at infinity, i.e.~for all $N\in\mathbb{N}$ there is an exhaustion $(\mathcal{D}_k)$ such that %(where $R$ is the Riemann curvature tensor)
  \begin{equation}\label{eq:stationary:condition}
    \lim_{k\to\infty}\int_{\partial\mathcal{D}_k}r^N\mathcal{L}_T g=0\qquad
    \lim_{k\to\infty}\int_{\partial\mathcal{D}_k}r^N\mathcal{L}_T R=0\,,
  \end{equation}
  then $(M,g)$ is in fact a genuine Killing vectorfield for $(\mathcal{M},g)$, namely
  \begin{equation}
    \mathcal{L}_T g \equiv 0 \,,\qquad \mathcal{L}_T R \equiv 0 \qquad\text{ : on }\mathcal{D}_{\omega^\prime}
  \end{equation}
  for some $0<\omega^\prime<\omega$.
\end{prop}

In Section~\ref{sec:vanishing} we have proven that all smooth \emph{time-periodic} solutions to the Einstein vacuum equations satisfy the conditions of Proposition~\ref{prop:stationary:conditional}. In particular, in Propositions~\ref{prop:non:radiating} and~\ref{prop:time:periodic:leading} we have shown that \eqref{eq:stationary:condition:curvature} holds, and in Proposition~\ref{prop:vanishing:induction} of Section~\ref{sec:vanishing:induction} we have proven that \eqref{eq:stationary:condition} is verified in the time-periodic setting. The following proof in conjunction with Proposition~\ref{prop:vanishing:induction} thus completes the proof of Theorem~\ref{thm:stationary:periodic}.

\bigskip
\noindent\emph{Proof of Prop.~\ref{prop:stationary:conditional}.}
The strategy is to apply the Carleman estimates of Theorem~\ref{thm:carleman} and Lemma~\ref{lemma:carleman:ode} to the system of equations \eqref{eq:wave:W} and \eqref{eq:transport:BP}, and derivatives thereof.

These equations are covariant and can be expressed in any sytem of coordinates. 
In order for Theorem~\ref{thm:carleman} to be applicable to wave equations derived in Section~\ref{sec:ik}
we have to ensure a sufficiently rapid fall-off of the Christoffel symbols, 
which is achieved by evaluating all tensors relative to asymptotically \emph{Cartesian coordinates}.
It is at this point that the method of \cite{ionescu-klainerman:extension} is essential.

\subsubsection*{Cartesian coordinates}

Given the coordinates $(u,s,y^1,y^2)$ of Section~\ref{sec:gauge} we may pass to Cartesian coordinates $(x^0,x^1,x^2,x^3)$ such that the metric takes the form
\begin{equation}\label{eq:metric:cartesian}
  g=-\ud {x^0}\otimes \ud {x^0}+\sum_{i=1}^3\ud {x^i}\otimes\ud {x^i}+\sum_{\mu,\nu=0}^3\mathcal{O}_2(\frac{1}{r})\ud x^\mu\ud x^\nu
\end{equation}
and $\partial_{x^0}\rvert_{x^i}$ coincides with $\partial_u\rvert_{s,y^A}$ as $r\to\infty$ where
\begin{equation}
  r=\sqrt{(x^1)^2+(x^2)^2+(x^3)^2}\,.
\end{equation}
Then, in these coordinates,
\begin{equation}
  \Gamma_{\mu\nu}^\alpha=\mathcal{O}_1\bigl(\frac{1}{r^2}\bigr)\,.
\end{equation}
It follows that
\begin{multline}
  \nabla^\alpha\nabla_\alpha W_{\mu\nu\gamma\delta}=g^{\alpha\beta}\partial_\alpha\nabla_\beta W_{\mu\nu\gamma\delta}\\+\mathcal{O}\bigl(\frac{1}{r^2}\bigr)\sum_{\lambda,\beta=0}^3\bigl(\nabla_\lambda W_{\mu\nu\gamma\delta}+\nabla_\beta W_{\lambda\nu\gamma\delta}+\nabla_\beta W_{\mu\lambda\gamma\delta}+\nabla_\beta W_{\mu\nu\lambda\delta}+\nabla_\beta W_{\mu\nu\gamma\lambda}\bigr)
\end{multline}
and
\begin{multline}
\partial_\alpha\nabla_\beta W_{\mu\nu\gamma\delta}=\partial_\alpha\partial_\beta W_{\mu\nu\gamma\delta}+\mathcal{O}\bigl(\frac{1}{r^3}\bigr)\sum_{\lambda=0}^3\bigl(W_{\lambda\nu\gamma\delta}+W_{\mu\lambda\gamma\delta}+W_{\mu\nu\lambda\delta}+W_{\mu\nu\gamma\lambda}\bigr)\\
+\mathcal{O}\bigl(\frac{1}{r^2}\bigr)\sum_{\lambda=0}^3\bigl(\partial_\alpha W_{\lambda\nu\gamma\delta}+\partial_\alpha W_{\mu\lambda\gamma\delta}+\partial_\alpha W_{\mu\nu\lambda\delta}+\partial_\alpha W_{\mu\nu\gamma\lambda}\bigr)
\end{multline}
which shows that the wave equations satisfied by the components of $W$ in these coordinates, for brevity now simply denoted by $(W)$, are related to the components of $\Box W$, denoted by $(\Box W)$, via
\begin{equation}
  \Box(W)=\mathcal{O}\bigl(\frac{1}{r^2}\bigr)\nabla (W)+\mathcal{O}\bigl(\frac{1}{r^3}\bigr)(W)+(\Box W)\,.
\end{equation}

\begin{nota}
  Here and in the following we denote by $(T)$ any of the components $T_{\alpha_1\ldots\alpha_k}$ of a tensorfield $T$ in the asymptotically Cartesian coordinate $(x^0,x^1,x^2,x^3)$ for which the metric takes the form \eqref{eq:metric:cartesian}.
\end{nota}
In conclusion, by virtue of \eqref{eq:wave:W},
\begin{multline}\label{eq:wave:W:components}
    \Box (W)=\mathcal{O}(R)(W)+\mathcal{O}(\frac{1}{r^3})(W)+\mathcal{O}\bigl(\frac{1}{r^2}\bigr)\nabla (W)\\+\mathcal{O}(\nabla R)\nabla( B)+\mathcal{O}(R^2)(B)+\mathcal{O}(R)\nabla (P)\,.
\end{multline}

Moreover, it will be convenient to have uniform bounds on the components of $\nabla L$, and $\nabla \nabla L$ in these coordinates.
Since
\begin{equation}
  L=\frac{\partial }{\partial x^0}+\sum_{i=1}^3\frac{x^i}{r}\frac{\partial}{\partial x^i}+\sum_{\mu=0}^3\mathcal{O}(r^{-1})\frac{\partial}{\partial x^\mu}
\end{equation}
we have
\begin{subequations}
\begin{gather}
  \nabla_\mu L^\rho=\partial_\mu L^\rho+\Gamma_{\mu\sigma}^\rho L^\sigma=\mathcal{O}(r^{-1})\\
  \nabla_\mu\nabla_\nu L^\rho=\mathcal{O}(r^{-2})\,.
\end{gather}
\end{subequations}
Also note that
\begin{equation}\label{eq:commute:L:cartesian}
   [\nabla_L,\nabla_\mu]=\sum_{\nu=0}^3\mathcal{O}(r^{-1})\frac{\partial}{\partial x^\nu}\,.
\end{equation}

In these coordinates, the covariant equations \eqref{eq:transport:BP} yield simple transport equations for the components of $B$, and $P$, (evaluated against the Cartesian frame above) simply denoted by $(B)$ and $(P)$, of the form:

\begin{subequations}\label{eq:transport:BP:components}
\begin{gather}
  \nabla_L (B)=(P)+\mathcal{O}(\frac{1}{r})(B) \label{eq:transport:B:components}\\
  \nabla_L (P)=(W)+\mathcal{O}(R)(B)+\mathcal{O}(\frac{1}{r})(P)\,.\label{eq:transport:P:components}
\end{gather}
\end{subequations}
Furthermore, after commuting \eqref{eq:transport:BP} with $\nabla_\mu$, taking into account \eqref{eq:commute:L:cartesian}, we obtain the following propagation equations for the cartesian derivatives of the components of $P$, and $B$, simply denoted by $\nabla (P)$, and $\nabla (B)$:
\begin{subequations}
\begin{equation}
  \nabla_L \nabla (B)=\nabla( P)+\mathcal{O}(\frac{1}{r})(P)+\mathcal{O}(\frac{1}{r})\nabla (B)+\mathcal{O}(\frac{1}{r^2})(B)\label{eq:transport:nabla:B}\\
\end{equation}
\begin{multline}
      \nabla_L \nabla (P)=\nabla (W)+\mathcal{O}(\frac{1}{r})(W)\\+\mathcal{O}(\nabla R)(B)+\mathcal{O}(R)\nabla (B)+\mathcal{O}(\frac{1}{r}(R))(B) +\mathcal{O}(\frac{1}{r})\nabla( P)+\mathcal{O}(\frac{1}{r^2})(P)\label{eq:transport:nabla:P}
\end{multline}
\end{subequations}

\subsubsection*{Infinite order vanishing condition}
We note that the assumptions \eqref{eq:stationary:condition} imply that $(W)$, $(B)$ and $(P)$ vanish to all orders at infinity.
In fact, by assumption $\pi=\mathcal{L}_Tg$ vanishes to all orders at infinity, and since $\omega$ satisfies the transport equation \eqref{eq:omega}, 
\begin{equation}
  \nabla_L(\omega)=\mathcal{O}(\frac{1}{r})(\pi)\,,
\end{equation}
(and $(\omega)=0$ on future null infinity by construction), also $(\omega)$ vanishes to all orders, which immediately implies that both
\begin{subequations}
\begin{gather}
  (B)=(\pi)+(\omega)\\
  (W)=(\mathcal{L}_T R)+\mathcal{O}(R)(B)
\end{gather}
\end{subequations}
vanish to all orders at infinity.
Since by \eqref{eq:transport:B:components} $(P)$ is directly related to $(B)$ by
\begin{equation}
    (P)=\nabla_L (B)+\mathcal{O}(\frac{1}{r})(B)
\end{equation}
we have that also $(P)$ vanishes to all orders. (Alternatively that can be infered from integrating \eqref{eq:transport:P:components}.)
Therefore by assumption \eqref{eq:stationary:condition} Theorem~\ref{thm:carleman} can be applied to the functions $(W)$, and Lemma~\ref{lemma:carleman:ode} to the functions $(P)$ and $(B)$, as well as its derivatives $\nabla(P)$, and $\nabla(B)$.

\subsubsection*{Cutoff functions}

The Carleman estimates are in fact not applied to the component functions $(W)$, and $(B)$, etc., but instead to $\chi \cdot (W)$, $\chi \cdot (B)$, etc., where $\chi$ is a cut-off function whose level sets coincide with those of $F\circ f$,
\begin{equation}
  \chi=1 \text{ : on } \mathcal{D}_{\omega_0}\,, \qquad \chi =0 \text{ : on } \mathcal{D}_{\omega_1}^c\,,\quad \omega_0<\omega_1<\omega\,.
\end{equation}
This ensures that the resulting functions satisfy the support conditions of Theorem~\ref{thm:carleman} in the interior; this part of the argument is entirely standard for unique continuation problems.

Given that the components of $W$ satisfy an equation of the form
\begin{equation}
  \Box(W)=\mathcal{M}((W),\nabla (B), (B),\nabla (P))\,,
\end{equation}
where $\mathcal{M}$ for brevity refers to a multiple of the quantities that follow,
then the function $\chi\cdot (W)$ clearly satisfy
\begin{multline}
      \Box(\chi (W))=(\Box \chi)(W)+\nabla\chi\cdot \nabla (W)+\chi \mathcal{M}((W),\nabla (B), (B),\nabla (P))\\
      =(\Box \chi)(W)+\nabla\chi\cdot \nabla (W)+\mathcal{M}(\nabla\chi)+\mathcal{M}(\chi (W), \nabla (\chi (B)), \chi (B), \nabla (\chi (P)))\,.
\end{multline}
Since the terms $\Box\chi$, and $\nabla\chi$, are only supported in the cut-off region of $\chi$, 
it suffices to focus in the following argument on the last terms, the multiples $\mathcal{M}$ coming from the equation for $W$.

For simplicity in notation, \emph{we shall thus suppress both the cut-off and the component notation in the argument that follows}.

\subsubsection*{Weighted Carleman estimates for systems}

We now proceed to prove the uniqueness of $W$, and $B$ using the Carleman estimates of Section~\ref{sec:carleman}. 

Let $\lambda>0$.
By \eqref{eq:carleman:wave} applied to $W$ (recall that by convention we simply write $W$, $B$, etc,  instead of $\chi\cdot(W)$, $\chi\cdot (B)$, etc.) we have
\begin{equation}\label{eq:carleman:init}
  \lambda^3\VertW{ f^\delta W } +\lambda \VertW{  f^{-\frac{1}{2}}\Psi^\frac{1}{2}\nabla W } \lesssim \VertW{ f^{-1}\Box W }
\end{equation}
and by \eqref{eq:wave:W:components} we can estimate
\begin{multline}\label{eq:carleman:W:init}
  \VertW{ f^{-1}\Box W }\lesssim \VertW{ f^{-1}RW }+\VertW{ f^{-1}r^{-3}W }+\VertW{ f^{-1}r^{-2}\nabla W }\\+ \VertW{f^{-1}(\nabla R)\nabla B}+ \VertW { f^{-1}R^2  B }+ \VertW{ f^{-1}R \nabla P }+\VertW{\nabla\chi\mathcal{M}}\,;
\end{multline}
(here and in the following $\nabla\chi\mathcal{M}$ refers to a term composed of multiples of $W$, $B$, etc, which is however only supported in the cut-off region of $\chi$, and will remain on the right hand side of the inequalities.)
Since by assumption
\begin{subequations}
  \begin{gather}
  R=\mathcal{O}(r^{-3})\,,\qquad \nabla R=\mathcal{O}(r^{-4})\,,    \\
  f^{-1}\lesssim r^2\,,\qquad f\gtrsim \frac{1}{r^2}\,,\qquad f\Psi \gtrsim\frac{1}{r^3}\,,
  \end{gather}
\end{subequations}
we can absorb the first three terms on the right hand side of \eqref{eq:carleman:W:init}, on left hand side of \eqref{eq:carleman:init} for $\lambda\gg 1$ sufficiently large, as long as $0<2\delta<1$.

In order to control the term involving $B$ on the right hand side of \eqref{eq:carleman:W:init}, we add to \eqref{eq:carleman:init} the inequality
\begin{multline}\label{eq:carleman:B:add}
\lambda \VertW{\frac{1}{r}f^{-1}\frac{1}{r^4}B }\lesssim\\\lesssim \VertW{ f^{-1}r^{-4}\nabla_L B } \lesssim \VertW{ f^{-1}r^{-4}  P } + \VertW{ f^{-1} r^{-5} B} +\VertW{ \nabla\chi\mathcal{M}}\,.
\end{multline}
which is obtained using the Carleman estimate \eqref{eq:carleman:ode} and \eqref{eq:transport:B:components};
the new term on the left hand side in particular controls
\begin{equation}\label{eq:carleman:B:init}
  \lambda \VertW{\frac{1}{r}f^{-1}\frac{1}{r^4}B } \gtrsim \VertW{ f^{-1}R^2 B }\,.
\end{equation}
While the second term on the right hand side of \eqref{eq:carleman:B:add} is already controlled by \eqref{eq:carleman:B:init}, so as to absorb the first term, we also add to \eqref{eq:carleman:init} the inequality
\begin{multline}
  \lambda \VertW{ \frac{1}{r}f^{-1}r^{-3} P}\lesssim \VertW{f^{-1}r^{-3}\nabla_L P}\lesssim\\
  \lesssim\VertW{f^{-1}r^{-3}W}+\VertW{f^{-1}r^{-3}R B}+\VertW{f^{-1}r^{-4}P}+\VertW{\nabla\chi\mathcal{M}}
\end{multline}
which is a consequence of the Carleman estimate \eqref{eq:carleman:ode} and \eqref{eq:transport:P:components}.
Note that all terms of the right hand side are already controlled.

%\begin{equation}  \lambda \VertW{ \frac{1}{r}f^{-1}r^{-3} P}\gtrsim \VertW{ f^{-1}r^{-4} P } \end{equation}

In other words, up to this point we have shown  that for $\lambda\gg 1$ large enough,
\begin{multline}\label{eq:carleman:int}
  \lambda \VertW{f^{-1}r^{-5}B } +  \lambda \VertW{ f^{-1}r^{-4} P}+\lambda^3\VertW{ f^\delta W } +\lambda \VertW{  f^{-\frac{1}{2}}\Psi^\frac{1}{2}\nabla W } \lesssim\\
  \lesssim  \VertW{f^{-1}(\nabla R)\nabla B}+ \VertW{ f^{-1}R \nabla P } +\VertW{\nabla\chi\mathcal{M}}\,.
\end{multline}

So as to absorb the term involving $\nabla B$, we add to \eqref{eq:carleman:int} the inequality
\begin{multline}\label{eq:carleman:nabla:B:add}
  \lambda \VertW{ \frac{1}{r}f^{-1}r^{-3}\nabla B } \lesssim \VertW{ f^{-1}r^{-3}\nabla_L \nabla B}\lesssim\\
  \lesssim \VertW{ f^{-1}r^{-3}\nabla P}+\VertW{ f^{-1}r^{-4}P}\\+\VertW{ f^{-1}r^{-4}\nabla B}+\VertW{ f^{-1}r^{-5}B}+\VertW{\nabla\chi\mathcal{M}}
\end{multline}
which in turn is obtained using the Carleman estimate \eqref{eq:carleman:ode} and \eqref{eq:transport:nabla:B}.
As desired, the new term on the left hand side in particular controls
\begin{equation}\label{eq:carleman:nabla:B:init}
  \lambda \VertW{ \frac{1}{r}f^{-1}r^{-3}\nabla B }\gtrsim\VertW{f^{-1}(\nabla R)\nabla B}\,.
\end{equation}
Moreover all terms on the right hand side of \eqref{eq:carleman:nabla:B:add} except for the first term involving $\nabla P$ can be absorbed on the left hand sides of \eqref{eq:carleman:int} and \eqref{eq:carleman:nabla:B:add}; (the cut-off term of course remains on the right hand side).
Now we add to \eqref{eq:carleman:int} the inequality
  \begin{multline}\label{eq:carleman:nabla:P:int}
    \lambda \VertW{ \frac{1}{r}f^{-1}\frac{1}{r^2}\nabla P} \lesssim  \VertW{ f^{-1}r^{-2}\nabla_L\nabla P }\lesssim\\
    \lesssim \VertW{ f^{-1}r^{-2}\nabla W}+\VertW{ f^{-1}r^{-3}W}+\VertW{ f^{-1}r^{-2}(\nabla R)B}+\VertW{ f^{-1}r^{-2}R \nabla B}\\
    +\VertW{ f^{-1}r^{-3}R B} +\VertW{ f^{-1}r^{-3}\nabla P}+\VertW{ f^{-1}r^{-4}P}+\VertW{\nabla\chi\mathcal{M}}\,.
  \end{multline}
which is the result of the Carleman estimate \eqref{eq:carleman:ode} and \eqref{eq:transport:nabla:P}.

%; in particular the new term on the left hand side controls
%\begin{equation}\label{eq:carleman:nabla:P:init}   \lambda \VertW{ \frac{1}{r}f^{-1}\frac{1}{r^2}\nabla P}\gtrsim \VertW{ f^{-1}r^{-3}\nabla P} \end{equation}

Now in particular the first term on the right hand side involving $\nabla W$ can be absorbed on the left hand side of \eqref{eq:carleman:init},
  \begin{equation}
    \lambda \VertW{  f^{-\frac{1}{2}}\Psi^\frac{1}{2}\nabla W } \gtrsim \VertW{ f^{-1}r^{-2}\nabla W}\,,
  \end{equation}
because
\begin{equation}
  f^{-\frac{1}{2}}\Psi^\frac{1}{2}\gtrsim f^{-1} (f\Psi)^\frac{1}{2} \gtrsim f^{-1} r^{-\frac{3}{2}}
\end{equation}
using that the strength of the pseudoconvexity is bounded below by $\Psi\gtrsim r^{-1}$.
 Also the second term on the right hand side of \eqref{eq:carleman:nabla:P:int} is already controlled because
\begin{equation}
   \lambda^3 \VertW{ f^\delta W } \gtrsim \VertW{ f^{-1}r^{-3}W}\,,
\end{equation}
as long as  $0<\delta<1/2$.
Moreover, all remaining terms on the right hand side of \eqref{eq:carleman:nabla:P:int} can be absorbed by the terms introduced on the left hand side of \eqref{eq:carleman:int}, \eqref{eq:carleman:nabla:B:add}, \eqref{eq:carleman:nabla:P:int}.

Finally, also the remaining term on the right hand side of \eqref{eq:carleman:int} involving $\nabla P$ can be absorbed on the left hand side of \eqref{eq:carleman:nabla:P:int}.

Following standard procedure, we now restrict the integrals on the left hand side to the smaller domain $\mathcal{D}_{\omega_0}$, where $\chi=1$, while integrals on the right hand side are over $\mathcal{D}_{\omega_1}\setminus\mathcal{D}_{\omega_0}$ where $\nabla\chi$ is supported. Since $F(f)$ is monotone increasing, the Carleman weight $e^{-\lambda F}$ can be dropped from the inequality and we obtain
\begin{multline}
  \lambda^3\lVert{f^\frac{1}{2} f^\delta W }\rVert_2 + \lambda\lVert{  \Psi^\frac{1}{2}\nabla W }\rVert_2
  +\lambda \lVert f^{-\frac{1}{2}}r^{-5}B \rVert_2 +\lambda\lVert f^{-\frac{1}{2}}r^{-4} P \rVert_2 \\+ \lambda\lVert f^{-\frac{1}{2}}r^{-4}\nabla B \rVert_2+\lambda \lVert f^{-\frac{1}{2}}r^{-3}\nabla P\rVert_2\lesssim \lVert f^\frac{1}{2} \nabla\chi\mathcal{M}\rVert_2
\end{multline}
which by taking $\lambda\to\infty$ implies in particular that $B\equiv 0$  and 
\begin{equation}
  \mathcal{L}_T g\equiv 0\,, \qquad\mathcal{L}_T R \equiv 0 \qquad\text{: on }\mathcal{D}_{\omega_0}\,.
\end{equation}

\qed

%% file: data-lemma.tex
In this section we prove that the space-times under consideration in Theorems~\ref{thm:stationary:periodic}, \ref{thm:stationary:non:radiating}, arising from initial data which are suitably close to a Kerr solution, fall under the assumptions of Section~\ref{sec:carleman}. The argument applies to any space-time that admits a time-like vectorfield which statisfies the Killing equation to sufficiently high order, in particular time-periodic, and smooth non-radiating space-times as established in Section~\ref{sec:vanishing}.

In \eqref{eq:initial:data} we have imposed that $g\rvert_{\Sigma}$, and $k\rvert_\Sigma$  asymptote suitably fast to the values induced by a Kerr metric $g_{(m,a)}^{\text{Kerr}}$. More specifically, we assume that $g\rvert_\Sigma$ and $\partial_t g\rvert_\Sigma$ are obtained on $\Sigma=\{t=0\}$ from the expression (for some $m>0$, and $0<\lvert a\rvert< m$),
\begin{multline}\label{eq:initial:data:specific}
  g=g^{\text{Kerr}}_{(m,a)}+g^\infty=
  -\Bigl(1-\frac{2m}{r}\Bigr)\ud t^2+\mathcal{O}_2^\infty(r^{-3})\ud t^2\\
  +\Bigl(1+\frac{2m}{r}-\frac{a^2}{r^2}\sin^2\vartheta^1+\frac{2m}{r}\frac{a^2}{r^2}\cos^2\vartheta^1\Bigr)\ud r^2+\mathcal{O}_2^\infty(r^{-4})\ud r^2\\
+r^2\bigl(\ud \vartheta^1\bigr)^2+r^2\sin^2\vartheta^1\bigl(\ud\vartheta^2\bigr)^2+\sum_{A,B=1}^2\mathcal{O}_2^\infty(r^{0})\ud \vartheta^A\ud\vartheta^B\\
+\sum_{A=1}^2\mathcal{O}_2^\infty(r^{-1})\ud t\ud \vartheta^A+\sum_{A=1}^2\mathcal{O}_2^\infty(r^{-3})\ud r\ud \vartheta^A+\mathcal{O}_2^\infty(r^{-4})\ud t\ud r\,.
\end{multline}
Note that the explicit leading order terms arise from the expansion in $1/r$ of the components of the Kerr metric in Boyer Lindquist coordinates; see e.g.~\cite{hawking-ellis}.

%Recall Kerr
%\begin{multline}  g^{\text{Kerr}}_{(m,a)}=-\ud t^2+\frac{2mr}{\rho^2}\Bigl(a\sin^2\vartheta^1\ud\vartheta^2-\ud t\Bigr)^2+\frac{\rho^2}{\Delta}\ud r^2\\+\Bigl(r^2+a^2\Bigr)\sin^2\vartheta^1\bigl(\ud\vartheta^2\bigr)^2+\rho^2\bigl(\ud \vartheta^1\bigr)^2\end{multline}
%where $\Delta=r^2-2Mr+a^2$, and $\rho^2=r^2+a^2\cos^2\vartheta^1$. 
%In particular, the space-time metric agrees with the Kerr solution near spatial infinity up to order $r^{-3}$ in the $g_{tt}$, and $g_{rr}$-components, and up to order $r^{-1}$ in the $g_{t\vartheta^A}$-components.

\begin{lemma}\label{lemma:data:ass}
  Let $(M,g)$ be a space-time satisfying the asymptotics of Section~\ref{sec:asymptotics}, and $g$ on $\Sigma\subset\mathcal{M}$ be given by \eqref{eq:initial:data:specific}.
  Suppose the conclusion of Prop.~\ref{prop:time:periodic:killing:all:orders} hold, namely with $T$ constructed as in Section~\ref{sec:gauge} we have that \eqref{eq:prop:killing:all:orders} holds.
  Then there exists a domain $\mathcal{D}$ whose boundary at future and past null infinity contains segments $\mathcal{I}^+_{u_0}$, and $\mathcal{I}^-_{v_0}$, respectively, c.f.~Fig.~\ref{fig:thm:domain}, and a system of coordinates on $\mathcal{D}$ for which $g$ satisfies the assumptions of Theorem~\ref{thm:carleman}, c.f.~Section.~\ref{sec:carleman}.  
\end{lemma}

\begin{proof}
We know that $g$ on $\Sigma$ is a perturbation of the Kerr metric.\footnote{As already noted above, for a general asymptotically flat initial data set, this can be achieved by choosing $m$, and $a$, suitably.} We will use the existence of a time-like vectorfield $T$ such that \eqref{eq:prop:killing:all:orders} holds to prove that $g$ is in fact a perturbation of the Kerr metric in the entire domain of dependence of $\Sigma\cap \{r\geq R\}$, for $R$ sufficiently large. The transformation to ``comoving'' coordinates of Appendix~A in \cite{alexakis-schlue-shao} then shows that $(\mathcal{D},g)$ is positive mass space-time, as discussed in Section~\ref{sec:carleman}, and in particular satisfies the assumptions of  Theorem~\ref{thm:carleman}.

Firstly, we show that on $\Sigma$ the vectorfield $T$ obtained in Section~\ref{sec:vanishing} agrees to leading orders with the vectorfield $\partial_t$ in Boyer Lindquist coordinates; in fact
\begin{equation}\label{eq:data:lemma:T:t}
  T-\frac{\partial}{\partial t}=\mathcal{O}(r^{-3})\frac{\partial}{\partial r}+\mathcal{O}(r^{-3})\frac{\partial}{\partial t}+\sum_{A=1}^2\mathcal{O}(r^{-2})\frac{\partial}{\partial \vartheta^A}\,.
\end{equation}
To prove this we write
\begin{equation}
  T=(1+\alpha)\partial_t+\beta\partial_r+\gamma^A\partial_{\vartheta^A}\,.
\end{equation}
Recall that in Section~\ref{sec:gauge} we have first defined $T$ on a single cut $S_0^\ast$ of future null infinity $\mathcal{I}^+$, which is specified by a freely chosen surface $S_0\subset\Sigma$; c.f.~Fig.~\ref{fig:foliation}. If $(\mathcal{M},g)$ is time-periodic we may choose $S_0$ such that the limiting sphere $S_0^\ast$ on $C_0$ can be identified with the sphere at infinity on $\Sigma$. Then by construction, $\alpha,\beta,\gamma\to 0$ as $r\to \infty$. In the general case this is achieved by taking $S_0=\partial B_d$, of an exhaustion $B_d\subset\Sigma$ of the initial hypersurface by balls $B_d$, and thus normalises $T$ on $S_0^\ast$ to agree with $\partial_t$ in the limit as $d\to\infty$.

Now consider the equation
\begin{multline}\label{eq:LTg}
  (\mathcal{L}_T g)_{\mu\nu}=(\partial_\mu\alpha) g_{t\nu}+(\partial_\nu \alpha)g_{\mu t}+(1+\alpha)(\mathcal{L}_{\partial_t}g)_{\mu\nu}\\
  +(\partial_\mu\beta)g_{r\nu}+(\partial_\nu\beta)g_{\mu r}+\beta(\mathcal{L}_{\partial_r} g)_{\mu\nu}
  +(\partial_\mu\gamma^A)g_{A\nu}+(\partial_\nu\gamma^A) g_{\mu A}+\gamma^A(\mathcal{L}_{\partial_{\vartheta^A}}g)_{\mu\nu}\,.
\end{multline}
On one hand, we have by \eqref{eq:prop:killing:all:orders}, that
 \begin{equation}\label{eq:data:LTg}
   \mathcal{L}_T g=\mathcal{O}(r^{-k})\qquad\forall k\in\mathbb{N}\,.
 \end{equation}
On the other hand, by assumption \eqref{eq:initial:data:specific}, we calculate
\begin{subequations}\label{eq:data:Ltg}
\begin{gather}
  (\mathcal{L}_{\partial_t}g)_{rr}=\mathcal{O}_1^\infty(r^{-4})\,,\label{eq:data:Ltg:rr}
  (\mathcal{L}_{\partial_t}g)_{tr}=\mathcal{O}_1^\infty(r^{-4})\,,
  (\mathcal{L}_{\partial_t}g)_{r\vartheta^A}=\mathcal{O}_1^\infty(r^{-3})\,,\\
  (\mathcal{L}_{\partial_r}g)_{rr}=-\frac{2m}{r^2}+\mathcal{O}_2^\infty(r^{-3})\,,
  (\mathcal{L}_{\partial_r}g)_{tr}=\mathcal{O}_2^\infty(r^{-5})\,,
  (\mathcal{L}_{\partial_r}g)_{r\vartheta^A}=\mathcal{O}_2^\infty(r^{-4})\,,\\
  (\mathcal{L}_{\partial_{\vartheta^A}}g)_{rr}=\mathcal{O}_1^\infty(r^{-2})\,,
  (\mathcal{L}_{\partial_{\vartheta^A}}g)_{tr}=\mathcal{O}_1^\infty(r^{-4})\,,
  (\mathcal{L}_{\partial_{\vartheta^A}}g)_{r\vartheta^B}=\mathcal{O}_1^\infty(r^{-3})  \,.
\end{gather}
\end{subequations}
Inserting \eqref{eq:data:LTg} and \eqref{eq:data:Ltg} into \eqref{eq:LTg} we obtain equations for $\alpha$, $\beta$, and $\gamma$.
Indeed the $(\mu,\nu)=(r,\vartheta^A)$ component of the equation \eqref{eq:LTg} reads, to leading order,
\begin{equation}
  (\partial_r\gamma^A) +\frac{1}{r^2}(\partial_{\vartheta^A}\beta)+\frac{1}{r^2}(\partial_r\alpha)\mathcal{O}_2^\infty(r^{-1})=\mathcal{O}(r^{-k})\qquad\forall k\in\mathbb{N}\,.
\end{equation}
In view of $\beta\to 0$, $\alpha\to 0$,  and smoothness of these functions, we have that $\partial_{\vartheta^A}\beta,\partial_r\alpha=\mathcal{O}(r^{-1})$ and we obtain that $\gamma^A=\mathcal{O}(r^{-2})$.
From the $(\mu,\nu)=(r,r)$ component of \eqref{eq:LTg} we get the equation
\begin{equation}
  2(\partial_r\beta)\Bigl(1+\frac{2m}{r}\Bigr)-\Bigl(\frac{2m}{r^2}+\mathcal{O}_2(r^{-2})\Bigr)\beta+\gamma^A\mathcal{O}_1^\infty(r^{-2})=\mathcal{O}(r^{-4})\,;
\end{equation}
note that we here used \eqref{eq:data:Ltg:rr}, which directly reflects the assumption made on $g_{rr}$ component up to order $r^{-4}$ in \eqref{eq:initial:data:specific}.
This implies that $\beta=\mathcal{O}(r^{-3})$.
Finally the  $(\mu,\nu)=(r,t)$ component reads
\begin{equation}
  -2(\partial_r\alpha)\Bigl(1-\frac{2m}{r}\Bigr)+(\partial_r\gamma^A)\mathcal{O}_2^\infty(r^{-1})+\gamma^A\mathcal{O}_1^\infty(r^{-3})=\mathcal{O}(r^{-4})
\end{equation}
where we used \eqref{eq:data:Ltg:rr}, which is due to the assumption on $g_{tr}$ in \eqref{eq:initial:data:specific}.
This also implies that $\alpha=\mathcal{O}(r^{-3})$.

Secondly, we pass to a new system of coordinates $(\tilde{t},r,\vartheta^1,\vartheta^2)$, where $\tilde{t}=0$ on $\Sigma$, $T\tilde{t}=1$, and the coordinates $(r,\vartheta^1,\vartheta^2)$, are defined to be constant along the integral curves of $T$. Now \eqref{eq:data:lemma:T:t} shows that $g$ on $\Sigma$ is again of the form
\begin{multline}\label{eq:initial:data:less:specific}
  g= -\Bigl(1-\frac{2m}{r}\Bigr)\ud \tilde{t}^2+\mathcal{O}_2^\infty(r^{-3})\ud \tilde{t}^2+\mathcal{O}_2^\infty(r^{-3})\ud \tilde{t}\ud r+\sum_{A=1}^2\mathcal{O}_2^\infty(r^{-1})\ud \tilde{t}\ud \vartheta^A\\
  +\Bigl(1+\frac{2m}{r}-\frac{a^2}{r^2}\sin^2\vartheta^1\Bigr)\ud r^2+\mathcal{O}_2^\infty(r^{-3})\ud r^2+\sum_{A=1}^2\mathcal{O}_2^\infty(r^{-3})\ud r\ud \vartheta^A\\
+r^2\bigl(\ud \vartheta^1\bigr)^2+r^2\sin^2\vartheta^1\bigl(\ud\vartheta^2\bigr)^2+\sum_{A,B=1}^2\mathcal{O}_2^\infty(r^{0})\ud \vartheta^A\ud\vartheta^B\,.
\end{multline}

Thirdly, in these coordinates,
\begin{subequations}\label{eq:data:lemma:integrate}
\begin{gather}
  g_{\mu\nu}(\tilde{t},r,\vartheta^1,\vartheta^2)=  g_{\mu\nu}(0,r,\vartheta^1,\vartheta^2)+\int_0^{\tilde{t}}\partial_{\tilde{t}}  g_{\mu\nu}(t,r,\vartheta^1,\vartheta^2)\ud t\\
  \lvert g_{\mu\nu}(\tilde{t},r,\vartheta^1,\vartheta^2)-  g_{\mu\nu}(0,r,\vartheta^1,\vartheta^2)\rvert \leq\int_0^{\tilde{t}}\lvert(\mathcal{L}_Tg)_{\mu\nu}(t,r,\vartheta^1,\vartheta^2)\rvert\ud t\,.
\end{gather}
\end{subequations}
Therefore, on a domain of the form
\begin{equation}
  \mathcal{D}=\bigl\{(\tilde{t},r,\vartheta^1,\vartheta^2): r\geq R, \lvert \tilde{t}\rvert\leq r+2m\log r -(R+2m\log R)\bigr\}
\end{equation}
we have from \eqref{eq:data:lemma:integrate}, by virtue of \eqref{eq:prop:killing:all:orders} say for $k=4$, $\vert\mathcal{L}_Tg\rvert=\mathcal{O}(r^{-4})$ that on $\mathcal{D}$ thus defined,
\begin{equation}
  \lvert g_{\mu\nu}(\tilde{t},r,\vartheta^1,\vartheta^2)-g_{\mu\nu}(0,r,\vartheta^1,\vartheta^2)\rvert=\mathcal{O}(r^{-3})\,.
\end{equation}
Thus $g$ is of the form \eqref{eq:initial:data:less:specific} in the entire domain $\mathcal{D}$, and the transformation to ``comoving coordinates'' as discussed in Appendix~A of \cite{alexakis-schlue-shao} can be applied to bring the metric into the desired form \eqref{eq:metric:positivemass}, such that all assumptions of Section~\ref{sec:carleman} are satisfied.
\end{proof}

%% file: alexakis-schlue-time-periodic.bbl
\providecommand{\bysame}{\leavevmode\hbox to3em{\hrulefill}\thinspace}
\providecommand{\MR}{\relax\ifhmode\unskip\space\fi MR }
% \MRhref is called by the amsart/book/proc definition of \MR.
\providecommand{\MRhref}[2]{%
  \href{http://www.ams.org/mathscinet-getitem?mr=#1}{#2}
}
\providecommand{\href}[2]{#2}